\documentclass[12pt]{amsart}

\usepackage{geometry}
\usepackage{amsaddr}
\usepackage{amsmath} 
\usepackage{amsthm} 
{
	\newtheorem{theorem}{Theorem}
	\newtheorem{lemma}{Lemma}
    \newtheorem{corollary}{Corollary}
	 \newtheorem{assumption}{Assumption}
	 \newtheorem{remark}{Remark}
  }
\usepackage{amssymb, mathtools}
\usepackage{amsfonts}
\usepackage{multirow}
\usepackage{booktabs}
\usepackage[authoryear]{natbib}
\usepackage{graphicx}
\usepackage{enumitem}

\usepackage{bm}
\usepackage{bbm}
\usepackage{algorithm}
\usepackage{algorithmic}
\usepackage{hyperref}
\usepackage{cleveref}
\usepackage{commath}
\usepackage{color}
\usepackage{soul}
\definecolor{purple}{rgb}{0.5,0,0.5}
\definecolor{orange}{rgb}{1,0.45,0}

\crefname{theorem}{Theorem}{Theorems}
\Crefname{theorem}{Theorem}{Theorems}
\crefname{algorithm}{Algorithm}{Algorithms}
\Crefname{algorithm}{Algorithm}{Algorithms}
\usepackage{comment}
\def\E{\mathbbmss{E}}

\def\R{\mathbbmss{R}}

\def\Var{\mathbbmss{V}\text{ar}}

\def\Cov{\mathbbmss{C}\text{ov}}
\def\Normal{\mathcal{N}}

\def\Ind{\bm{\mathbbm{1}}}

\def\transpose{\prime}

\def\CiD{\Longrightarrow_d}
\def\CiP{\stackrel{p}{\longrightarrow}}

\def\bZ{\mathbf{Z}}

\def\bW{\mathbf{W}}
\def\bX{\mathbf{X}}

\DeclareMathOperator*{\argmin}{argmin}

\def\bT{\bm{T}}
\def\bSigma{\bm{\Sigma}}

\usepackage{bbm}

\title[V-fold jackknife]{The V-fold jackknife for semiparametric inference: variance estimation, confidence intervals, and simultaneous confidence bands}
\author{Yi Li\textsuperscript{1}, Ashkan Ertefaie\textsuperscript{2}, Mark van der Laan\textsuperscript{1}}
\thanks{\textsuperscript{1}Division of Biostatistics, University of California, Berkeley}
\thanks{\textsuperscript{2}University of Pennsylvania}

\begin{document}
\raggedbottom

\begin{abstract}

For decades, the bootstrap has been a default tool for statistical inference because of its broad applicability and minimal analytic requirements. Although its validity is well understood for smooth parametric estimators, its theoretical properties for many modern semiparametric and machine-learning estimators remain largely unstudied. Nevertheless, bootstrap procedures are often used routinely in such settings, even when their validity is unknown and their computational cost is substantial. We develop the $V$-fold jackknife as a computationally efficient and theoretically justified alternative for semiparametric inference. The procedure requires only $V$ leave-fold-out refits and uses the empirical dispersion of jackknife pseudo-values to quantify uncertainty, without requiring analytic derivation or numerical evaluation of an influence function. For regular asymptotically linear estimators of pathwise differentiable parameters, we show that, for fixed $V$, the Studentized $V$-fold jackknife statistic converges to a $t$-distribution with $V-1$ degrees of freedom. Thus, valid confidence intervals are available even though the jackknife variance estimator does not converge in probability.  When $V\to\infty$, we establish consistency of the jackknife variance estimator at rate $V^{-1/2}$, allowing $V$ to diverge slowly, for example at rate $\log n$. We also develop simultaneous confidence bands based on the correct componentwise-Studentized limiting distribution. Finally, we extend the theory to generalized asymptotically linear estimators with diverging influence-function variance and slower-than-$\sqrt n$ convergence; scale invariance of Studentization eliminates the need to know the effective convergence rate. Simulations on the average treatment effect, Kaplan--Meier survival curve, and highly adaptive lasso dose-response curves confirm reliable inference, including where influence-function-based standard errors are anti-conservative or unstable.

\end{abstract}

\maketitle

\section{Introduction}

For decades, resampling methods have been central to modern statistical inference. 
The bootstrap, in particular, has become a default tool for variance estimation and 
confidence interval construction due to its broad applicability and minimal analytic 
requirements \citep{efron1979bootstrap,efron1993bootstrap,vanderVaart1998asymptotic}. 
However, in contemporary applications involving high-dimensional nuisance estimation, 
adaptive model selection, and machine learning algorithms, the validity of the ordinary 
nonparametric bootstrap is often unclear. Bootstrap resampling does not merely perturb 
the empirical distribution; it may also change the data-adaptive estimation path, including 
tuning-parameter selection, selected basis functions, internal cross-validation splits, 
and fitted nuisance estimators. As a result, the estimator computed on a bootstrap sample 
need not satisfy the same first-order expansion as the original estimator, and the bootstrap 
distribution need not consistently approximate the sampling distribution of the estimator.

This concern is especially pronounced for semiparametric and machine learning--based 
estimators, where valid inference typically relies on asymptotic linearity together with 
stable remainder terms. In practice, bootstrap standard errors are often reported for such 
estimators even when the regularity conditions needed for bootstrap validity have not been 
verified. This statistical issue is distinct from, and in many settings more fundamental 
than, computational cost. Repeatedly refitting a complex estimator across hundreds or 
thousands of bootstrap samples can be expensive, but even when such computation is 
feasible, the resulting inference may not be theoretically justified.

These considerations motivate an inferential framework that is both computationally lighter
than the bootstrap and more directly tied to the estimator's own first-order linearization.
We propose to use the grouped, or $V$-fold, jackknife for this purpose. The procedure
evaluates the estimator on genuine leave-fold-out subsamples, rather than on bootstrap
samples generated with replacement, and bases inference on the dispersion of fold-level
pseudo-values. For asymptotically linear estimators, these pseudo-values recover fold-level
influence-function averages up to negligible remainder terms, yielding Studentized inference
under transparent fold-level stability conditions.

In this paper, we develop a systematic theory for this $V$-fold jackknife approach as a
practical inference tool for regular asymptotically linear estimators of pathwise
differentiable parameters --- a class that includes targeted minimum loss-based estimators,
one-step estimators, and augmented IPW estimators --- and, more broadly, for estimators
with non-standard convergence rates whose influence-function variance diverges with sample
size. The proposed approach is attractive for three related reasons. First, it is
computationally efficient: it requires only $V$ refits of the estimator, compared with the
hundreds or thousands of refits typically used by the bootstrap, and valid asymptotic
inference already holds for fixed $V \ge 2$. Second, it avoids analytic influence-function
calculations in implementation. The practitioner needs only to evaluate the estimator on
the full sample and on leave-fold-out subsamples; the fold-level pseudo-values then provide
a data-driven approximation to the sampling variability. Third, Studentizing the $V$ fold-level pseudo-values with $t_{V-1}$ critical values provides the correct fixed-$V$ asymptotic degrees-of-freedom correction and acts as an exact correction for Gaussian fold means. When $n/V$ is moderately large, the fold-level averages are approximately normal, so the resulting Studentized statistic is well approximated by a $t_{V-1}$ distribution. For linear estimators, this is especially transparent because the statistic is built from sample means over $n/V$ i.i.d.\ observations, and the approximation improves as $V$ decreases and each block contains more observations. For asymptotically linear estimators, the same logic applies to the leading linear component, while the second-order remainder may still have a finite-sample impact even though it is asymptotically negligible. The heavier tails of $t_{V-1}$ for small $V$ automatically widen the intervals, yielding more reliable coverage than normal-based methods.
A further advantage is scale invariance. In non-$\sqrt n$ settings where the variance
$\sigma_n^2$ of the influence-function sequence diverges (e.g., dose-response curve estimation), the unknown convergence rate
cancels between the numerator and denominator of the Studentized statistic. Thus the same
procedure can be used without knowing the effective dimension or deriving the growing
influence curve explicitly.

Our theoretical and methodological contributions are sixfold. 
First, we extend the classical fixed-$V$ $t_{V-1}$ result of 
\citet{Brillinger1964} from maximum likelihood estimators to the full class of 
regular asymptotically linear estimators of pathwise differentiable parameters 
(\cref{thm:fixedV}), and show that the resulting confidence interval uses the 
correct fixed-$V$ asymptotic degrees-of-freedom correction 
(\cref{cor:CI}; see also \cref{rem:finite_sample}). 
Second, we establish that, when $V \to \infty$, the jackknife variance estimator 
consistently estimates the asymptotic variance at rate $1/\sqrt{V}$, subject to 
an explicit upper-growth condition on $V$ that ensures fold-level remainder 
stability; in particular, slow choices such as $V=\log n$ are covered under the 
sufficient conditions in Appendix~\ref{app:V_rate} 
(\cref{thm:divergingV}). 
Third, we develop, to our knowledge, the first simultaneous confidence band 
construction based on the grouped jackknife. The band uses a 
componentwise-Studentized limit distribution that accounts for the Wishart 
structure of the fold-level denominators, and requires neither analytic 
influence-function computation 
\citep{vanderlaan2019simultaneous,DiazvanderLaan2013} nor many bootstrap refits 
\citep{MontielOleaPlagborgMoller2019} 
(\cref{thm:multi_t}). 
Fourth, we extend the framework to generalized asymptotically linear estimators 
whose influence-curve variance diverges with $n$, so that convergence occurs at 
a rate slower than $n^{-1/2}$; in this setting, scale invariance ensures that 
the unknown rate cancels in the Studentized statistic 
(\cref{thm:fixedV_slower,thm:divergingV_slower}). 
Fifth, we introduce bias-corrected simultaneous confidence bands to address 
finite-sample bias (\cref{thm:bias_corrected}). 
Finally, we demonstrate the method for TMLE estimation of the average treatment 
effect, the Kaplan--Meier survival curve, and HAL-based dose-response curve 
estimation, including settings where influence-function-based standard errors 
are anti-conservative or numerically unstable.

Bootstrap techniques are frequently used with data-adaptive algorithms even when
the theoretical conditions needed for bootstrap validity are difficult to verify or
are not established. The proposed $V$-fold jackknife offers an alternative route:
its validity is tied directly to asymptotic linearity and fold-level remainder
stability of the original estimator, rather than to a bootstrap approximation of
the full adaptive estimation procedure. At the same time, it is substantially
cheaper computationally, requiring only $V$ refits, and the fixed-$V$
Studentization yields the appropriate $t_{V-1}$ degrees-of-freedom correction. 


\begin{table}[htbp]
\color{black}
\centering
\caption{Comparison of general-purpose inference methods for semiparametric estimators.}
\label{tab:method_comparison}
\small
\begin{tabular}{l ccc}
\toprule
Feature & EIC / Delta method & Bootstrap & $V$-fold Jackknife \\
\midrule
Requires IC computation & Yes & No & No \\
Number of refits & 1 & $B$ (100s--1000s) & $V$ (e.g., 5--20) \\
Valid for non-$\sqrt{n}$ rates & Requires knowing $\sigma_n$ & In principle & Yes (scale-invariant) \\
Finite-sample correction & No (normal quantiles) & No (normal quantiles) & Yes ($t_{V-1}$ quantiles) \\
Simultaneous bands & Via IC stacking & Via bootstrap & Via Studentization \\
Can fail numerically & Yes (up to 6.2\% in HAL) & No & No \\
\bottomrule
\end{tabular}
\end{table}

The paper is organized as follows. Sections~3--5 develop the core fixed-$V$ and diverging-$V$ theory for $\sqrt{n}$-rate estimators. Section~6 constructs the simultaneous confidence band. Section~\ref{sec:slower_rates} extends the framework to non-standard convergence rates, with a detailed application to HAL dose-response curves. Section~8 presents simulations and Section~9 discusses implications and extensions.

\section{Related Work}

The resampling scheme studied in this paper is a grouped, or delete-a-group, jackknife.
Beginning with the foundational jackknife work of \citet{quenouille1949approximate} and
\citet{tukey1958bias}, the grouped variant partitions the sample into a fixed number of
disjoint groups and recomputes the estimator after deleting one group at a time. The
resulting leave-group-out estimates are then combined into pseudo-values, whose empirical
dispersion provides a natural variance estimator. Early theoretical analyses focused mainly
on classical parametric settings, including maximum likelihood estimators
\citep{Brillinger1964}, and \citet{Miller1974} provides a broad review of jackknife
methods. Related grouped and pseudo-value constructions also appear in the
$U$-statistics literature \citep{Arvesen1969}.

Our work adopts the same grouped pseudo-value construction, but gives it a semiparametric
foundation. We develop fixed-$V$ jackknife inference (\cref{thm:fixedV}) for regular asymptotically linear
estimators of pathwise differentiable parameters, a class that includes targeted minimum
loss-based estimators, one-step estimators, augmented inverse probability weighted
estimators, and other procedures for which explicit influence-function-based variance
formulas may be cumbersome to derive or implement. We show that the grouped pseudo-values
are asymptotically equivalent to fold-level averages of the influence function, so the
corresponding Studentized statistic converges to a $t_{V-1}$ limit
(\cref{thm:fixedV}). This yields valid fixed-$V$ confidence intervals even though the grouped
jackknife variance estimator does not concentrate around the true asymptotic variance.
Formal notation for the fold partition, leave-fold-out estimators, and jackknife
pseudo-values is introduced in Section~\ref{sec:vfold}.


A large literature studies \emph{delete-$d$} jackknife schemes, in which a subset
of size $d$ is removed and the statistic is recomputed over many subsets.
\citet{Wu1986} studies jackknife and bootstrap methods in regression settings, and
\citet{Shao1989} develops general conditions under which jackknife variance estimators
are consistent, emphasizing the role of smoothness or differentiability of the underlying
functional. Our diverging-$V$ analysis (\cref{thm:divergingV}) builds on this line of work but targets a different
regime: deleting one fold from a balanced $V$-fold partition gives a structured delete-$d$
scheme with $d=n/V$. Thus, when $V\to\infty$ with $V=o(n)$, the deleted fraction satisfies
$d/n=1/V\to0$ while the deleted block size satisfies $d\to\infty$. We show that, under
fold-level remainder conditions tailored to asymptotic linearity, this single-partition
construction yields a consistent jackknife variance estimator. This differs from much of
the classical delete-$d$ literature in two ways. First, we use a single partition into disjoint
folds, rather than enumerating or averaging over many deleted subsets. Second, our
assumptions are stated directly in terms of the RAL expansion and fold-wise remainder
stability, which are natural for semiparametric estimators.

The classical comparison between the jackknife and the bootstrap is closely tied to
smoothness. \citet{efron1979bootstrap} showed that the ordinary delete-one jackknife
can fail for non-smooth functionals such as the sample median, because the leave-one-out
perturbations are too local to capture the relevant first-order variability. Efron also noted
that deleting observations in larger groups can alleviate this problem; for the median,
his calculation suggests deleting groups of order $n^{1/2}$. For a balanced $V$-fold
jackknife, this corresponds heuristically to groups of size $n/V$.

For the purposes of the present paper, the lesson is therefore not that the jackknife inherently requires stronger regularity than the bootstrap. Rather, both methods are governed by the availability of a reliable first-order linearization, but their respective fragilities manifest differently in modern contexts. In contemporary terms, Efron's smooth finite-support argument aligns closely with the later compact-differentiability framework for bootstrap validity \citep{gill1989non}. However, in modern machine-learning and adaptive-estimation settings, the ordinary (multinomial) nonparametric bootstrap introduces a distinct algorithmic fragility. Resampling with replacement inevitably creates duplicate observations, meaning each bootstrap sample contains only about 63.2\% unique data points. This disruption of the discrete support can severely destabilize the internal cross-validation of data-adaptive estimators. This particular mechanism is an artifact of multinomial resampling: perturbation-type schemes---the Bayesian, multiplier, and other exchangeably weighted bootstraps---instead attach random weights to the \emph{entire} sample and so preserve its support. Such reweighting does not, however, resolve the deeper question of bootstrap validity for adaptive estimators: like the nonparametric bootstrap, its justification rests on compact (Hadamard) differentiability of the functional \citep{gill1989non,PraestgaardWellner1993}, a regularity that machine-learning--based, adaptively-tuned estimators can violate. This phenomenon helps explain the historical turn toward analytic influence-curve-based inference.

Efron's observation provides useful intuition for grouped jackknife methods, but it is not
the basis of our theory. We do not seek to repair arbitrary non-smooth functionals by
choosing large deletion groups. Instead, our analysis relies on asymptotic linearity, or its
generalized slower-rate analogue, together with fold-level remainder stability. Under these
conditions, the grouped pseudo-values recover fold-level averages of the relevant
influence-function sequence. This representation supports both fixed-$V$ Studentized
inference, where a $t_{V-1}$ limit arises, and diverging-$V$ variance consistency. It also
explains why the procedure remains valid in non-$\sqrt n$ settings: the unknown scale of
the influence-function sequence cancels in the Studentized statistic.

A complementary line of work seeks to reduce the computational burden of bootstrap-based
inference. \citet{Lam2022} proposes the Cheap Bootstrap, and \citet{Ohlendorff2025}
develop a related cheap subsampling method based on subsampling without replacement and
Studentization over $B$ Monte Carlo replicates. These methods share our computational
motivation, but they introduce additional tuning choices, most notably the subsample size
or subsample fraction $m/n$. In particular, \citet{Ohlendorff2025} note that keeping
$m/n$ bounded away from one can help control remainder terms
\citep[Remark~1]{Ohlendorff2025}, so finite-sample performance may be sensitive to this
choice. Moreover, the resulting $t_B$ correction reflects Monte Carlo replication error
rather than the finite number of data components, and disappears as $B\to\infty$
\citep[Theorem~2]{Ohlendorff2025}. These methods also do not address the simultaneous
confidence bands, slower-rate inference, or bias-corrected procedures developed here (\cref{thm:multi_t,thm:fixedV_slower,thm:divergingV_slower,thm:bias_corrected}).

Grouped jackknife methods are widely used in complex survey sampling under the name
\emph{delete-a-group jackknife} (DAGJK), where variance estimation is based on a finite set
of replicate estimates; see \citet{Kott1998,Kott2001,KottGarren2011} for practical guidance
and small-sample considerations. DAGJK standard errors have also appeared in semiparametric
applications as a practical alternative when analytic influence-function-based variance
formulas are difficult to characterize. For example, \citet{YangPieperCools2020} use a
delete-a-group jackknife with $V=500$ groups for continuous-time structural failure time
models, following \citet{Kott1998}. We provide a general theory for this type of construction
in modern semiparametric settings.

\section{Notation}

Let $O_1,\dots,O_n$ be i.i.d.\ observations drawn from an unknown distribution $P_0$ on a measurable space $\mathcal{O}$. Let $\Psi$ denote a parameter mapping from a statistical model $\mathcal{M}$ to a real value (i.e., $\Psi: \mathcal{M}\Longrightarrow \R$). Let $\Psi(P_0)$ denote a real-valued target parameter, and suppose that an estimator $ \hat\Psi(P_n)$ admits an asymptotically linear expansion of the form
\[
\hat \Psi(P_n) - \Psi(P_0)
=
P_n \varphi + R(P_n, P_0),
\]
where $P_n$ is the empirical measure, $\varphi=\varphi_{P_0}$ is a mean-zero influence function satisfying $\E_{P_0}[\varphi^2]<\infty$, and $R(P_n, P_0)=o_p(n^{-1/2})$ is a remainder term. In Section~\ref{sec:slower_rates}, we generalize this framework to allow the influence function variance $\sigma_n^2 = \E_{P_0}[\varphi_n^2]$ to diverge with $n$, covering estimators that converge at rates slower than $\sqrt{n}$, such as HAL functional estimators.

To estimate the variance $\sigma^2=\E_{P_0}[\varphi^2]$, we consider a $V$-fold jackknife construction. Let $\{I_1,\dots,I_V\}$ denote a partition of $\{1,\dots,n\}$ into disjoint folds of \textcolor{black}{exactly equal size $|I_v|=n/V$}. \textcolor{black}{Throughout the theoretical development, we assume that $n$ is divisible by $V$ so that the folds are exactly balanced. This preserves the exact algebraic structure of the jackknife pseudo-values. In practice, when $n$ is not a multiple of $V$, approximately balanced folds introduce only negligible bookkeeping terms and do not change the conceptual argument.} Let
\[
P_{n,v}^1 f = \frac{1}{|I_v|}\sum_{i\in I_v} f(O_i)
\]
denote the empirical average over fold $v$, and let $P_{n,-v}$ denote the empirical measure based on the complement of fold $v$.

\section{V-fold Jackknife } \label{sec:vfold}

For an asymptotically linear estimator $\hat\Psi(P_n)$, we have
\[
\hat\Psi(P_n) - \Psi(P_0) = \frac{1}{V} \sum_{v=1}^V P^1_{n,v} \, \varphi + R(P_n, P_0),
\]
and
\[
\hat\Psi(P_{n,-v}) - \Psi(P_0) = \frac{1}{V-1} \sum_{\substack{v'=1 \\ v' \neq v}}^V P^1_{n,v'} \, \varphi + R(P_{n,-v}, P_0).
\]
Let $IC_{Jack}(v) = V\hat\Psi(P_n)-(V-1)\hat\Psi(P_{n,-v})$ denote the $v$th \emph{jackknife pseudo-value} \citep{tukey1958bias}, which can also be viewed as a fold-level empirical influence curve. Then we have,
\[
IC_{Jack}(v)-\Psi(P_0)  = P^1_{n,v} \, \varphi + R(P_{n,-v}, P_0) + V \{R(P_n, P_0)-R(P_{n,-v}, P_0)\}.
\]
Let $\hat\Psi_{Jack}(P_n) = \frac{1}{V} \sum_{v=1}^V IC_{Jack}(v)$ denote the \emph{jackknife estimator} of $\Psi(P_0)$, i.e., the mean of the pseudo-values.  We define the $V$-fold jackknife variance estimator of $\sigma^2=\Var(\varphi(O))$ as
\begin{align} \label{eq:jackest}
    \hat S_{Jack}^2=(V-1)^{-1}\sum_{v=1}^V (IC_{Jack}(v)-\hat\Psi_{Jack}(P_n))^2.
\end{align}

In this formulation, the true variance is $\Var(\hat\Psi(P_n)) \approx \frac{\sigma^2}{n}$. By the independence of folds, $\Var(P^1_{n,v} \, \varphi) = \left(\frac{1}{n/V}\right)^2 \sum_{i \in I_v} \sigma^2 = \frac{V}{n} \sigma^2$. Hence, the sample variance of the fold-level averages $S_W^2 = (V-1)^{-1}\sum_v(P^1_{n,v} \, \varphi - \bar W)^2$ where $\bar W = V^{-1} \sum_{v=1}^V P^1_{n,v} \, \varphi$ naturally estimates $\frac{V}{n} \sigma^2$. Consequently, the scaled V-fold Jackknife variance estimator $\frac{n}{V}\hat S_{Jack}^2$ estimates $\sigma^2$.

{\color{black}
Define the jackknife standard error $\widehat{SE} = \hat S_{Jack} / \sqrt{V}$. The $V$-fold jackknife $(1-\alpha)$-level confidence interval for $\Psi(P_0)$ is
\[
    \hat\Psi(P_n) \pm t_{V-1,\,1-\alpha/2} \cdot \widehat{SE},
\]
where $t_{V-1,\,1-\alpha/2}$ denotes the $(1-\alpha/2)$-quantile of the Student $t$-distribution with $V-1$ degrees of freedom. The use of $t_{V-1}$ critical values rather than normal quantiles \textcolor{black}{provides the fixed-$V$ asymptotic degrees-of-freedom correction}. Its theoretical justification is that, under the conditions stated below, the Studentized statistic converges in distribution to $t_{V-1}$ when $V$ is fixed. We emphasize that this interval is centered at the full-sample estimator $\hat\Psi(P_n)$ rather than at the jackknife mean of the pseudo-values $\hat\Psi_{Jack}(P_n) = V^{-1}\sum_{v=1}^V IC_{Jack}(v)$. The two centerings are equal to first order: their difference equals the jackknife bias estimate $\hat b$, which satisfies $\hat b = o_p(\widehat{SE})$, so replacing $\hat\Psi_{Jack}(P_n)$ by $\hat\Psi(P_n)$ does not affect the limiting $t_{V-1}$ distribution or the coverage of the interval. This equivalence is made precise in \cref{lem:centering} and \cref{cor:CI} below. The complete procedure is summarized in \cref{alg:univariate}.

\begin{algorithm}[tb]
\color{black}
\caption{\color{black}Univariate $V$-fold jackknife confidence interval}
\label{alg:univariate}
\begin{algorithmic}[1]
\STATE \textbf{Input:} Estimator $\hat\Psi$, observations $O_1,\dots,O_n$, number of folds $V \ge 2$, confidence level $1-\alpha$.
\STATE \textbf{Output:} $(1-\alpha)$-level confidence interval for $\Psi(P_0)$.
\medskip
\STATE Compute $\hat\Psi(P_n)$ using all $n$ observations.
\STATE Partition $\{1,\dots,n\}$ into $V$ disjoint folds $I_1,\dots,I_V$ of approximately equal size.
\FOR{$v = 1$ \TO $V$}
    \STATE Compute the leave-fold-out estimate $\hat\Psi(P_{n,-v})$ using $\{1,\dots,n\} \setminus I_v$.
    \STATE Compute the pseudo-value $IC_{Jack}(v) = V\hat\Psi(P_n) - (V-1)\hat\Psi(P_{n,-v})$.
\ENDFOR
\STATE Compute $\widehat{SE} = \hat S_{Jack}/\sqrt{V}$, where $\hat S_{Jack}^2 = (V-1)^{-1}\sum_{v=1}^V \big(IC_{Jack}(v) - \bar{IC}_{Jack}\big)^2$.
\STATE \textbf{Return} the confidence interval $\hat\Psi(P_n) \pm t_{V-1,\,1-\alpha/2} \cdot \widehat{SE}$.
\end{algorithmic}
\end{algorithm}
}

\section{Properties of V-fold Jackknife}\label{sec:properties}

We study a single Studentized jackknife procedure for constructing confidence intervals for $\Psi(P_0)$, but analyze it in two asymptotic regimes: one with a fixed number of folds $V$, and one with $V \to \infty$. When $V$ is fixed, the $t_{V-1}$ correction reflects the finite number of fold-level contributions; when $V$ grows, the same procedure remains valid and the $t_{V-1}$ critical values approach the corresponding normal quantiles. This makes the trade-off transparent: small fixed $V$ typically improves the finite-sample $t$ approximation but yields wider intervals, whereas slowly increasing $V$ produces narrower intervals while preserving asymptotic validity. To clarify this behavior, we also derive the exact algebraic decomposition of the jackknife variance estimator, which isolates the dominant linear component and makes the role of the remainder explicit. In particular, it shows why the $t$ correction is central: for linear estimators it acts directly on the fold-level means, and for asymptotically linear estimators the same logic applies to the leading linear term.

\subsection{Fixed $V$ results}

We now consider the regime in which $V$ is fixed.
With $V$ fixed, the variance estimator does not concentrate to the asymptotic variance in probability as $n \to \infty$, because the effective number of independent components remains $V$. Nevertheless, the fold-level fluctuations still support valid asymptotic $t$-type inference. Specifically, we prove that
\[
\frac{\sqrt{V}\,(\hat\Psi_{Jack}(P_n)-\Psi(P_0))}{\hat S_{Jack}} \Longrightarrow_d\ t_{V-1}.
\]
The intuition is straightforward: if the remainder terms are negligible, then $IC_{Jack}(v)-\Psi(P_0)=P^1_{n,v}\varphi$ for each $v$, so each fold mean is asymptotically normal and Studentization across the $V$ independent fold means yields the $t_{V-1}$ limit.

\begin{lemma}[Equivalence of $\hat S_{Jack}$ and $S_W$ for fixed $V$]
\label{lem:Sequiv}
Fix $V\ge 2$ and let $|I_v|=n/V\to\infty$. Denote $W_v=P^1_{n,v}\varphi$,
$\bar W=V^{-1}\sum_{v=1}^V W_v$, and
\[
S_W^2=\frac{1}{V-1}\sum_{v=1}^V (W_v-\bar W)^2,
\qquad
\hat S_{Jack}^2=\frac{1}{V-1}\sum_{v=1}^V\big(IC_{Jack}(v)-\hat\Psi_{Jack}(P_n)\big)^2.
\]
Assume that $\max_{1\le v\le V}\big|R(P_{n,-v},P_0)\big| = o_p(n^{-1/2})$, 
\textcolor{black}{$\E_{P_0}\{\varphi(O)\}=0$}, and that \textcolor{black}{$0<\sigma^2=\Var(\varphi(O))<\infty$}. {\color{black} Then
\[
\frac{n}{V}S_W^2 \CiD \sigma^2\frac{\chi^2_{V-1}}{V-1},
\]
and the limiting distribution has no atom at zero. Moreover,} $\hat S_{Jack}^2 = S_W^2 + o_p(n^{-1})$, which implies $\hat S_{Jack}^2 = S_W^2\{1+o_p(1)\}$. 
\end{lemma}

\begin{proof}
See Appendix~\ref{app:proofs}.
\end{proof}

\begin{theorem}\label{thm:fixedV}
    {\color{black} Fix $V\ge2$ and assume the folds are exactly balanced with $|I_v|=n/V\to\infty$. Let $W_v=P^1_{n,v}\varphi$ and $S_W^2=(V-1)^{-1}\sum_{v=1}^V (W_v-\bar W)^2$. Suppose $\E_{P_0}[\varphi]=0$, $0<\sigma^2=\E_{P_0}[\varphi^2]<\infty$, $R(P_{n},P_0)=o_p(n^{-1/2})$, and $\max_{1\le v\le V} |R(P_{n,-v}, P_0)| = o_p(n^{-1/2})$. Then}  
    \[
\frac{\sqrt{V}\,(\hat\Psi_{Jack}(P_n)-\Psi(P_0))}{\hat S_{Jack}}
\ \CiD \  t_{V-1}.
\]
\end{theorem}
\begin{proof}
See Appendix~\ref{app:proofs}.
\end{proof}

The Studentized statistic in \cref{thm:fixedV} is centered at $\hat\Psi_{Jack}(P_n)$, whereas the confidence intervals we construct are centered at the full-sample estimator $\hat\Psi(P_n)$. The following lemma shows that these two centerings are asymptotically equivalent, connecting the centering mismatch to the jackknife bias estimate and to the remainder structure already established in \cref{lem:Sequiv}.

\begin{lemma}[Centering equivalence for fixed $V$]
\label{lem:centering}
Under the conditions of \cref{thm:fixedV}, define the jackknife bias estimate
\begin{equation}\label{eq:jack_bias_scalar}
    \hat b = (V-1)\big(\bar\Psi_{(-v)} - \hat\Psi(P_n)\big), \qquad \bar\Psi_{(-v)} = \frac{1}{V}\sum_{v=1}^V \hat\Psi(P_{n,-v}).
\end{equation}
Then $\hat\Psi_{Jack}(P_n) - \hat\Psi(P_n) = -\hat b$ and $\hat b = o_p(\widehat{SE})$, where $\widehat{SE} = \hat S_{Jack}/\sqrt{V}$.
\end{lemma}

\begin{proof}
See Appendix~\ref{app:proofs}.
\end{proof}

As an immediate consequence, we can center the confidence interval at the full-sample estimator $\hat\Psi(P_n)$.

\begin{corollary}[Confidence interval centered at $\hat\Psi(P_n)$]
\label{cor:CI}
Under the conditions of \cref{thm:fixedV},
\[
\frac{\sqrt{V}\,\big(\hat\Psi(P_n)-\Psi(P_0)\big)}{\hat S_{Jack}}
\ \CiD \  t_{V-1}.
\]
Hence a $(1-\alpha)$-level confidence interval for $\Psi(P_0)$ is
\begin{equation}\label{eq:CI_scalar}
    \hat\Psi(P_n) \pm t_{V-1,\,1-\alpha/2}\cdot \widehat{SE}, \qquad \widehat{SE} = \frac{\hat S_{Jack}}{\sqrt{V}}.
\end{equation}
\end{corollary}

\begin{proof}
See Appendix~\ref{app:proofs}.
\end{proof}

In the simulation studies section we show that the confidence intervals \eqref{eq:CI_scalar} have valid coverage even when $V=5$. One disadvantage of choosing smaller values of $V$ is that the resulting confidence intervals can be wide. As we increase the number of folds, the quantiles of the $t$-distribution get closer to those from a normal distribution and thus the confidence intervals shrink.

{\color{black}
\begin{remark}[Finite-sample behavior and the role of $V$]\label{rem:finite_sample}
{
The fixed-$V$ result of \cref{thm:fixedV} also has a useful finite-sample interpretation. Consider first \emph{linear estimators} --- those for which $\hat\Psi(P_n) - \Psi(P_0) = P_n\varphi$ exactly, with no remainder ($R(P_n, P_0) = 0$). Then the fold-level averages $W_v = P^1_{n,v}\varphi$ are exact sample means of $n/V$ i.i.d.\ random variables. Choosing $V$ small enough that $n/V$ is moderately large improves the normal approximation for each $W_v$ by the central limit theorem and therefore improves the $t_{V-1}$ approximation.

For \emph{asymptotically linear estimators} with a nonzero remainder $R(P_{n,-v}, P_0) = o_p(n^{-1/2})$, the $t_{V-1}$ Studentization still acts directly on the dominant linear component, namely the fold-level averages $W_v$. Only the negligibility of the remainder relies on asymptotics, and smaller $V$ helps here as well: larger fold sizes $n/V$ make $\hat\Psi(P_{n,-v})$ closer to its population target and thereby shrink the remainder terms. In practice, one may therefore choose a small $V$ to improve the $t_{V-1}$ approximation, accepting wider confidence intervals (because of the heavier tails of $t_{V-1}$) in exchange for more reliable coverage.
}
\end{remark}
}

\subsection{Diverging $V$}

We study the asymptotic behavior of the $V$-fold jackknife variance estimator 
in the regime where the number of folds $V$ diverges with $n$. 
Our objective is to characterize its consistency, convergence rate, 
and asymptotic linear representation. In this setting, the primary question is whether the fold-based 
construction consistently estimates the asymptotic variance 
$\sigma^2 = \E_{P_0}[\varphi^2]$ of the underlying estimator. 
We show that consistency holds under mild remainder conditions 
\textcolor{black}{whenever $V \to \infty$ and the fold-level remainder-difference condition below holds. Appendix~\ref{app:V_rate} gives sufficient TMLE-style structural conditions under which this condition follows from an explicit upper-growth bound on $V$, and slow divergence such as $V=\log n$ satisfies these sufficient bounds.} Thus, the usual Wald type confidence intervals for the asymptotic linear estimators can be constructed based on the estimated $\hat \sigma$

For a diverging number of folds, define the fold-level remainder differences and their normalized fold RMS norm by
\[
d_{n,v}=R(P_n,P_0)-R(P_{n,-v},P_0),
\qquad
\|d_n\|_{V,2}
=
\left(\frac{1}{V}\sum_{v=1}^V d_{n,v}^2\right)^{1/2}.
\]
This notation distinguishes the fold RMS norm used in the jackknife variance proof from an analytic $L_2(P_0)$ norm on nuisance functions.

\begin{theorem}[Consistency under $V \to \infty$]
\label{thm:divergingV}

Let $\{O_i\}_{i=1}^n$ be i.i.d.\ from $P_0$. 
Suppose the estimator $\hat\Psi(P_n)$ satisfies
\[
\hat\Psi(P_n) - \Psi(P_0)
= P_n \varphi + R(P_n,P_0),
\]
where $\E_{P_0}[\varphi]=0$, 
\textcolor{black}{$0<\E_{P_0}[\varphi^2]=\sigma^2 < \infty$}, 
and $R(P_n,P_0) = o_p(n^{-1/2})$. Partition the sample into $V$ disjoint folds of \textcolor{black}{exactly} equal size $n/V$, 
and define the $V$-fold jackknife variance estimator
\[
\hat{\sigma}^2
= \frac{n}{V} \hat S_{Jack}^2,
\quad
\hat S_{Jack}^2
= \frac{1}{V-1} \sum_{v=1}^V(IC_{Jack}(v)-\hat\Psi_{Jack}(P_n))^2.
\]

 \textcolor{black}{Assume that $V=V_n \to \infty$, $V=o(n)$, and the fold-level remainder differences satisfy}
\[
\textcolor{black}{
\|d_n\|_{V,2}
=
o_p\!\left(\frac{1}{\sqrt n\,V}\right).
}
\]
\textcolor{black}{Appendix~\ref{app:V_rate} gives sufficient TMLE, one-step, and AIPW-style structural conditions under which this stability condition follows from the explicit upper-growth bound $V=o(n/\tilde\sigma_n^4)$, where $\tilde\sigma_n^2$ is the variance scale of the nuisance estimator's linear approximation.}
 Then
\[
\hat\sigma^2-\sigma^2
=
\Big(\frac{1}{n}\sum_{i=1}^n \varphi_i^2-\sigma^2\Big)
+
T_n
+
o_p(V^{-1/2}),
\]
where
$
T_n
=
\frac{1}{n}\sum_{v=1}^V \sum_{\substack{i,j\in I_v\\ i\neq j}}
\varphi_i\varphi_j .
$
{\color{black} The cross-fold term satisfies $T_n=O_p(V^{-1/2})$.} If $V \to \infty$, then
$ \hat{\sigma}^2
\;\xrightarrow{p}\;
\sigma^2.$ 
{\color{black} If, additionally, $\E_{P_0}[\varphi^4]<\infty$, then}
\[
{\color{black}
\sqrt{V}
\left(
 \hat{\sigma}^2 - \sigma^2
\right)
= O_p(1).}
\]

\end{theorem}
\begin{proof}
See Appendix~\ref{app:proofs}.
\end{proof}

\begin{corollary}[Diverging-$V$ Wald inference]\label{cor:divergingV_wald}
\begingroup\color{black}
Under the assumptions of \cref{thm:divergingV}, suppose also that $\sqrt n\,P_n\varphi \CiD \Normal(0,\sigma^2)$. Then
\[
\frac{\hat\Psi(P_n)-\Psi(P_0)}{\hat\sigma/\sqrt n}
\CiD \Normal(0,1),
\qquad
\hat\sigma^2=\frac{n}{V}\hat S_{Jack}^2.
\]
Consequently,
\[
\hat\Psi(P_n)\pm t_{V-1,\,1-\alpha/2}\frac{\hat\sigma}{\sqrt n}
=
\hat\Psi(P_n)\pm t_{V-1,\,1-\alpha/2}\frac{\hat S_{Jack}}{\sqrt V}
\]
has asymptotic coverage $1-\alpha$, since $t_{V-1,\,1-\alpha/2}\to z_{1-\alpha/2}$ as $V\to\infty$.
\endgroup
\end{corollary}
\begin{proof}
See Appendix~\ref{app:proofs}.
\end{proof}

\section{Simultaneous V-fold Jackknife Confidence Bands}\label{sec:simult}

We now extend the $V$-fold jackknife framework from a scalar parameter $\Psi(P_0) \in \R$ to a vector-valued parameter $\Psi(P_0) = (\Psi_1(P_0), \dots, \Psi_m(P_0))^\transpose \in \R^m$, where $m$ is a fixed positive integer (not growing with $n$). Throughout this section, $\varphi_j$ denotes the influence function for the $j$th component, consistent with the notation $\varphi$ used in Section~2 for the scalar case. In Section~\ref{sec:slower_rates}, we show that all results in this section extend to the generalized asymptotic linearity framework of Assumption~\ref{assump:slower_rate}, where the influence function variance may diverge. A leading example is the bivariate survival function evaluated on a grid: $\Psi_j(P_0) = S(t_{1,j}, t_{2,j}) = P_0(T_1 > t_{1,j}, T_2 > t_{2,j})$ for $j = 1, \dots, m$. The goal is to construct a simultaneous confidence band $[\hat\Psi_j \pm q_\alpha \cdot \widehat{SE}_j]$ that covers all $m$ components simultaneously with probability at least $1 - \alpha$.

\subsection{Vector-valued V-fold jackknife}

Suppose that for each $j = 1, \dots, m$, the estimator $\hat\Psi_j(P_n)$ admits an asymptotically linear expansion
\[
\hat\Psi_j(P_n) - \Psi_j(P_0) = P_n \varphi_j + R_{n,j},
\]
where $\varphi_j = \varphi_{j,P_0}$ is a mean-zero influence function with $\E_{P_0}[\varphi_j^2] = \sigma_j^2 < \infty$ and $R_{n,j} = o_p(n^{-1/2})$. Define the $m \times m$ true covariance matrix $\bSigma = (\sigma_{jk})$ with $\sigma_{jk} = \E_{P_0}[\varphi_j \varphi_k]$.

For each fold $v = 1, \dots, V$ and component $j = 1, \dots, m$, the $V$-fold jackknife pseudo-value is
\[
IC_{Jack}^{(j)}(v) = V \hat\Psi_j(P_n) - (V-1) \hat\Psi_j(P_{n,-v}).
\]
Let $\hat\Psi_{Jack}^{(j)}(P_n) = V^{-1} \sum_{v=1}^V IC_{Jack}^{(j)}(v)$ denote the componentwise jackknife estimator, and define the $V$-fold jackknife covariance matrix estimator
\begin{equation}\label{eq:cov_jack}
    \hat\Sigma_{Jack}(j,k) = \frac{1}{V-1} \sum_{v=1}^V \big(IC_{Jack}^{(j)}(v) - \hat\Psi_{Jack}^{(j)}(P_n)\big) \big(IC_{Jack}^{(k)}(v) - \hat\Psi_{Jack}^{(k)}(P_n)\big).
\end{equation}
The componentwise jackknife variance and standard error are
\[
\hat\sigma_j^2 = \hat\Sigma_{Jack}(j,j), \qquad \widehat{SE}_j = \frac{\hat\sigma_j}{\sqrt{V}},
\]
and the jackknife correlation matrix is
\begin{equation}\label{eq:corr_jack}
    \hat R_{Jack}(j,k) = \frac{\hat\Sigma_{Jack}(j,k)}{\sqrt{\hat\Sigma_{Jack}(j,j)\,\hat\Sigma_{Jack}(k,k)}}.
\end{equation}

\begin{theorem}[Multivariate asymptotic equivalence, fixed $V$]\label{thm:multi_consistency}
Let $\Var(\varphi_j(O)) = \sigma_j^2$ and $\E_{P_0}[\varphi_j \varphi_k] = \sigma_{jk}$ for each $j, k = 1, \dots, m$, where $m$ is fixed. Suppose that for each component $j$, the remainder terms satisfy $\max_{1\le v\le V}|R_j(P_{n,-v}, P_0)| = o_p(n^{-1/2})$. Define the oracle covariance estimator $\hat S_W^2(j,k) = (V-1)^{-1} \sum_{v=1}^V (W_v^{(j)} - \bar W^{(j)})(W_v^{(k)} - \bar W^{(k)})$ where $W_v^{(j)} = P^1_{n,v} \varphi_j$. Then, for $V$ fixed:
\[
\frac{n}{V}\hat\Sigma_{Jack}(j,k) = \frac{n}{V}\hat S_W^2(j,k) + o_p(1).
\]
Furthermore, when folds are of exactly equal size $n/V$, $\frac{n}{V}\hat S_W^2(j,k)$ is an unbiased estimator of $\sigma_{jk}$ with $\E[\frac{n}{V}\hat S_W^2(j,k)] = \sigma_{jk}$.
\end{theorem}

\begin{proof}
See Appendix~\ref{app:proofs}.
\end{proof}

\subsection{Joint distribution of the Studentized vector and simultaneous critical values}

When $V$ is fixed, Theorem~\ref{thm:fixedV} established that the scalar Studentized statistic converges in distribution to $t_{V-1}$. For simultaneous inference over $m$ components, we characterize the joint distribution of the componentwise-Studentized vector. Define
\[
\bT = \left(\frac{\sqrt{V}(\hat\Psi_{Jack}^{(1)}(P_n) - \Psi_1(P_0))}{\hat\sigma_1}, \dots, \frac{\sqrt{V}(\hat\Psi_{Jack}^{(m)}(P_n) - \Psi_m(P_0))}{\hat\sigma_m}\right)^\transpose.
\]

\begin{theorem}[Componentwise-Studentized distribution, fixed $V$]\label{thm:multi_t}
Suppose the conditions of Theorem~\ref{thm:multi_consistency} hold. \textcolor{black}{Assume also that $R_j(P_n,P_0)=o_p(n^{-1/2})$ and $0<\sigma_j^2<\infty$ for each $j=1,\dots,m$.} Let $R_0 = D^{-1} \bSigma D^{-1}$ denote the true correlation matrix, where $D = \text{diag}(\sigma_1, \dots, \sigma_m)$. As $n \to \infty$ with $V$ fixed, the vector $\bT$ converges in distribution: $\bT \CiD \bT^\infty$, where
\begin{equation}\label{eq:studentized_rep}
    \bT^\infty = \left(\frac{\sqrt{V}\,\bar{Z}^{(1)}}{S_Z^{(1)}},\;\dots,\;\frac{\sqrt{V}\,\bar{Z}^{(m)}}{S_Z^{(m)}}\right)^\transpose,
\end{equation}
with $\bZ_1, \dots, \bZ_V \stackrel{i.i.d.}{\sim} \Normal_m(\bm{0}, R_0)$, $\bar{\bZ} = V^{-1}\sum_{v=1}^V \bZ_v$, and $S_Z^{(j)} = [(V-1)^{-1}\sum_{v=1}^V(Z_v^{(j)}-\bar{Z}^{(j)})^2]^{1/2}$. Each marginal $T_j^\infty$ has an exact $t_{V-1}$ distribution.
\end{theorem}

\begin{proof}
See Appendix~\ref{app:proofs}.
\end{proof}

We stress that \cref{thm:multi_t} is a \emph{fixed-$V$} convergence-in-distribution result: the limit law $\bT^\infty$ is governed by the true influence-curve correlation matrix $R_0$. At fixed $V$, however, $R_0$ is not consistently estimable, since the jackknife correlation estimate $\hat R_{Jack}$ is built from only $V$ fold-level vectors and is singular whenever $m > V-1$. In this sense the limiting joint distribution is identified as a functional of $R_0$, but $R_0$ itself---and hence the simultaneous critical value $q_\alpha$ computed from it---is not identifiable from the data at fixed $V$. Consistent estimation of $R_0$, and therefore formal plug-in validity of the simultaneous band constructed below, requires $V\to\infty$ (slowly); this point is developed in \cref{sec:sim_km} and in the diverging-$V$ analysis. We nonetheless find that fixed, moderate $V$ performs well empirically.

The asymptotic normality of the fold-level vectors $\bW_v$ places the problem in the classical framework of simultaneous inference for the mean of a multivariate normal distribution. However, a crucial distinction arises for the componentwise-Studentized vector $\bT$. While a standard multivariate $t$-distribution models a vector scaled by a single shared $\chi^2$ denominator \citep[Type~I in the classification of][]{KotzNadarajah2004}, the limit distribution $\bT^\infty$ features $m$ dependent denominators corresponding to the diagonal elements of a sample covariance matrix (which follows a Wishart distribution). Because the sample variances $(S_Z^{(1)})^2, \dots, (S_Z^{(m)})^2$ are correlated but not identically equal, approximating this joint distribution with a standard multivariate $t$-distribution using a single shared denominator will lead to an incorrect limit and substantial undercoverage, particularly when $V$ is small. Thus, simultaneous critical values must be simulated directly from the correct joint law of $\bT^\infty$, which accurately preserves the dependency structure of both the numerators and the specific denominators.

The confidence band (\ref{eq:simul_band}) below is centered at $\hat\Psi_j(P_n)$, while the Studentized statistic uses $\hat\Psi_{Jack}^{(j)}(P_n)$. By the same argument as in \cref{lem:centering}, applied componentwise, $\hat\Psi_{Jack}^{(j)}(P_n) - \hat\Psi_j(P_n) = -\hat b_j$ where $\hat b_j = (V-1)(\bar\Psi_{j,(-v)} - \hat\Psi_j(P_n))$ is the componentwise jackknife bias estimate, and $\hat b_j = o_p(\widehat{SE}_j)$. Hence the centering mismatch does not affect coverage.

A $(1-\alpha)$-level simultaneous confidence band is
\begin{equation}\label{eq:simul_band}
    \hat\Psi_j \pm q_\alpha \cdot \widehat{SE}_j, \qquad j = 1, \dots, m,
\end{equation}
where $q_\alpha$ satisfies
\begin{equation}\label{eq:q_def}
    P\!\left(\max_{1 \le j \le m} |T_j^\infty| \le q_\alpha\right) = 1 - \alpha.
\end{equation}
Since the distribution of $\max_j |T_j^\infty|$ does not admit a closed form for general $R_0$ and $m$, the critical value $q_\alpha$ is computed by Monte Carlo simulation.

\subsection{Monte Carlo computation of the critical value}\label{sec:mc_critical}

Given the estimated correlation matrix $\hat R_{Jack}$, the critical value $q_\alpha$ is approximated by simulating directly from the joint componentwise-Studentized limit distribution $\bT^\infty$. This relies on $V$ independent Gaussian draws per Monte Carlo step to ensure the correct Wishart-diagonal denominator structure is enforced. The procedure is given in \cref{alg:mc_critical}.

\begin{algorithm}[tb]
\color{black}
\caption{\color{black}Monte Carlo computation of the simultaneous critical value $\hat q_\alpha$}
\label{alg:mc_critical}
\begin{algorithmic}[1]
\STATE \textbf{Input:} Estimated correlation matrix $\hat R_{Jack} \in \R^{m \times m}$, number of folds $V$, confidence level $1-\alpha$, Monte Carlo size $B$ (e.g., $B = 10{,}000$).
\STATE \textbf{Output:} Critical value $\hat q_\alpha$ for the simultaneous confidence band.
\medskip
\STATE Compute eigendecomposition $\hat R_{Jack} = U \Lambda U^\transpose$. 
\FOR{$b = 1$ \TO $B$}
    \FOR{$v = 1$ \TO $V$}
        \STATE Draw $\bZ_v^{(b)} \sim \Normal_m(\bm{0}, \hat R_{Jack})$ via $\bZ_v^{(b)} = U \Lambda^{1/2} \bm{\eta}_v^{(b)}$, where $\bm{\eta}_v^{(b)} \sim \Normal_m(\bm{0}, I_m)$. (This procedure remains valid when $m > V-1$ and $\hat R_{Jack}$ is singular; see \cref{sec:eigen_validity}.)
    \ENDFOR
    \STATE Compute the sample mean vector $\bar{\bZ}^{(b)} = \frac{1}{V} \sum_{v=1}^V \bZ_v^{(b)}$.
    \FOR{$j = 1$ \TO $m$}
        \STATE Compute the sample standard deviation $S_Z^{(j, b)} = \left[ \frac{1}{V-1} \sum_{v=1}^V \big(Z_{v, j}^{(b)} - \bar{Z}_j^{(b)}\big)^2 \right]^{1/2}$.
        \STATE Form the Studentized statistic $T_j^{(b)} = \sqrt{V}\, \bar{Z}_j^{(b)} / S_Z^{(j, b)}$.
    \ENDFOR
    \STATE Record $M^{(b)} = \max_{1 \le j \le m} |T_j^{(b)}|$.
\ENDFOR
\STATE \textbf{Return} $\hat q_\alpha = Q_{1-\alpha}\!\big(\{M^{(b)}\}_{b=1}^B\big)$, where $Q_p$ denotes the $p$-th sample quantile.
\end{algorithmic}
\end{algorithm}

Under the asymptotic normality of the fold-level vectors, the sample correlation $\hat R_{Jack}$ is the maximum likelihood estimate of $R_0$ based on $V$ independent $\Normal_m$ observations. As $V \to \infty$, $\hat R_{Jack} \CiP R_0$, $S_Z^{(j)} \CiP 1$, and the joint componentwise-Studentized distribution converges to the supremum of $\Normal_m(\bm{0}, R_0)$.

\begin{theorem}[Plug-in simultaneous validity under $V \to \infty$]\label{thm:plugin_simult_divergingV}
\begingroup\color{black}
Let $m$ be fixed. Suppose the componentwise asymptotically linear expansions in this section hold with $R_j(P_n,P_0)=o_p(n^{-1/2})$, $0<\sigma_j^2<\infty$, and $\E_{P_0}[\varphi_j^4]<\infty$ for every $j=1,\dots,m$. Assume the folds are exactly balanced, $V=V_n\to\infty$, $V=o(n)$, and, for every component,
\[
d_{n,v}^{(j)}=R_j(P_n,P_0)-R_j(P_{n,-v},P_0),
\qquad
\|d_{n,j}\|_{V,2}
=
\left\{\frac{1}{V}\sum_{v=1}^V \big(d_{n,v}^{(j)}\big)^2\right\}^{1/2}
=o_p\!\left(\frac{1}{\sqrt n\,V}\right).
\]
Then the jackknife covariance matrix is entrywise consistent for the influence-curve covariance matrix:
\[
\frac{n}{V}\hat\Sigma_{Jack}(j,k)\CiP \sigma_{jk},
\qquad j,k=1,\dots,m,
\]
and consequently $\hat R_{Jack}\CiP R_0$ entrywise. Let $\hat q_\alpha$ be the critical value computed from the plug-in componentwise-Studentized law in Algorithm~\ref{alg:mc_critical}, with Monte Carlo error $o_p(1)$. Then the simultaneous confidence band satisfies
\[
P\!\left(
\Psi_j(P_0)\in
\left[\hat\Psi_j(P_n)-\hat q_\alpha\widehat{SE}_j,\,
\hat\Psi_j(P_n)+\hat q_\alpha\widehat{SE}_j\right]
\ \forall j=1,\dots,m
\right)
\to 1-\alpha .
\]
\endgroup
\end{theorem}

\begin{proof}
See Appendix~\ref{app:proofs}.
\end{proof}

\subsection{Validity of the Eigendecomposition and the Role of $V$ in Simultaneous Inference}\label{sec:eigen_validity}

When $m > V - 1$, the estimated correlation matrix $\hat R_{Jack}$ has rank at most $V - 1 < m$, and the Monte Carlo draws in Step~1 are generated through the eigendecomposition $\bZ_v^{(b)} = U \Lambda^{1/2}\bm{\eta}_v^{(b)}$. This is exact sampling from the degenerate plug-in distribution $\Normal_m(\bm{0}, \hat R_{Jack})$: the zero-eigenvalue directions carry no variation and are therefore correctly excluded from the simulation.

The real inferential question is not whether this eigendecomposition is legitimate, but whether $\hat R_{Jack}$ is close enough to the true correlation matrix $R_0$. When $V$ is fixed, Theorem~\ref{thm:multi_t} identifies the correct componentwise-Studentized limit distribution conditional on the true correlation matrix $R_0$, and Theorem~\ref{thm:multi_consistency} shows that the jackknife covariance estimator is asymptotically unbiased. \textcolor{black}{Formal plug-in simultaneous validity requires consistent estimation of $R_0$ and hence $V \to \infty$; this is established in \cref{thm:plugin_simult_divergingV} under the same type of fold-level remainder-difference condition used for scalar diverging-$V$ variance consistency.} \textcolor{black}{For the variance results developed here, $V$ must also satisfy the fold-level remainder-difference condition from \cref{thm:divergingV}; under the sufficient structural conditions in Appendix~\ref{app:V_rate}, slow choices such as $V \asymp \log n$ are covered.}

To understand why the plug-in procedure can nevertheless work well when $V$ is small, it is useful to introduce the notion of \emph{effective rank}. For a positive semidefinite matrix $A$, define
\begin{equation}\label{eq:effective_rank}
r(A) = \frac{\operatorname{tr}(A)}{\|A\|_{\operatorname{op}}},
\end{equation}
where $\|A\|_{\operatorname{op}}$ denotes the operator norm. For a correlation matrix $R_0$, the trace is exactly $m$, so $r(R_0) = m / \lambda_{\max}(R_0)$. When the eigenvalues of $R_0$ decay rapidly, the effective rank can be much smaller than the ambient dimension $m$.

Concentration inequalities for sample covariance operators show that convergence rates depend on this effective rank rather than directly on $m$. Let $\hat\Sigma_{Jack}$ denote the sample covariance matrix of the fold-level vectors and let $\Sigma_0$ denote the corresponding population covariance. Results of \citet{Vershynin2012} and \citet{KoltchinskiiLounici2017} imply
\begin{equation}\label{eq:eff_rank_bound_cov}
\E\big[\|\hat \Sigma_{Jack} - \Sigma_0\|_{\operatorname{op}}\big]
\le
C \cdot \|\Sigma_0\|_{\operatorname{op}}
\left(\sqrt{\frac{r(\Sigma_0)}{V}} + \frac{r(\Sigma_0)}{V}\right),
\end{equation}
for a universal constant $C$. Because the sample correlation matrix is obtained by symmetric diagonal scaling, $\hat R_{Jack} = \hat D^{-1}\hat\Sigma_{Jack}\hat D^{-1}$ with $\hat D = \operatorname{diag}(\hat\Sigma_{Jack})^{1/2}$, standard algebraic decompositions show that $\|\hat R_{Jack} - R_0\|_{\operatorname{op}}$ is bounded by a constant multiple of $\|\hat \Sigma_{Jack} - \Sigma_0\|_{\operatorname{op}}$ whenever the marginal variances are bounded away from zero. Thus the same effective-rank rate carries over to the correlation scale.

Consequently, even when $V \ll m$, the sample correlation matrix $\hat R_{Jack}$ can still be close to $R_0$ whenever $V$ is large relative to the effective rank $r(R_0)$. For smooth functional parameters such as survival curves or dose-response curves evaluated on a grid, adjacent evaluation points are highly correlated and the eigenvalues of $R_0$ often decay quickly. In that case, the dominant eigendirections of $R_0$ are captured by the rank-deficient $\hat R_{Jack}$ even when $V$ is moderate, which helps explain the robust finite-sample performance seen in our simulations. This bound-based explanation supports the practical use of small $V$, but it does not replace the formal asymptotic requirement that $V \to \infty$ for plug-in simultaneous validity.

This perspective also suggests a practical finite-sample tuning rule. Along a candidate sequence such as $V \in \{5, 10, 20, 40, \dots\}$, compute the corresponding estimated correlation matrix $\hat R_{Jack}$ and its effective rank
\begin{equation}\label{eq:effective_rank_hat}
\hat r_V := r(\hat R_{Jack}) = \frac{\operatorname{tr}(\hat R_{Jack})}{\|\hat R_{Jack}\|_{\operatorname{op}}} = \frac{m}{\lambda_{\max}(\hat R_{Jack})},
\end{equation}
where the final equality uses that $\hat R_{Jack}$ is a correlation matrix and therefore has trace $m$. One may then choose the smallest $V$ such that
\begin{equation}\label{eq:V_stopping_rule}
V \ge 2 \hat r_V.
\end{equation}
This rule should be viewed as a heuristic implementation of the effective-rank argument, not as a substitute for the formal asymptotic requirement above. When such tuning is not carried out, the simulations in Sections~\ref{sec:sim_km} and~\ref{sec:sim_hal} suggest that $V = 20$ is a reasonable default in the settings we study.

For functional estimators with growing dimension, a natural next step would be to establish uniform convergence of the jackknife correlations to the true correlation function, which would then support Gaussian-process-style simultaneous inference.

\begin{remark}
    Constructing a simultaneous band requires the distribution of the \emph{maximum} over the $m$ grid points of the componentwise-Studentized statistics, $\max_{1 \le j \le m} |\sqrt{V}\,(\hat\Psi_{Jack}^{(j)}(P_n) - \Psi_j(P_0))/\hat\sigma_j|$. Because these statistics are built from the $V$ fold-level pseudo-value vectors, this is exactly a maximum of (approximately) high-dimensional sample means, with $V$ playing the role of the sample size and $m$ the role of the dimension; it therefore connects directly to the literature on Gaussian and multiplier-bootstrap approximation for maxima of high-dimensional means \citep{ChernozhukovChetverikovKato2013}. However, those results are not needed for the development of our main theory and require additional growth conditions linking $V$ and $m$. Since our main message is that formal simultaneous validity follows once $V \to \infty$ while the fold-level remainder condition remains valid---for example, for slow choices such as $V \asymp \log n$ under the sufficient bounds in Appendix~\ref{app:V_rate}---we defer this connection to Appendix~\ref{app:gaussian_approx}.
\end{remark}

\subsection{Bias-corrected simultaneous confidence bands}\label{sec:bias_corrected_bands}

When the asymptotic linear expansion is not uniformly negligible across all $m$ grid points, the simultaneous coverage of \eqref{eq:simul_band} can be eroded by finite-sample bias. Following  \citep[e.g.,][]{CaivanderLaan2020}, we address this by \emph{enlarging} the confidence intervals to account for the estimated bias. Extending the scalar jackknife bias estimate \eqref{eq:jack_bias_scalar} to the $j$th component, define
\begin{equation}\label{eq:jack_bias}
    \hat b_j = (V-1)\left(\bar\Psi_{j,(-v)} - \hat\Psi_j(P_n)\right), \qquad \bar\Psi_{j,(-v)} = \frac{1}{V} \sum_{v=1}^V \hat\Psi_j(P_{n,-v}).
\end{equation}
This estimate is computable from the same fold-level estimates used for variance estimation and requires no additional computation; it is reused, as a post-hoc reliability diagnostic, in Appendix~\ref{sec:app_diagnostic}. The bias-corrected simultaneous confidence band is
\begin{equation}\label{eq:bc_band}
    \left[\hat\Psi_j(P_n) - q_\alpha \cdot \widehat{SE}_j - |\hat b_j|, \;\; \hat\Psi_j(P_n) + q_\alpha \cdot \widehat{SE}_j + |\hat b_j|\right], \qquad j = 1, \dots, m,
\end{equation}
where $\hat b_j$ is the jackknife bias estimate defined in (\ref{eq:jack_bias}).

\begin{theorem}[Bias-corrected simultaneous coverage]\label{thm:bias_corrected}
Suppose the conditions of Theorem~\ref{thm:multi_t} hold. Additionally, suppose the jackknife bias estimate satisfies {\color{black} the uniform relative-error condition
\[
\varepsilon_n
=
\max_{1 \le j \le m}
\frac{|\hat b_j-b_j|}{\widehat{SE}_j}
\CiP 0,
\]
where $b_j = \E[\hat\Psi_j(P_n)] - \Psi_j(P_0)$ is the true bias.} Then
\[
P\!\left(\Psi_j(P_0) \in \left[\hat\Psi_j(P_n) - q_\alpha \widehat{SE}_j - |\hat b_j|, \;\; \hat\Psi_j(P_n) + q_\alpha \widehat{SE}_j + |\hat b_j|\right] \;\;\forall\, j = 1, \dots, m\right) \ge 1 - \alpha + o(1).
\]
\end{theorem}

\begin{proof}
See Appendix~\ref{app:proofs}.
\end{proof}

The bias correction may also be paired with the post-hoc diagnostic of Appendix~\ref{sec:app_diagnostic}, which uses the same jackknife bias estimate $\hat b_j$ to flag grid points whose linearization is under stress, complementing the widening provided by the bias-corrected band.  \cref{alg:simultaneous} summarizes the complete procedure for constructing simultaneous $V$-fold jackknife confidence bands.

\begin{algorithm}[tb]
\color{black}
\caption{\color{black}Simultaneous $V$-fold jackknife confidence band}
\label{alg:simultaneous}
\begin{algorithmic}[1]
\STATE \textbf{Input:} Estimator $\hat\Psi = (\hat\Psi_1, \dots, \hat\Psi_m)$, observations $O_1,\dots,O_n$, number of folds $V \ge 2$, confidence level $1-\alpha$, Monte Carlo size $B$ (e.g., $B = 10{,}000$).
\STATE \textbf{Output:} $(1-\alpha)$-level simultaneous confidence band for $(\Psi_1(P_0), \dots, \Psi_m(P_0))$.
\medskip
\STATE Compute $\hat\Psi_j(P_n)$ for $j = 1, \dots, m$ using all $n$ observations.
\STATE Partition $\{1, \dots, n\}$ into $V$ disjoint folds $I_1, \dots, I_V$ of approximately equal size.
\FOR{$v = 1$ \TO $V$}
    \STATE Compute $\hat\Psi_j(P_{n,-v})$ for all $j = 1, \dots, m$.
    \FOR{$j = 1$ \TO $m$}
        \STATE Compute the pseudo-value $IC_{Jack}^{(j)}(v) = V \hat\Psi_j(P_n) - (V-1) \hat\Psi_j(P_{n,-v})$.
    \ENDFOR
\ENDFOR
\STATE Compute the jackknife covariance matrix $\hat\Sigma_{Jack}$ via \eqref{eq:cov_jack}, the standard errors $\widehat{SE}_j = \hat\sigma_j / \sqrt{V}$, and the correlation matrix $\hat R_{Jack}$ via \eqref{eq:corr_jack}.
\STATE Compute the critical value $\hat q_\alpha$ via \cref{alg:mc_critical}.
\STATE \textbf{Return} the simultaneous confidence band: $\hat\Psi_j(P_n) \pm \hat q_\alpha \cdot \widehat{SE}_j$ for $j = 1, \dots, m$. Optionally, apply the bias correction \eqref{eq:bc_band}. See Appendix~\ref{sec:app_diagnostic} for an optional diagnostic that highlights grid points with large $|\hat b_j|/\widehat{SE}_j$.
\end{algorithmic}
\end{algorithm}

\section{Extension to Generalized Asymptotically Linear Estimators and Slower Convergence Rates}
\label{sec:slower_rates}

The theory developed in the preceding sections assumes the target parameter $\Psi(P_0)$ is pathwise differentiable, meaning the estimator converges at the parametric $n^{-1/2}$ rate with an influence function $\varphi$ having finite variance $\sigma^2 = \E_{P_0}[\varphi^2] < \infty$. However, in many modern semiparametric and nonparametric applications, researchers target non-pathwise differentiable parameters. A prominent example is the causal dose-response curve for a continuous treatment evaluated at a specific point, estimated via the highly adaptive lasso (HAL) \citep{vanderLaan2023, Shi2024}.

In such settings, the estimator may still admit a generalized asymptotic linear representation, but the influence curve depends on the sample size $n$, and its variance diverges. Here \emph{first-order} HAL refers to HAL with one degree of smoothness ($k=1$, a piecewise-linear spline basis; the $k$-th order spline basis functions are defined in \citet[Definition~7]{vanderLaan2023}), the smoothness order used throughout this paper and in the simulations of Section~\ref{sec:sim_hal}. For the first-order HAL dose-response curve, the variance of the (sample-size-dependent) influence function scales with the effective dimension $J_n$ of the oracle working model, $\sigma_n^2 = \Var_{P_0}(\varphi_n(O)) \asymp J_n \to \infty$ with $J_n \asymp n^{1/5}$, so $\sigma_n \asymp n^{1/10}$, yielding the slower stabilization rate $\sqrt{n/J_n}$. Pointwise asymptotic normality of the HAL regression at this rate is established by \citet[Corollary~1]{vanderLaan2023} for smoothness order $k \ge 1$, and the corresponding plug-in functional---the dose-response curve is a non-pathwise-differentiable functional of the HAL regression---is treated by \citet[Appendix~K, Theorem~23 and Corollary~6]{vanderLaan2023}, which give the slower-rate $(n/J_n)^{1/2}$ functional normality for arbitrary (non-pathwise-differentiable) functions of sieve and HAL-MLEs; see also \citet{Shi2024} for the dose-response specialization.

A profound mathematical and operational advantage of the $V$-fold jackknife is its inherent \textit{scale invariance}. Because the jackknife pseudo-values are computed directly from the estimator, their empirical dispersion automatically tracks the inflating variance of the influence curve. Consequently, the unknown diverging rate $\sigma_n \asymp \sqrt{J_n}$ perfectly cancels out in the Studentized $t$-statistic. This allows practitioners to obtain valid confidence intervals using the exact same jackknife procedure without needing to analytically derive, estimate, or explicitly divide by the effective dimension $J_n$.

We formalize this by extending our fixed-$V$ and diverging-$V$ results to a sequence of sample-size-dependent influence functions.

\begin{assumption}[Generalized Asymptotic Linearity]\label{assump:slower_rate}
Suppose the estimator admits the expansion
\begin{equation}\label{eq:AL_slower}
    \hat\Psi(P_n) - \Psi(P_0) = P_n \varphi_n + R_n(P_n, P_0),
\end{equation}
where $\varphi_n = \varphi_{n, P_0}$ is a sequence of mean-zero influence functions with diverging variance $\sigma_n^2 = \E_{P_0}[\varphi_n^2(O)] \to \infty$ as $n \to \infty$. We assume:
\begin{enumerate}
    \item \textbf{Scaled Remainder:} The remainder terms are asymptotically negligible relative to the standard error: $R_n(P_n, P_0) = o_p(\sigma_n n^{-1/2})$ and $\max_{1 \le v \le V} |R_n(P_{n,-v}, P_0)| = o_p(\sigma_n n^{-1/2})$.
    \item \textbf{Lindeberg Condition:} The standardized influence functions satisfy the Lindeberg condition: for any $\epsilon > 0$,
    \[
    \lim_{n \to \infty} \frac{1}{\sigma_n^2} \E_{P_0}\left[ \varphi_n^2(O) \Ind\left\{|\varphi_n(O)| > \epsilon \sigma_n \sqrt{n} \right\} \right] = 0.
    \]
\end{enumerate}
\end{assumption}

\begin{remark}
Because $\sigma_n \to \infty$, the generalized remainder condition $o_p(\sigma_n n^{-1/2})$ is strictly weaker than the standard $o_p(n^{-1/2})$ requirement. This mathematical relaxation aligns with the intuition that when the target parameter converges at a slower rate (i.e., has a larger standard error), a proportionally larger finite-sample bias is statistically tolerable for valid coverage. 
\end{remark}

\subsection{Fixed-$V$ inference for slower rates}

We first generalize the fixed-$V$ framework.

\begin{theorem}[Fixed-$V$ Inference for Slower Rates]\label{thm:fixedV_slower}
Under Assumption~\ref{assump:slower_rate}, for any fixed $V \ge 2$, the Studentized $V$-fold jackknife statistic adapts to the unknown convergence rate $\sigma_n$, yielding
\[
\frac{\sqrt{V}\,(\hat\Psi_{Jack}(P_n)-\Psi(P_0))}{\hat S_{Jack}}
\ \CiD \ t_{V-1}.
\]
\end{theorem}

\begin{proof}
See Appendix~\ref{app:proofs}.
\end{proof}

\begin{corollary}[Confidence interval centered at $\hat\Psi(P_n)$]
\label{cor:CI_slower}
Under the conditions of \cref{thm:fixedV_slower},
\[
\frac{\sqrt{V}\,\big(\hat\Psi(P_n)-\Psi(P_0)\big)}{\hat S_{Jack}}
\ \CiD \ t_{V-1}.
\]
Hence, a valid $(1-\alpha)$-level confidence interval for $\Psi(P_0)$ is given by $\hat\Psi(P_n) \pm t_{V-1,\,1-\alpha/2} \cdot \frac{\hat S_{Jack}}{\sqrt{V}}$.
\end{corollary}
\begin{proof}
See Appendix~\ref{app:proofs}.
\end{proof}

\subsection{Diverging-$V$ consistency for slower rates}

The diverging-$V$ consistency also naturally extends to this regime. We show that the relative ratio of the jackknife variance estimator to the true diverging variance converges to one.

\begin{theorem}[Diverging-$V$ Consistency for Slower Rates]\label{thm:divergingV_slower}
Under Assumption \ref{assump:slower_rate}, define the scaled variance estimator $\hat{s}_n^2 = \frac{n}{V} \hat S_{Jack}^2$. Suppose that a Weak Law of Large Numbers holds for the squared influence curve: $\frac{1}{n \sigma_n^2} \sum_{i=1}^n \varphi_n^2(O_i) \CiP 1$ (which is satisfied if $\E_{P_0}[\varphi_n^4(O)] = o(n \sigma_n^4)$). 

\textcolor{black}{For the slower-rate setting, define}
\[
\textcolor{black}{
d_{n,v}=R_n(P_n,P_0)-R_n(P_{n,-v},P_0),
\qquad
\|d_n\|_{V,2}
=
\left(\frac{1}{V}\sum_{v=1}^V d_{n,v}^2\right)^{1/2}.}
\]
\textcolor{black}{Assume $V \to \infty$, $V=o(n)$, and}
\[
\textcolor{black}{
\|d_n\|_{V,2}
=
o_p\!\left(\frac{\sigma_n}{\sqrt n\,V}\right).
}
\]
Then
\[
\frac{\hat{s}_n^2}{\sigma_n^2} \;\CiP\; 1.
\]
\end{theorem}

\begin{proof}
See Appendix~\ref{app:proofs}.
\end{proof}

Appendix~\ref{app:V_rate} gives sufficient TMLE, one-step, and AIPW-style structural conditions under which this stability condition follows from the explicit upper-growth bound $V=o\!\left(n\sigma_n^2/\tilde\sigma_n^4\right)$, where $\tilde\sigma_n^2$ is the variance scale of the nuisance estimator's linear approximation; explicit HAL dose-response examples are worked out there.

\begin{corollary}[Diverging-$V$ Wald inference for slower rates]\label{cor:divergingV_wald_slower}
\begingroup\color{black}
Under the assumptions of \cref{thm:divergingV_slower}, suppose also that the triangular-array linear term satisfies
\[
\frac{\sqrt n\,P_n\varphi_n}{\sigma_n}\CiD \Normal(0,1).
\]
Then
\[
\frac{\hat\Psi(P_n)-\Psi(P_0)}{\hat s_n/\sqrt n}
\CiD \Normal(0,1),
\qquad
\hat s_n^2=\frac{n}{V}\hat S_{Jack}^2.
\]
Consequently,
\[
\hat\Psi(P_n)\pm t_{V-1,\,1-\alpha/2}\frac{\hat s_n}{\sqrt n}
=
\hat\Psi(P_n)\pm t_{V-1,\,1-\alpha/2}\frac{\hat S_{Jack}}{\sqrt V}
\]
has asymptotic coverage $1-\alpha$, since $t_{V-1,\,1-\alpha/2}\to z_{1-\alpha/2}$ as $V\to\infty$.
\endgroup
\end{corollary}
\begin{proof}
See Appendix~\ref{app:proofs}.
\end{proof}

\subsection{Generalized simultaneous confidence bands}

The simultaneous confidence band results developed in the preceding section also extend to the generalized asymptotic linearity framework of Assumption~\ref{assump:slower_rate}. Consider a vector-valued parameter $\Psi(P_0) = (\Psi_1(P_0), \dots, \Psi_m(P_0))^\transpose$ where each component admits
\[
\hat\Psi_j(P_n) - \Psi_j(P_0) = P_n \varphi_{n,j} + R_{n,j}(P_n, P_0),
\]
with $\sigma_{n,j}^2 = \E_{P_0}[\varphi_{n,j}^2] \to \infty$, possibly at different rates for different $j$. The componentwise Studentization eliminates $\sigma_{n,j}$ from each component independently:
\[
T_j = \frac{\sqrt{V}(\hat\Psi_{Jack}^{(j)}(P_n) - \Psi_j(P_0))}{\hat\sigma_j} \CiD T_j^\infty \sim t_{V-1},
\]
with the joint distribution determined by the correlation matrix $R_{0,n} = D_n^{-1} \bSigma_n D_n^{-1}$, where $D_n = \text{diag}(\sigma_{n,1}, \dots, \sigma_{n,m})$ and $\bSigma_n(j,k) = \E_{P_0}[\varphi_{n,j} \varphi_{n,k}]$. This correlation matrix has entries in $[-1, 1]$ regardless of the scale of the variances. Consequently, all results in the preceding section---including the componentwise-Studentized distribution (\cref{thm:multi_t}), the Monte Carlo critical value algorithm, and the bias-corrected bands (\cref{thm:bias_corrected})---hold verbatim under Assumption~\ref{assump:slower_rate} applied componentwise, since all proofs are scale-invariant.

\subsection{Application: HAL dose-response curve}

The causal dose-response curve $\Psi(a) = \E_{P_0}[\E_{P_0}(Y \mid A = a, W)]$ estimated via the HAL plug-in $\hat\Psi(a) = \frac{1}{n}\sum_{i=1}^n \hat Q_n(W_i, a)$ provides the primary example of the generalized framework, where $\hat Q_n$ denotes the highly adaptive lasso estimator for the conditional mean outcome $Q_0(W, A) = \E_{P_0}[Y \mid A, W]$.

To establish the pointwise asymptotic normality of the HAL estimator without delving into the complex proof details of $L_1$-penalization, we outline the logical progression developed by \citet{vanderLaan2023} (the $k$-th order spline basis is constructed in their Definition~7). The analysis centers on the construction of a designed oracle working model. Let $\mathcal{R}_{0,n}$ denote a fixed, oracle set of spline basis functions. This set is designed to be as sparse as possible---meaning its dimension $J_n = |\mathcal{R}_{0,n}|$ grows slowly with $n$---while still being sufficiently rich to maintain the optimal approximation error rate for the true function $Q_0$. 

Let $D(\mathcal{R}_{0,n})$ denote the finite-dimensional linear working model spanned by these $J_n$ basis functions, and let $Q_{0,n} = \argmin_{Q \in D(\mathcal{R}_{0,n})} P_0 L(Q)$ be the theoretical oracle maximum likelihood estimator (MLE) on this model under a loss function $L$. A standard analysis of the unpenalized empirical MLE on this fixed working model readily establishes its asymptotic normality. 

The key observation in \citet{vanderLaan2023} is that the actual HAL estimator $\hat Q_n$, despite selecting its basis adaptively via $L_1$-regularization, effectively solves the unpenalized score equations associated with the basis functions in the oracle model $D(\mathcal{R}_{0,n})$. Because $\hat Q_n$ solves the scores of this oracle model, it is asymptotically equivalent to the oracle MLE and seamlessly inherits its asymptotic normality. This yields the generalized asymptotic linear representation:
\[
(\hat Q_n - Q_{0,n})(x) \approx (P_n - P_0) D_{Q_{0,n},x},
\]
where $x = (W, a)$, and the influence-curve-like object $D_{Q_{0,n},x} = \sum_{j \in \mathcal{R}_{0,n}} S_{Q_{0,n}}(\phi_j^*) \phi_j^*(x)$ is a sum over $J_n$ terms. Here, $\{\phi_j^*\}_{j \in \mathcal{R}_{0,n}}$ forms an orthonormalized basis of $D(\mathcal{R}_{0,n})$ in $L^2(P_0)$, and $S_{Q_{0,n}}(\phi) = \frac{d}{d\epsilon} L(Q_{0,n} + \epsilon \phi)\big|_{\epsilon=0}$ is the corresponding score operator evaluated at $Q_{0,n}$. 

The variance of this expansion, $\sigma_n^2(x) = P_0\{D_{Q_{0,n},x}\}^2 = J_n \cdot \bar\sigma_n^2(x)$, diverges linearly with the dimension $J_n$ of the oracle model, while the normalized variance $\bar\sigma_n^2(x) = J_n^{-1} \sum_{j \in \mathcal{R}_{0,n}} \{\phi_j^*(x)\}^2$ remains bounded. (The bounded, normalized scale $\bar\sigma_n^2(x)$ should not be confused with the nuisance linearization scale $\tilde\sigma_n^2$ of Appendix~\ref{app:V_rate}: for the HAL plug-in, the nuisance linearization is $D_{Q_{0,n},x}$ itself, so that scale equals the diverging $\sigma_n^2(x) = J_n \bar\sigma_n^2(x)$, which is precisely what makes the comparability assumption $\tilde\sigma_n \asymp \sigma_n$ of Appendix~\ref{app:V_rate} natural here.) Standardizing by this diverging standard error establishes the pointwise asymptotic normality of the HAL estimator at a slower convergence rate \citep[Corollary~1, with oracle dimension $J_n \asymp n^{1/(2k^{*}+1)}$, $k^{*}=k+1$, for smoothness order $k\ge1$]{vanderLaan2023}: $(n/J_n)^{1/2}(\hat Q_n - Q_{0,n})(x) / \bar\sigma_n(x) \Longrightarrow_d \Normal(0, 1)$. Integrating over the empirical distribution of $W$ extends this pointwise linearization to the dose-response plug-in estimator $\hat\Psi(a)$, a non-pathwise-differentiable functional of $\hat Q_n$ whose plug-in functional normality at the slower rate $(n/J_n)^{1/2}$ is established in \citet[Appendix~K, Theorem~23 and Corollary~6]{vanderLaan2023}.

\textcolor{black}{The oracle model results of \citet{vanderLaan2023} suggest that Assumption~\ref{assump:slower_rate} is mathematically plausible for HAL dose-response estimators.} The relevant condition is not ``sparsity'' of the target function per se, but rather that the scores that the data-adaptive HAL fit actually solves approximate the corresponding oracle working-model scores well enough---so that the fold-level influence-curve contributions used by the jackknife reproduce the oracle expansion up to an $o_p(n^{-1/2})$ remainder. Precise theorems giving conditions of this type (control of the empirical HAL scores relative to the oracle scores) are available in \citet{vanderLaan2023}.\textcolor{black}{A rigorous theoretical verification for exact, data-adaptive HAL implementations is beyond the scope of this section, because the selected basis functions are themselves data-adaptive. In the oracle formulation, the scaled remainder condition is supported by the approximation behavior of the sparse working model, and the Lindeberg condition is supported by the representation of $D_{Q_{0,n},x}$ as a sum of bounded basis-function contributions.}

\textcolor{black}{Under Assumption~\ref{assump:slower_rate},} the $V$-fold jackknife confidence interval $\hat\Psi(a) \pm t_{V-1, 1-\alpha/2} \cdot \hat S_{Jack}(a) / \sqrt{V}$ therefore provides asymptotically valid pointwise inference for the dose-response curve. Crucially, the jackknife standard error $\hat S_{Jack}(a) / \sqrt{V} \asymp_p \sqrt{J_n / n}$ automatically adapts to the correct diverging rate without requiring the practitioner to know the effective dimension $J_n$. This scale invariance explains why the delta method---which requires explicitly extracting the active basis functions, computing the growing $J_n$-dimensional influence curve, and calculating its diverging variance---can fail numerically (producing NAs at rates up to 6.2\% in our simulations). In contrast, the $V$-fold jackknife reliably captures the diverging variance automatically through the fold-level variability of the pseudo-values. \textcolor{black}{The simulation study in Section~\ref{sec:sim_hal} provides empirical support for this scale-adaptive behavior.}

\section{Simulation Studies}\label{sec:simulations}

We evaluate the finite-sample performance of the $V$-fold jackknife through three simulation studies of increasing complexity. Section~\ref{sec:sim_ate} validates the fixed-$V$ $t_{V-1}$ inference from \cref{thm:fixedV} for a scalar parameter (the average treatment effect). Section~\ref{sec:sim_km} extends the evaluation to a functional parameter (the Kaplan--Meier survival curve), comparing pointwise and simultaneous confidence bands from \cref{thm:multi_t}. Section~\ref{sec:sim_hal} demonstrates the practical value of the approach for HAL-based dose-response curve estimation, where analytic standard errors may be unreliable.

\subsection{Scalar Inference: Average Treatment Effect}\label{sec:sim_ate}

We first examine pointwise confidence intervals for the average treatment effect $\Psi(P_0) = \E_{P_0}[\E_{P_0}(Y \mid A=1, W) - \E_{P_0}(Y \mid A=0, W)]$ estimated by targeted minimum loss-based estimation (TMLE) with Super Learner. The data generating process draws eight binary and continuous covariates $W = (W_1, \ldots, W_8)$, a binary treatment $A$, and a continuous outcome $Y$. We consider four scenarios in which the propensity score model $g$ is correctly specified, varying whether the outcome regression $Q$ is correctly specified or misspecified, and whether positivity violations are moderate or severe. Nuisance functions are estimated via Super Learner with library \texttt{(SL.glm, SL.step, SL.glm.interaction)}. We report results for $n \in \{200, 1000\}$ based on 500 replications.

We compare four approaches to constructing 95\% confidence intervals:
\begin{enumerate}[label=(\roman*)]
\item \emph{EIC-based}: the plug-in influence function variance estimator with normal critical values, as implemented in the \texttt{tmle} package \citep{Gruber2012};
\item \emph{Targeted bootstrap}: bootstrap standard errors ($B = 10{,}000$) with normal critical values;
\item \emph{Robust variance}: the robust variance estimator of \citet{tran2018robust}, as implemented in \texttt{ltmle};
\item \emph{$V$-fold jackknife} ($V \in \{5, 10, 20\}$): the proposed estimator with $t_{V-1}$ critical values.
\end{enumerate}
All four methods use the same TMLE point estimate; only the variance estimation and critical values differ.

\begin{table}[htbp]
\centering
\caption{Coverage (\%) and mean CI width for the average treatment effect via TMLE.
Scenarios~A--D have the propensity score model $g$ correctly specified.
A: $Q$ correct, moderate positivity; B: $Q$ correct, severe positivity;
C: $Q$ misspecified, moderate positivity; D: $Q$ misspecified, severe positivity. 500 replications.}
\label{tab:ate}
\small
\begin{tabular}{l cc cc cc cc}
\toprule
 & \multicolumn{2}{c}{Scenario A} & \multicolumn{2}{c}{Scenario B}
 & \multicolumn{2}{c}{Scenario C} & \multicolumn{2}{c}{Scenario D} \\
\cmidrule(lr){2-3} \cmidrule(lr){4-5} \cmidrule(lr){6-7} \cmidrule(lr){8-9}
Method & Cov & Width & Cov & Width & Cov & Width & Cov & Width \\
\midrule
\multicolumn{9}{l}{\textit{$n = 200$}} \\
\midrule
EIC-based & 86.6 & 0.426 & 73.2 & 0.390 & 84.0 & 0.412 & 69.8 & 0.362 \\
Targeted bootstrap & 94.8 & 0.542 & 95.2 & 0.642 & 96.9 & 0.565 & 95.8 & 0.688 \\
Robust variance & 92.8 & 0.580 & 88.8 & 0.785 & 94.6 & 0.587 & 90.2 & 0.809 \\
$V$-fold ($V{=}5$) & 96.8 & 0.753 & 92.2 & 1.091 & 97.4 & 0.831 & 94.4 & 1.262 \\
$V$-fold ($V{=}10$) & 96.8 & 0.726 & 93.8 & 1.025 & 98.0 & 0.754 & 95.2 & 1.138 \\
$V$-fold ($V{=}20$) & 98.2 & 0.750 & 94.8 & 1.107 & 98.4 & 0.786 & 96.0 & 1.253 \\
\midrule
\multicolumn{9}{l}{\textit{$n = 1000$}} \\
\midrule
EIC-based & 90.8 & 0.222 & 81.8 & 0.234 & 91.2 & 0.224 & 79.6 & 0.237 \\
Targeted bootstrap & 93.8 & 0.252 & 96.4 & 0.395 & 95.6 & 0.258 & 96.6 & 0.402 \\
Robust variance & 93.6 & 0.241 & 95.8 & 0.377 & 93.8 & 0.247 & 94.8 & 0.379 \\
$V$-fold ($V{=}5$) & 94.0 & 0.339 & 95.0 & 0.492 & 94.2 & 0.340 & 94.2 & 0.501 \\
$V$-fold ($V{=}10$) & 96.2 & 0.287 & 96.2 & 0.443 & 97.6 & 0.301 & 95.6 & 0.444 \\
$V$-fold ($V{=}20$) & 95.6 & 0.282 & 96.0 & 0.432 & 96.0 & 0.286 & 96.0 & 0.440 \\
\bottomrule
\end{tabular}
\end{table}

\Cref{tab:ate} summarizes coverage and CI width across all four scenarios. Under moderate positivity (Scenarios~A and~C), the EIC-based method achieves only 84--87\% coverage at $n=200$ because the plug-in variance estimator underestimates the true sampling variability when nuisance estimation is imprecise. Under severe positivity (Scenarios~B and~D), the EIC-based undercoverage worsens to 70--73\%. In contrast, the $V$-fold jackknife with $t_{V-1}$ critical values achieves 93--98\% coverage across all four scenarios at $n=200$, confirming the finite-sample validity of \cref{thm:fixedV}. The targeted bootstrap achieves near-nominal coverage uniformly, while the robust variance estimator undercovers slightly (89--93\%) in the severe positivity settings. At $n=1000$, the EIC-based method remains substantially below nominal under severe positivity (80--82\%), whereas all other methods converge toward the 95\% level.

Two patterns concerning interval width are worth highlighting. First, at the smaller sample size ($n=200$) the targeted bootstrap attains valid coverage with substantially narrower intervals than the $V$-fold jackknife (e.g., width $0.542$ versus $0.726$--$0.753$ in Scenario~A): with only $V$ leave-fold-out refits, the fixed-$V$ jackknife pays for its distribution-free validity through the wider $t_{V-1}$ critical values and a non-concentrating variance estimate, so it is more conservative than the bootstrap in small samples. Second, within the jackknife itself, increasing $V$ moves the intervals toward the targeted bootstrap width, because the $t_{V-1}$ quantiles shrink toward the normal quantiles as $V$ grows. Thus the price for avoiding the targeted bootstrap's $B$-fold refitting cost is wider intervals at small $n$ and small $V$; this gap narrows both as $n$ increases (by $n=1000$ the $V$-fold widths are much closer to the targeted bootstrap) and as $V$ increases at fixed $n$.

\cref{fig:ate_coverage} illustrates the coverage--width tradeoff as a function of~$V$. The top row shows that coverage remains above the nominal 95\% level for all $V \geq 2$ at $n = 200$ in the moderate-positivity scenarios, while under severe positivity the $V$-fold jackknife still substantially outperforms the EIC-based method. The bottom row confirms that CI width decreases with $V$ as the $t_{V-1}$ quantiles approach the normal quantiles. This tradeoff is consistent with the theory: for fixed $V$, the variance estimator does not concentrate, and the $t$-distribution accounts for this additional uncertainty. The reduction in width is largest for small $V$ and diminishes as $V$ grows: most of the gain comes from the rapid shrinkage of the $t_{V-1}$ quantiles at small degrees of freedom (e.g.\ from $V=5$ to $V=10$), whereas for larger $V$ the quantiles are already close to the normal value and offer little further narrowing. Beyond this point additional folds mainly reduce the number of observations used in each leave-fold-out refit, which does not tighten---and can slightly inflate---the variance estimate, so increasing $V$ past a moderate value yields diminishing returns in width while adding computational cost.

\begin{figure}[htbp]
\centering
\includegraphics[width=\textwidth]{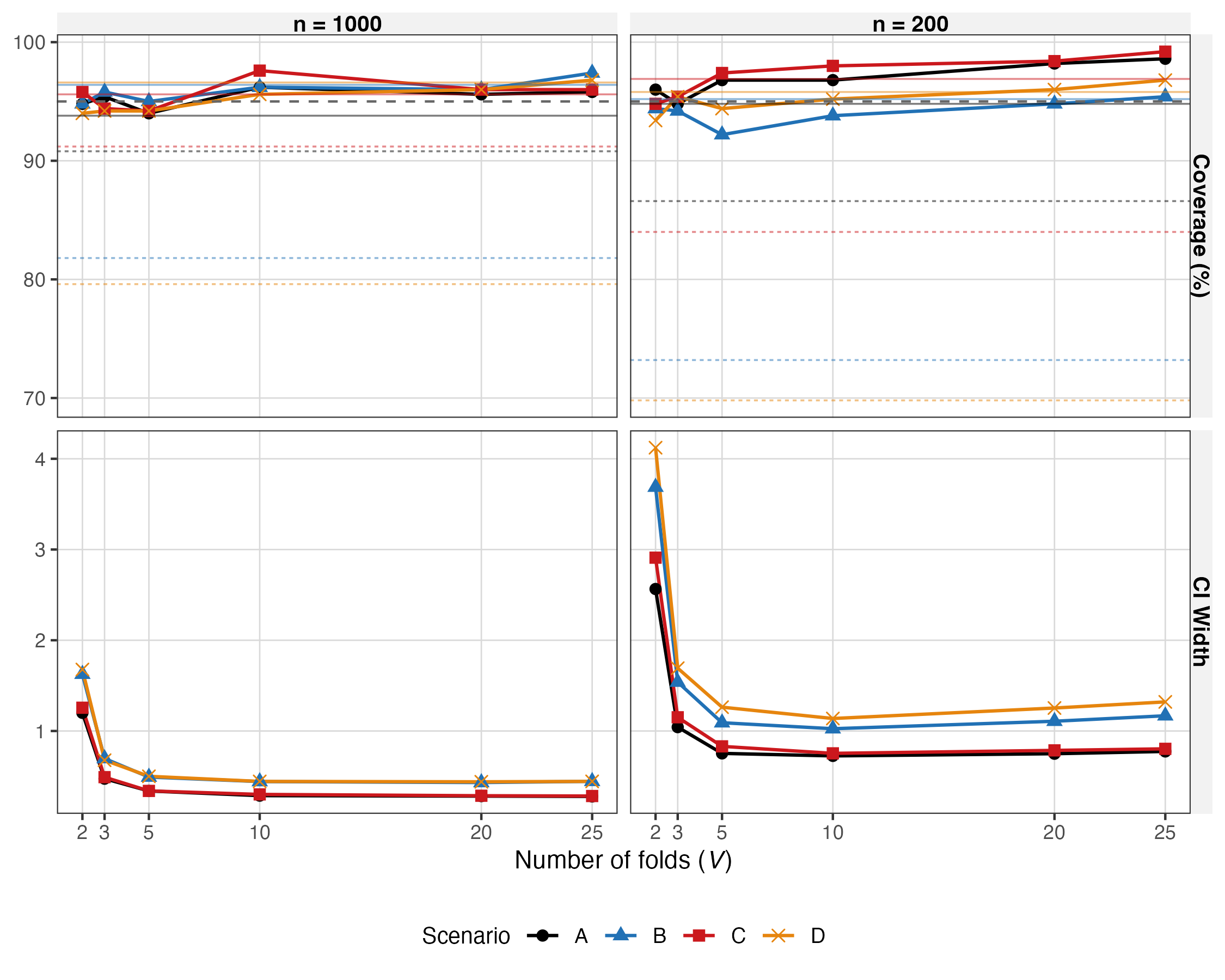}
\caption{Coverage (top) and CI width (bottom) of $V$-fold jackknife confidence intervals as a function of the number of folds $V$, for four scenarios and two sample sizes. Horizontal dashed line in the top panels: nominal 95\% level. Faded horizontal lines: coverage of the EIC-based and bootstrap methods for each scenario. The $V$-fold jackknife maintains coverage above the nominal level for moderate positivity and substantially outperforms the EIC-based method under severe positivity, with intervals narrowing as $V$ increases.}
\label{fig:ate_coverage}
\end{figure}

\subsection{Survival Curve: Kaplan--Meier Estimator}\label{sec:sim_km}

We next evaluate both pointwise and simultaneous confidence bands for the Kaplan--Meier estimator of the survival function $S(t) = P(T > t)$. This setting provides a natural testbed because the Greenwood formula yields an analytic standard error for comparison. We generate right-censored survival data from exponential distributions with event rate $\lambda = 0.1$ and censoring rate $\lambda_c = 0.05$, yielding approximately 33\% censoring. The sample size is $n = 200$, and we evaluate the survival function at $m = 91$ equally spaced time points on $[2, 20]$. Results are based on 500 replications.

\textbf{Pointwise inference.}
\cref{fig:km_pointwise} and \cref{tab:km_pw} compare pointwise coverage for the Greenwood formula with normal critical values, the $V$-fold jackknife with $t_{V-1}$ critical values ($V \in \{5, 10, 20\}$), and the bootstrap Wald method ($B = 1000$). All methods achieve approximately 95\% coverage, averaged over evaluation times, with the $V$-fold jackknife tracking the Greenwood formula closely. \cref{tab:km_pw} confirms that the $V$-fold jackknife variance estimator is well calibrated for the Kaplan--Meier estimator even at moderate sample sizes: the $V$-fold $t$ achieves 94.9\% ($V = 5$), 94.8\% ($V = 10$), and 94.7\% ($V = 20$), comparable to the Greenwood formula (94.6\%) and bootstrap (94.7\%).

\begin{figure}[htbp]
\centering
\includegraphics[width=0.7\textwidth]{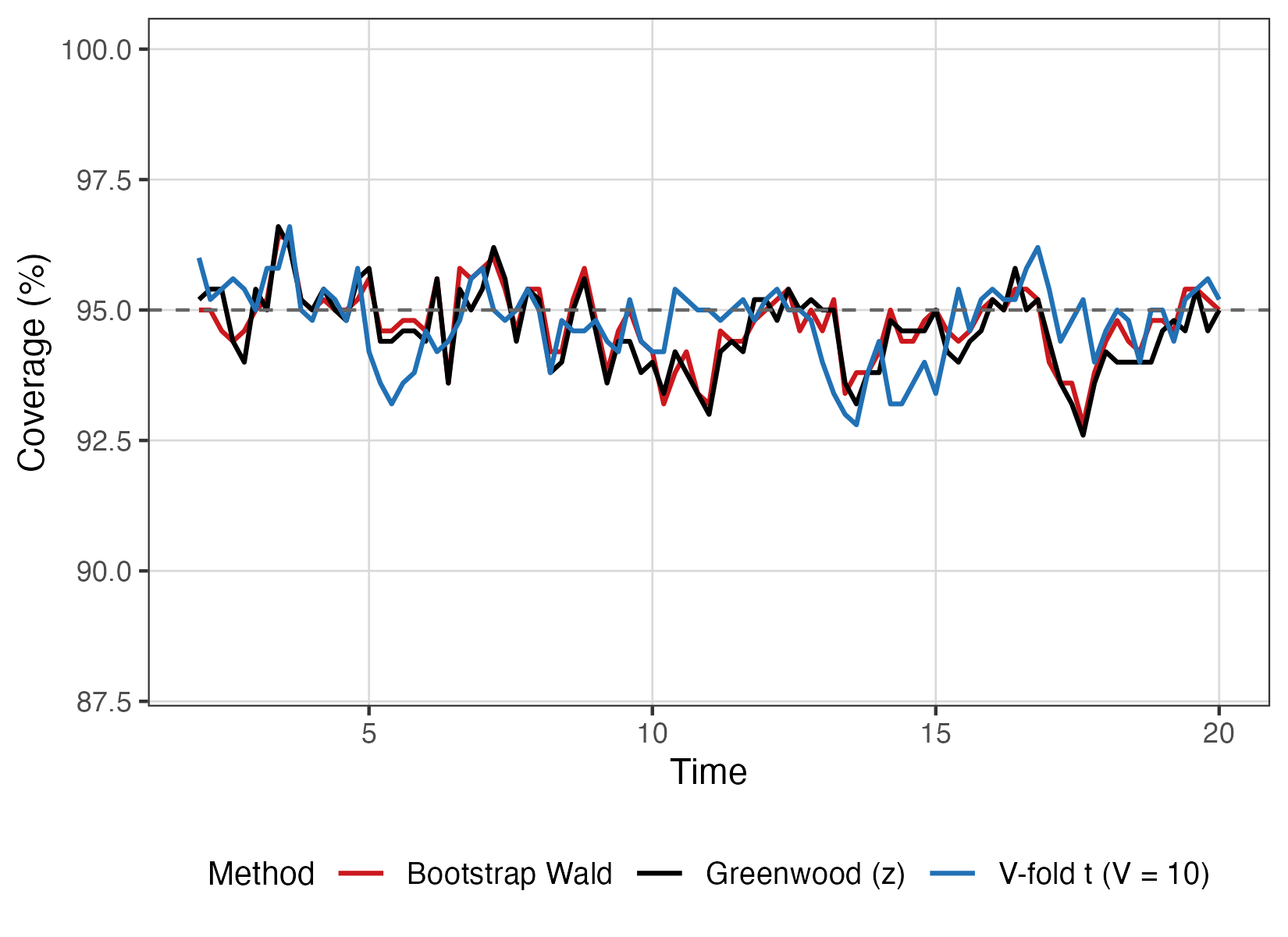}
\caption{Pointwise coverage of 95\% confidence intervals for the Kaplan--Meier survival function across evaluation times ($n = 200$, 500 replications). The $V$-fold jackknife with $t_9$ critical values ($V = 10$) closely tracks the Greenwood formula, and the bootstrap Wald method performs comparably.}
\label{fig:km_pointwise}
\end{figure}

\begin{table}[htbp]
\centering
\caption{Pointwise coverage (\%) and mean CI width for the Kaplan--Meier
survival function ($n=200$, $m=91$ evaluation points, 500 replications,
averaged over all evaluation times).}
\label{tab:km_pw}
\begin{tabular}{lcc}
\toprule
Method & Coverage (\%) & Mean Width \\
\midrule
Greenwood ($z$) & 94.6 & 0.145 \\
$V$-fold $t$ ($V{=}5$) & 94.9 & 0.197 \\
$V$-fold $t$ ($V{=}10$) & 94.8 & 0.164 \\
$V$-fold $t$ ($V{=}20$) & 94.7 & 0.153 \\
Bootstrap Wald & 94.7 & 0.145 \\
\bottomrule
\end{tabular}
\end{table}

\textbf{Simultaneous inference.}
\cref{tab:km} reports simultaneous coverage for six methods: the $V$-fold simultaneous band from \cref{thm:multi_t} (using the componentwise-Studentized critical value from \cref{sec:mc_critical}), the Greenwood normal band (Greenwood SE with multivariate normal critical values), the classical Hall--Wellner and Equal Precision (EP) bands, the studentized bootstrap band \citep{DavisonHinkley1997}, the bias-corrected $V$-fold band from \cref{thm:bias_corrected}, and a bias-corrected bootstrap (BC Bootstrap) band that widens the studentized bootstrap band by the absolute bootstrap bias $|\bar S^*_b(t) - \hat S(t)|$ at each evaluation point, where $\bar S^*_b(t) = B^{-1}\sum_{b=1}^{B}\hat S^*_b(t)$ is the mean of the bootstrap Kaplan--Meier estimates.

\begin{table}[htbp]
\centering
\caption{Simultaneous coverage (\%) and mean band width for the Kaplan--Meier survival function ($n=200$, $m=91$ evaluation points, 500 replications).
``BC'' denotes the bias-corrected band from \cref{sec:bias_corrected_bands}.}
\label{tab:km}
\begin{tabular}{lcc}
\toprule
Method & Simult.\ Cov.\ (\%) & Mean Width \\
\midrule
Equal Precision & 96.6 & 0.239 \\
Hall--Wellner & 96.6 & 0.251 \\
Bootstrap & 95.8 & 0.226 \\
$V$-fold simultaneous ($V{=}20$) & 94.6 & 0.239 \\
$V$-fold simultaneous ($V{=}5$) & 93.6 & 0.454 \\
$V$-fold simultaneous ($V{=}10$) & 93.4 & 0.282 \\
Greenwood Normal & 92.4 & 0.210 \\
\addlinespace[2pt]
BC simultaneous ($V{=}20$) & 94.8 & 0.240 \\
BC simultaneous ($V{=}5$) & 93.8 & 0.456 \\
BC simultaneous ($V{=}10$) & 93.4 & 0.283 \\
BC Bootstrap & 95.8 & 0.228 \\
\bottomrule
\end{tabular}
\end{table}

The $V$-fold simultaneous band achieves 93.4--94.6\% coverage across all three choices of $V$, approaching the classical EP (96.6\%) and Hall--Wellner (96.6\%) bands. For $V = 20$ the band achieves 94.6\% coverage with width 0.239, comparable to these established methods, while smaller $V$ produces wider but slightly less precise bands: $V = 5$ achieves 93.6\% coverage with mean width 0.454, and $V = 10$ achieves 93.4\% with width 0.282. Larger $V$ benefits from a more precise correlation matrix estimate $\hat R_{\mathrm{Jack}}$ (since $V$ is closer to the $m = 91$ evaluation points), which outweighs the effect of the lighter-tailed $t_{V-1}$ critical values. The net effect is that all three choices of $V$ achieve reasonable simultaneous coverage, with $V = 20$ offering the best coverage and narrowest bands. The bias-corrected bands yield coverage nearly identical to the standard bands (e.g., 94.8\% versus 94.6\% for $V = 20$) with negligible width increase (0.240 versus 0.239), reflecting the fact that the Kaplan--Meier estimator has minimal bias in this setting and the jackknife bias estimates $|\hat b_j|$ are effectively zero. Similarly, the BC Bootstrap band achieves the same 95.8\% coverage as the standard bootstrap with only marginally wider bands (0.228 versus 0.226), confirming that the bootstrap bias correction $|\bar S^*_b(t) - \hat S(t)|$ is also effectively zero for the Kaplan--Meier estimator.

\textbf{Effective rank of the correlation matrix.}
Since $m = 91$ exceeds $V - 1$ for all choices of $V$ considered here, the estimated correlation matrix $\hat R_{\mathrm{Jack}}$ is singular and the simultaneous critical value is computed via the eigendecomposition described in \cref{sec:eigen_validity}. To verify that the rank deficiency is benign, we estimate the true correlation matrix $R_0$ from the 500 simulation replications: the $500 \times 91$ matrix of Kaplan--Meier estimates across replications yields a sample correlation matrix whose eigendecomposition reveals the intrinsic dimensionality of the problem. The estimated effective rank is $\hat r(R_0) = m / \hat\lambda_{\max} = 91 / 53.8 \approx 1.7$, with the first two eigenvalues alone capturing 77\% of the total variance and only 9 eigenvalues needed to explain 95\% (\cref{fig:eigenvalue_decay}). This extremely low effective rank reflects the strong correlation among adjacent time points on a smooth survival curve: the Kaplan--Meier estimator at nearby time points is driven by essentially the same at-risk set, producing a nearly one-dimensional correlation structure. Even $V = 5$ folds far exceed the effective rank, confirming that the eigendecomposition accurately captures the dominant correlation structure of $R_0$.

\subsection{Dose-Response Curve: HAL Estimator}\label{sec:sim_hal}

Our final study demonstrates the $V$-fold jackknife for inference on a dose-response curve estimated by the highly adaptive lasso (HAL). Unlike the pathwise differentiable parameters in Sections~\ref{sec:sim_ate} and~\ref{sec:sim_km}, the dose-response curve is a non-pathwise differentiable functional whose HAL plug-in estimator converges at the slower rate $\sqrt{n/J_n}$, where $J_n$ is the number of non-zero HAL basis functions. The theoretical framework for this setting is developed in Section~\ref{sec:slower_rates}, where we show that the $V$-fold jackknife naturally adapts to this non-standard rate through the scale invariance of Studentization. Unlike the Kaplan--Meier estimator, deriving and implementing analytic standard errors for HAL estimates is non-trivial: the delta method requires computing the influence curve of the HAL-based plug-in estimator, which can fail numerically. The $V$-fold jackknife bypasses this difficulty entirely, requiring only the ability to refit the estimator on leave-fold-out subsamples.

We consider four data generating processes of increasing complexity from \citet{Shi2024}: DGP~1 (smooth, monotonically increasing dose-response), DGP~2 (oscillating), DGP~3 (piecewise smooth with a transition at $A = 2$), and DGP~4 (discontinuous). All have binary outcomes $Y$, continuous treatment $A \in [0,5]$, and a single confounder $W \sim N(0,1)$. The dose-response curve is evaluated at 20 equally spaced points on $[0, 5]$ with $n = 500$ and 500 replications. All HAL fits use first-order (piecewise linear) basis functions, corresponding to smoothness order~1 in the framework of \citet{Shi2024}, who also consider zero-order and adaptive choices. We consider two HAL fitting strategies: cross-validated HAL (CV-HAL) and undersmoothed HAL ($\lambda = \lambda_{\text{CV}} / \sqrt{\log n}$).

We compare three methods: the delta method using the HAL influence curve with normal critical values, the $V$-fold jackknife with $t_{V-1}$ critical values ($V = 20$), and the bootstrap Wald method ($B = 200$).

\textbf{Pointwise inference.}
\cref{tab:hal_pw} reports pointwise coverage, mean CI width, and the delta method failure rate (percentage of replications producing NA standard errors) across all four DGPs. For DGP~1 (smooth), all methods achieve near-nominal coverage (93--95\%), with the $V$-fold jackknife at 93.7\% (CV-HAL) and 95.0\% (undersmoothed). For DGP~2 (oscillating), the delta method's coverage drops sharply to 77.8\% (CV-HAL) and 72.2\% (undersmoothed), with failure rates of 0.8\% and 6.2\%, respectively. In contrast, the $V$-fold jackknife achieves 94.3\% and 95.1\%, demonstrating its robustness when the delta method is unreliable. For DGPs~3 (piecewise) and~4 (discontinuous), coverage is lower across all methods due to estimator bias, but the $V$-fold jackknife consistently provides the highest coverage: for undersmoothed HAL, 89.8\% and 90.0\% compared to 81.8\% and 85.5\% for the delta method, and 85.4\% and 86.3\% for the bootstrap.

\begin{table}[htbp]
\centering
\caption{Pointwise coverage (\%), mean CI width, and delta method failure rate
for HAL dose-response curve estimation ($n=500$, 500 replications, averaged
over 20 evaluation points).}
\label{tab:hal_pw}
\small
\begin{tabular}{ll ccc ccc}
\toprule
 & & \multicolumn{3}{c}{CV-HAL} & \multicolumn{3}{c}{Undersmoothed HAL} \\
\cmidrule(lr){3-5} \cmidrule(lr){6-8}
DGP & Method & Cov & Width & NA\% & Cov & Width & NA\% \\
\midrule
1 (smooth) & Delta & 93.1 & 0.129 & 0.0 & 93.6 & 0.171 & 0.4 \\
 & $V$-fold $t$ ($V{=}20$) & 93.7 & 0.137 & 0.0 & 95.0 & 0.185 & 0.0 \\
 & Bootstrap & 92.8 & 0.126 & 0.0 & 93.6 & 0.161 & 0.0 \\
\addlinespace
2 (oscillating) & Delta & 77.8 & 0.105 & 0.8 & 72.2 & 0.122 & 6.2 \\
 & $V$-fold $t$ ($V{=}20$) & 94.3 & 0.177 & 0.0 & 95.1 & 0.214 & 0.0 \\
 & Bootstrap & 90.5 & 0.138 & 0.0 & 92.8 & 0.161 & 0.0 \\
\addlinespace
3 (piecewise) & Delta & 75.2 & 0.161 & 0.0 & 81.8 & 0.208 & 1.4 \\
 & $V$-fold $t$ ($V{=}20$) & 82.2 & 0.211 & 0.0 & 89.8 & 0.288 & 0.0 \\
 & Bootstrap & 77.5 & 0.163 & 0.0 & 85.4 & 0.207 & 0.0 \\
\addlinespace
4 (discontinuous) & Delta & 76.0 & 0.219 & 0.0 & 85.5 & 0.279 & 0.0 \\
 & $V$-fold $t$ ($V{=}20$) & 80.2 & 0.255 & 0.0 & 90.0 & 0.343 & 0.0 \\
 & Bootstrap & 75.3 & 0.209 & 0.0 & 86.3 & 0.267 & 0.0 \\
\bottomrule
\end{tabular}
\end{table}

\cref{fig:hal_pointwise} shows pointwise coverage across the dose grid for all four DGPs. For DGP~1 with undersmoothed HAL, all methods maintain coverage near 95\% uniformly over the evaluation points. For DGP~2, the delta method exhibits severe undercoverage throughout the dose range, while the $V$-fold jackknife remains near nominal. For DGPs~3 and~4, coverage deteriorates at specific dose levels where the true dose-response curve is non-smooth: for DGP~3, coverage dips appear for $A > 2$ (where oscillations begin), and for DGP~4, coverage dips concentrate at $A \approx 2$ and $A \approx 4$ (the discontinuity locations). As the bias diagnosis in \cref{fig:hal_bias_diagnosis} confirms, these are precisely the points where the absolute bias of the first-order HAL estimator exceeds the estimated standard error, indicating that coverage loss reflects estimator bias rather than variance estimation failure. In all settings, the $V$-fold jackknife provides the highest coverage, reflecting the conservatism of the $t_{V-1}$ critical values.

\begin{figure}[htbp]
\centering
\includegraphics[width=\textwidth,height=0.78\textheight,keepaspectratio]{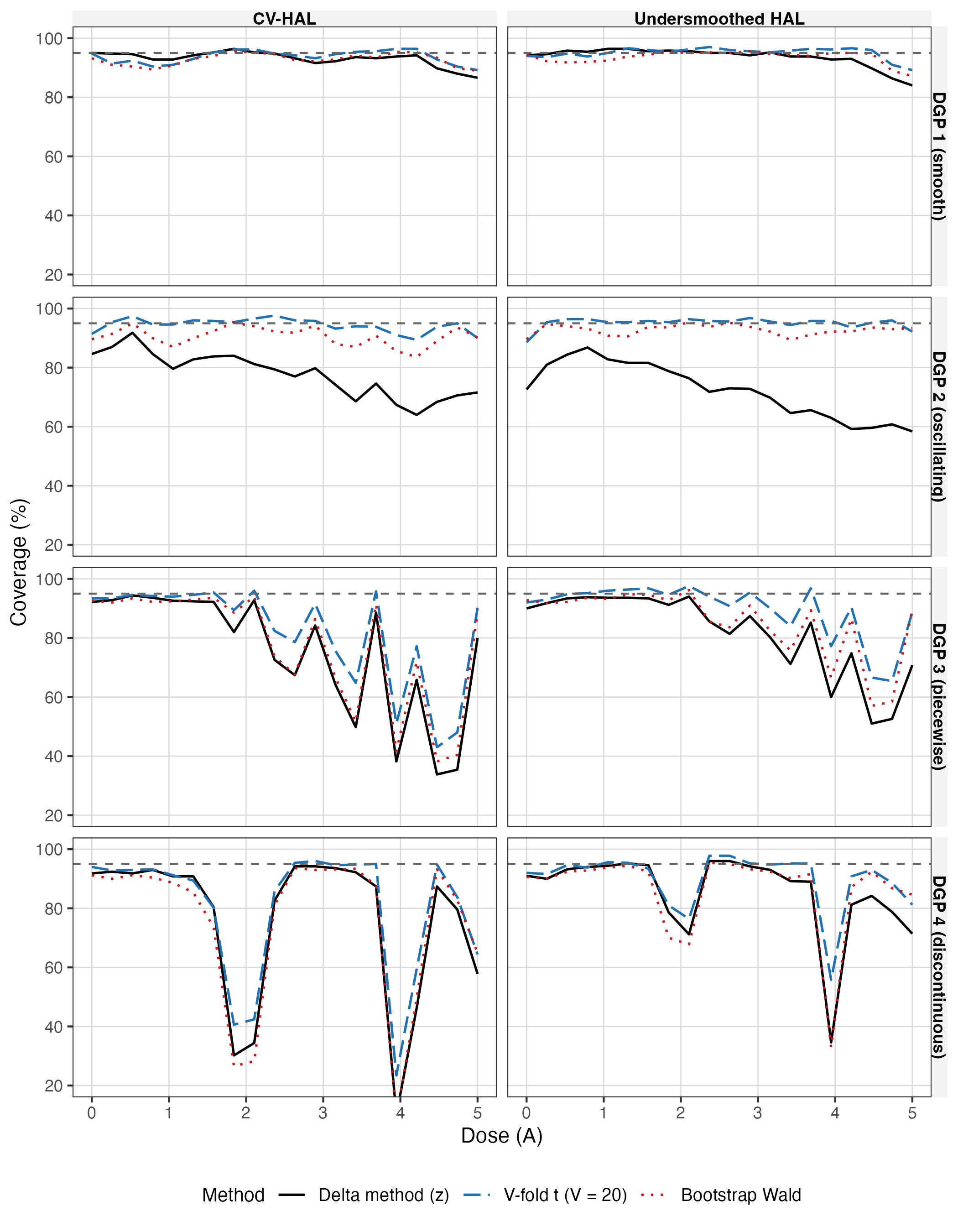}
\caption{Pointwise coverage of 95\% confidence intervals across the dose grid ($20$ equally spaced points on $[0,5]$, i.e.\ $A \in \{0, 0.263, 0.526, \dots, 5\}$) for HAL dose-response curve estimation ($n = 500$, 500 replications). Each point on the horizontal axis represents a dose level at which the dose-response curve $\E[Y \mid A = a, W]$ is evaluated and coverage is computed separately. Rows: DGP~1 (smooth), DGP~2 (oscillating), DGP~3 (piecewise), DGP~4 (discontinuous). Columns: CV-HAL and undersmoothed HAL. The $V$-fold jackknife with $t_{19}$ critical values ($V = 20$) provides the highest coverage across all DGPs. The delta method fails most notably for DGP~2 (oscillating), where its failure rate also exceeds 5\% for undersmoothed HAL.}
\label{fig:hal_pointwise}
\end{figure}

To diagnose the source of coverage loss in DGPs~3 and~4, \cref{fig:hal_bias_diagnosis} displays the absolute bias $|\bar{\hat\psi}(a) - \psi_0(a)|$ (averaged over replications) alongside the mean estimated standard error at each dose point, using the $V$-fold jackknife ($V = 20$). When the absolute bias exceeds the standard error, the asymptotic linearity assumption underlying the confidence interval is stressed, and coverage loss is expected. For DGP~1 (smooth), the bias is uniformly small relative to the standard error across the entire dose grid, consistent with the near-nominal coverage in \cref{fig:hal_pointwise}. For DGP~2 (oscillating), the bias is similarly controlled, which explains why the $V$-fold jackknife achieves excellent coverage despite the complex curve shape---the first-order HAL can approximate the oscillating function adequately at this sample size. For DGP~3 (piecewise), the bias is small for $A \le 2$ (where the true curve is smooth and nearly linear) but rises sharply for $A > 2$ (where the oscillating pattern begins), producing a bias-to-SE ratio that exceeds~1 at several dose points, particularly for CV-HAL. These are precisely the dose points where \cref{fig:hal_pointwise} shows coverage dips. For DGP~4 (discontinuous), the bias spikes at $A \approx 2$ and $A \approx 4$---the locations of the jump discontinuities---where the piecewise-linear HAL basis cannot represent the true function, again matching the coverage dip locations in \cref{fig:hal_pointwise}. These patterns confirm that the coverage loss in DGPs~3 and~4 is attributable to estimator bias at non-smooth points, not to a failure of the $V$-fold jackknife variance estimator.

\begin{figure}[htbp]
\centering
\includegraphics[width=\textwidth,height=0.78\textheight,keepaspectratio]{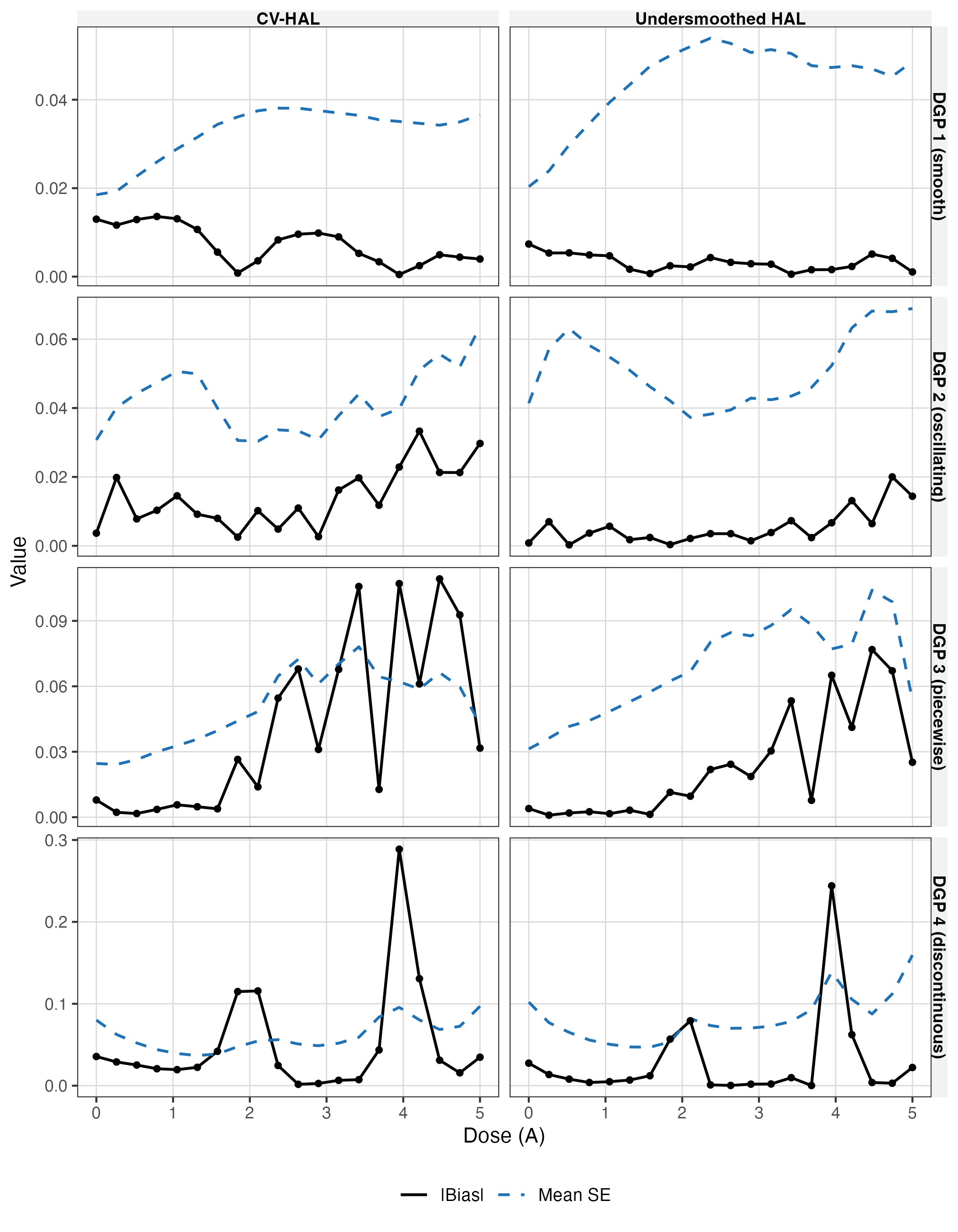}
\caption{Bias diagnosis for HAL dose-response curve estimation ($n = 500$, 500 replications). Solid black line: absolute bias $|\bar{\hat\psi}(a) - \psi_0(a)|$ at each of the 20 evaluation points, where $\bar{\hat\psi}(a)$ is the simulation average of the point estimate. Dashed blue line: mean estimated standard error from the $V$-fold jackknife ($V = 20$). When the bias exceeds the standard error, the asymptotic linearity condition is violated and coverage loss is expected. The spatial pattern of bias peaks corresponds to the coverage dips in \cref{fig:hal_pointwise}: for DGP~3, bias rises for $A > 2$ where oscillations begin; for DGP~4, bias spikes at $A = 2$ and $A = 4$ where the true curve is discontinuous.}
\label{fig:hal_bias_diagnosis}
\end{figure}

\textbf{Simultaneous inference.}
\cref{tab:hal_simult} reports simultaneous coverage for the delta method simultaneous band (multivariate normal with estimated IC covariance), the $V$-fold simultaneous band with $V \in \{5, 10, 20\}$ (componentwise-Studentized critical value from \cref{thm:multi_t}), the studentized bootstrap band \citep{DavisonHinkley1997}, the bias-corrected bootstrap band, and the bias-corrected $V$-fold band from \cref{thm:bias_corrected}.

\begin{table}[htbp]
\centering
\caption{Simultaneous coverage (\%) and mean band width for HAL dose-response
curve estimation ($n=500$, 500 replications, 20 evaluation points).
``BC'' denotes the bias-corrected band from \cref{sec:bias_corrected_bands}.}
\label{tab:hal_simult}
\small
\begin{tabular}{ll cc cc}
\toprule
 & & \multicolumn{2}{c}{CV-HAL} & \multicolumn{2}{c}{Undersmoothed HAL} \\
\cmidrule(lr){3-4} \cmidrule(lr){5-6}
DGP & Method & Cov & Width & Cov & Width \\
\midrule
1 (smooth) & Delta simult. & 84.8 & 0.179 & 82.4 & 0.249 \\
 & $V$-fold simult.\ ($V{=}5$) & 91.4 & 0.338 & 91.0 & 0.479 \\
 & $V$-fold simult.\ ($V{=}10$) & 91.2 & 0.230 & 88.8 & 0.325 \\
 & $V$-fold simult.\ ($V{=}20$) & 89.8 & 0.199 & 88.8 & 0.288 \\
 & Bootstrap simult. & 90.0 & 0.186 & 88.4 & 0.256 \\
\addlinespace[2pt]
 & BC simult.\ ($V{=}5$) & 94.8 & 0.355 & 94.6 & 0.513 \\
 & BC simult.\ ($V{=}10$) & 96.0 & 0.255 & 92.8 & 0.372 \\
 & BC simult.\ ($V{=}20$) & 94.2 & 0.239 & 95.0 & 0.357 \\
 & BC Bootstrap simult. & 92.4 & 0.194 & 90.0 & 0.268 \\
\addlinespace
2 (oscillating) & Delta simult. & 27.6 & 0.156 & 19.2 & 0.181 \\
 & $V$-fold simult.\ ($V{=}5$) & 90.0 & 0.447 & 92.2 & 0.541 \\
 & $V$-fold simult.\ ($V{=}10$) & 86.8 & 0.311 & 88.0 & 0.382 \\
 & $V$-fold simult.\ ($V{=}20$) & 87.8 & 0.281 & 89.0 & 0.351 \\
 & Bootstrap simult. & 85.6 & 0.230 & 89.4 & 0.281 \\
\addlinespace[2pt]
 & BC simult.\ ($V{=}5$) & 94.2 & 0.497 & 95.8 & 0.603 \\
 & BC simult.\ ($V{=}10$) & 94.8 & 0.380 & 91.8 & 0.470 \\
 & BC simult.\ ($V{=}20$) & 96.6 & 0.393 & 95.4 & 0.485 \\
 & BC Bootstrap simult. & 87.6 & 0.247 & 91.6 & 0.299 \\
\addlinespace
3 (piecewise) & Delta simult. & 17.4 & 0.235 & 26.2 & 0.313 \\
 & $V$-fold simult.\ ($V{=}5$) & 64.2 & 0.519 & 79.2 & 0.712 \\
 & $V$-fold simult.\ ($V{=}10$) & 47.2 & 0.366 & 65.4 & 0.500 \\
 & $V$-fold simult.\ ($V{=}20$) & 38.4 & 0.331 & 59.8 & 0.462 \\
 & Bootstrap simult. & 30.0 & 0.261 & 56.4 & 0.349 \\
\addlinespace[2pt]
 & BC simult.\ ($V{=}5$) & 70.2 & 0.553 & 80.0 & 0.738 \\
 & BC simult.\ ($V{=}10$) & 53.4 & 0.430 & 72.0 & 0.583 \\
 & BC simult.\ ($V{=}20$) & 51.2 & 0.446 & 72.2 & 0.593 \\
 & BC Bootstrap simult. & 33.4 & 0.278 & 60.8 & 0.372 \\
\addlinespace
4 (discontinuous) & Delta simult. & 10.8 & 0.324 & 31.6 & 0.423 \\
 & $V$-fold simult.\ ($V{=}5$) & 52.2 & 0.656 & 77.6 & 0.883 \\
 & $V$-fold simult.\ ($V{=}10$) & 30.6 & 0.455 & 65.0 & 0.616 \\
 & $V$-fold simult.\ ($V{=}20$) & 26.4 & 0.406 & 59.6 & 0.556 \\
 & Bootstrap simult. & 16.0 & 0.342 & 56.8 & 0.474 \\
\addlinespace[2pt]
 & BC simult.\ ($V{=}5$) & 64.4 & 0.714 & 84.8 & 0.977 \\
 & BC simult.\ ($V{=}10$) & 49.6 & 0.552 & 78.6 & 0.741 \\
 & BC simult.\ ($V{=}20$) & 45.4 & 0.549 & 76.8 & 0.736 \\
 & BC Bootstrap simult. & 21.0 & 0.363 & 63.4 & 0.503 \\
\bottomrule
\end{tabular}
\end{table}

For DGP~1 (smooth), the standard $V$-fold simultaneous band achieves 89.8--91.4\% coverage (CV-HAL) and 88.8--91.0\% (undersmoothed), the highest among the uncorrected methods. For DGP~2 (oscillating), the advantage is striking: the delta method achieves only 27.6\% (CV-HAL) and 19.2\% (undersmoothed), while the $V$-fold band reaches 87.8--90.0\% and 88.0--92.2\%. For DGPs~3 and~4, all uncorrected methods exhibit substantial undercoverage due to estimator bias at non-smooth points, but the $V$-fold simultaneous band consistently provides the best coverage. With undersmoothed HAL on DGP~4, the $V$-fold band ($V = 5$) reaches 77.6\% compared to 31.6\% for the delta method and 56.8\% for the bootstrap.

The bias-corrected bands from \cref{thm:bias_corrected} provide a substantial improvement. For DGPs~1 and~2, where the estimator bias is moderate, the bias-corrected $V$-fold band achieves near-nominal coverage: 94.2--96.0\% (CV-HAL) and 92.8--95.0\% (undersmoothed) for DGP~1, and 94.2--96.6\% and 91.8--95.8\% for DGP~2. These results validate \cref{thm:bias_corrected} in finite samples. For DGPs~3 and~4, the bias correction substantially improves coverage relative to the standard band---for example, on DGP~4 with undersmoothed HAL the bias-corrected band ($V = 5$) reaches 84.8\% versus 77.6\% without correction---but coverage remains below nominal because the estimator bias is too large for the jackknife bias estimate to fully compensate. The bias-corrected bands are wider than the standard bands (e.g., for DGP~2 with $V = 20$: 0.393 versus 0.281 for CV-HAL), reflecting the price of the $|\hat b_j|$ correction term.

The bias-corrected bootstrap band, which widens the standard bootstrap band by $|\hat b_{\mathrm{boot},j}|$ where $\hat b_{\mathrm{boot},j} = \bar\psi_{\mathrm{boot}}(a_j) - \hat\psi(a_j)$, similarly improves over the standard bootstrap by 2--7 percentage points across all DGPs. For DGPs~1 and~2, it reaches 87.6--92.4\% (CV-HAL) and 90.0--91.6\% (undersmoothed). For DGPs~3 and~4, the improvement is more modest (e.g., 21.0\% versus 16.0\% for CV-HAL on DGP~4). However, the bias-corrected bootstrap consistently achieves lower coverage than the bias-corrected $V$-fold band, indicating that the jackknife bias estimate is more effective than the bootstrap bias estimate for bias correction in this setting.

These results illustrate three complementary points. First, the $V$-fold jackknife provides valid variance estimation and the most reliable inference among the methods compared, regardless of the HAL tuning strategy. The delta method can fail both numerically (producing NAs at rates up to 6.2\%) and statistically (coverage as low as 19.2\% for simultaneous bands), whereas the $V$-fold jackknife never fails and consistently achieves the highest coverage. Second, the bias-corrected simultaneous band from \cref{sec:bias_corrected_bands} recovers near-nominal coverage when the estimator bias is moderate (DGPs~1--2), confirming the theoretical guarantee of \cref{thm:bias_corrected}; the bias-corrected bootstrap similarly improves over the standard bootstrap, though the jackknife bias estimate proves more effective for bias correction than the bootstrap bias estimate. Third, achieving nominal simultaneous coverage ultimately requires the estimator's bias to be negligible relative to its standard error---a condition that depends on the estimator and data generating process rather than the variance estimation method. The $V$-fold jackknife correctly quantifies sampling variability; the remaining coverage gap in non-smooth settings is attributable to estimator bias, which the bias correction partially but not fully addresses when the bias is large. The spatial pattern of the coverage dips in \cref{fig:hal_simult_pointwise} mirrors the bias profile in \cref{fig:hal_bias_diagnosis}: the simultaneous coverage is driven down by the same high-bias dose points identified in the pointwise analysis.

\textbf{Effective rank of the correlation matrix.}
For the simultaneous bands with $V = 5$ or $V = 10$, we have $m = 20 > V - 1$, so $\hat R_{\mathrm{Jack}}$ is rank-deficient and the eigendecomposition of \cref{sec:eigen_validity} is used. We estimate the effective rank $r(R_0)$ from the 500 simulation replications, as for the KM setting. For CV-HAL, the effective rank ranges from 2.8 (DGP~1, smooth) to 3.8 (DGP~3, piecewise); for undersmoothed HAL, it ranges from 4.1 (DGP~2, oscillating) to 5.3 (DGP~3 and~4). The higher effective rank for undersmoothed HAL reflects the reduced smoothing: adjacent curve points are less correlated, distributing variance across more eigendirections. \textcolor{black}{For CV-HAL, the effective rank remains strictly below $V-1=4$ when $V=5$. For undersmoothed HAL, it slightly exceeds the available degrees of freedom for $V=5$, but remains drastically lower than the ambient dimension $m=20$. This explains why $V=5$ can still perform reasonably well, while also highlighting why larger $V$ improves simultaneous coverage.} \cref{fig:eigenvalue_decay} displays the eigenvalue spectrum and cumulative variance explained for all settings, showing rapid eigenvalue decay in every case.

\begin{figure}[htbp]
\centering
\includegraphics[width=\textwidth]{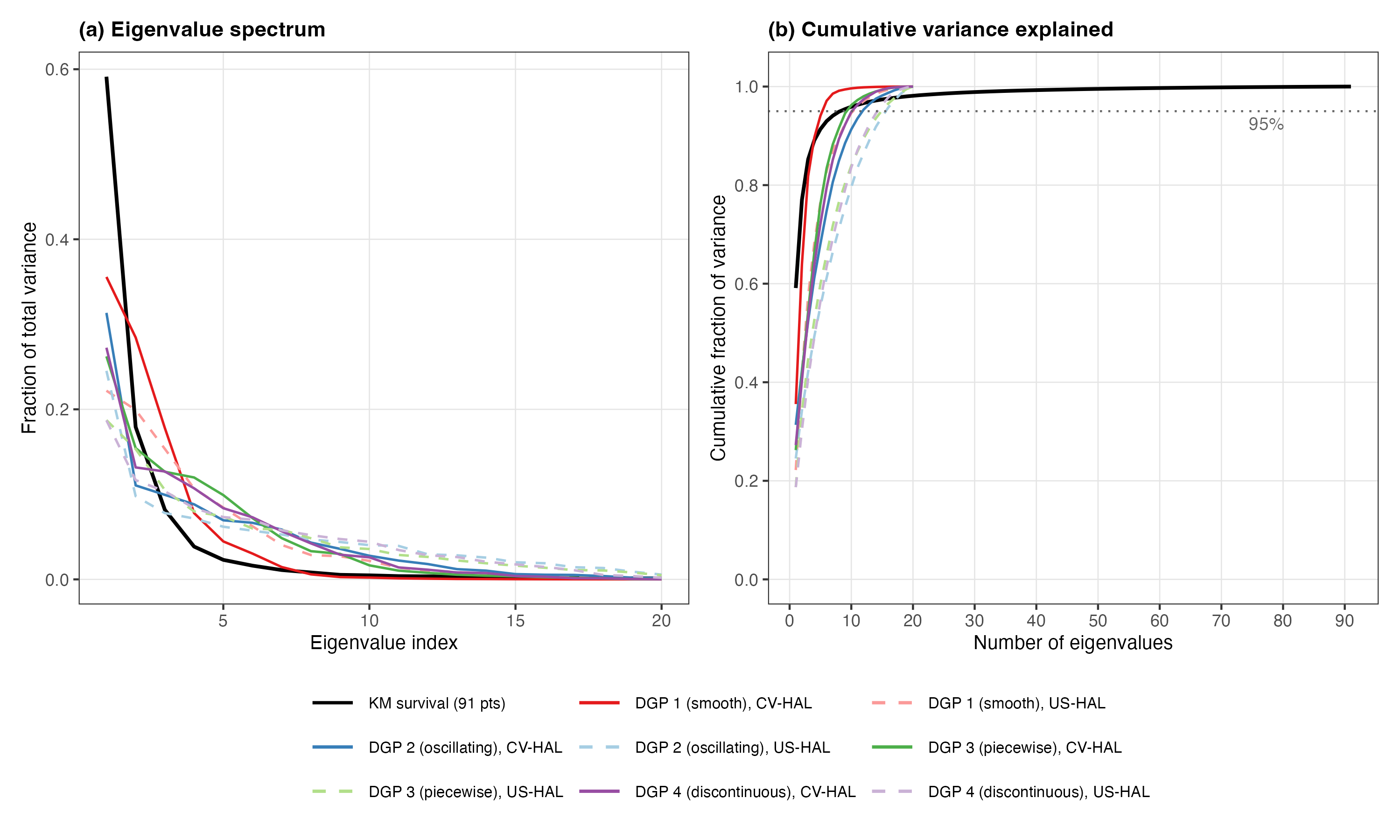}
\caption{Eigenvalue spectrum of the estimated correlation matrix $R_0$ for the Kaplan--Meier survival curve ($m = 91$ time points) and the HAL dose-response curve ($m = 20$ evaluation points, four DGPs, CV-HAL and undersmoothed variants). The correlation matrix is estimated from the $500 \times m$ matrix of point estimates across simulation replications. (a)~Normalized eigenvalue spectrum (fraction of trace $m$), truncated at index~20 for visibility. The KM curve exhibits the steepest decay, with a single dominant eigenvalue accounting for 59\% of the total variance. (b)~Cumulative fraction of variance explained. For the KM curve, two eigenvalues suffice for 77\% and nine for 95\%; for the HAL settings, 6--16 eigenvalues are needed for 95\%. In all cases, the effective rank $r(R_0) = m / \lambda_{\max}$ is much smaller than $m$, confirming that the eigendecomposition of the rank-deficient $\hat R_{\mathrm{Jack}}$ accurately captures the dominant correlation structure even when $V \ll m$.}
\label{fig:eigenvalue_decay}
\end{figure}

\section{Discussion}

We have studied the grouped ($V$-fold) jackknife as a practical inference tool for regular asymptotically linear estimators of pathwise differentiable parameters. Building on the classical theory of \citet{tukey1958bias}, \citet{Brillinger1964}, and \citet{Shao1989}, we extended the fixed-$V$ $t_{V-1}$ Studentization result to the RAL class (\cref{thm:fixedV}) and established diverging-$V$ consistency at rate $\sqrt{V}$ under explicit remainder conditions (\cref{thm:divergingV}), \textcolor{black}{subject to the fold-level remainder-difference condition, with Appendix~\ref{app:V_rate} giving sufficient upper-growth bounds on $V$.} The principal new theoretical contributions are twofold: the extension to simultaneous confidence bands for vector-valued parameters (\cref{thm:multi_t}), where the correct componentwise-Studentized limit distribution conditional on the true correlation matrix $R_0$---distinct from the standard multivariate $t$-distribution---is used to compute critical values via Monte Carlo simulation; and the generalization to estimators with non-standard convergence rates (\cref{thm:fixedV_slower,thm:divergingV_slower}), where the scale invariance of Studentization ensures that the $t_{V-1}$ result holds even when the influence function variance diverges. Bias-corrected (\cref{thm:bias_corrected}) variants were introduced to address finite-sample bias, and a post-hoc diagnostic for flagging unreliable grid points is provided in Appendix~\ref{sec:app_diagnostic}.

The simulation studies support both the formal conclusions and the practical heuristics developed in the paper. For scalar inference on the average treatment effect, the $V$-fold jackknife with $t_{V-1}$ critical values achieves strong coverage even at $n=200$, whereas the plug-in influence-function-based variance estimator undercovers substantially. For the Kaplan--Meier survival function, the grouped jackknife yields pointwise and simultaneous bands that are practically comparable to established methods. For HAL dose-response curve estimation, where the delta method can fail numerically and can be severely anti-conservative, the $V$-fold jackknife never fails and consistently delivers the strongest empirical coverage among the methods compared. When the estimator bias is moderate, the bias-corrected simultaneous band restores near-nominal coverage.

\subsection{Practical guidance on choosing $V$}

The choice of $V$ involves a trade-off between computational cost and statistical precision. Fewer folds require fewer refits but produce wider confidence intervals, because the $t_{V-1}$ quantiles exceed their normal counterparts; more folds narrow the intervals as the $t$ quantiles approach the normal quantiles and the variance estimator concentrates. For scalar inference, the simulations suggest that $V \in \{10, 20\}$ offers a practical balance: for the ATE at $n = 1000$, increasing $V$ from $5$ to $20$ reduces the mean CI width by roughly 15\% while maintaining coverage above the nominal level. For simultaneous confidence bands, the more relevant issue is the quality of the estimated correlation matrix $\hat R_{Jack}$. As discussed in \cref{sec:eigen_validity}, a practical approach is to increase $V$ along a sequence such as $\{5, 10, 20, 40, \dots\}$ until it satisfies $V \ge 2 \hat r_V$; absent such tuning, the simulations suggest that $V = 20$ is a reasonable default in the settings studied. When $V < m$, the estimated correlation matrix is singular and the Monte Carlo draws are confined to a rank-$(V-1)$ subspace of $\R^m$, but this confinement is exact for the plug-in distribution and is often benign when the dominant eigendirections are already captured. As a simple fallback rule, $V = 10$ to $20$ remains reasonable for scalar and moderate-dimensional problems, with larger $V$ preferred when $m$ is large or when the computational budget allows it. The total cost of the procedure is $V + 1$ fits of the estimator (one full-sample fit and $V$ leave-fold-out fits), compared to thousands for the bootstrap.

\subsection{Scope and limitations}

The theoretical guarantees rest on the estimator admitting an asymptotically linear expansion with a pathwise differentiable influence function. This includes a broad class of semiparametric estimators---TMLE, one-step estimators, augmented IPW estimators---and the generalized asymptotic linearity framework of Section~\ref{sec:slower_rates} extends the scope beyond $\sqrt{n}$-rate estimators. Parameters fall outside the present theory when the required asymptotic linearity or fold-stable remainder conditions fail, so genuinely non-smooth cases require separate analysis. Efron's sample-median example shows that the ordinary delete-$1$ jackknife can fail in such settings, but that example should be read as a cautionary boundary case rather than as a blanket exclusion of all quantile-type parameters. The key insight is that the Studentized jackknife statistic is scale-invariant: $\sigma_n$ cancels between numerator and denominator, so the $t_{V-1}$ inference remains valid regardless of whether the influence function variance is bounded or diverges with the number of basis functions. This allows practitioners to obtain valid confidence intervals without knowing $J_n$---in contrast to the delta method, which requires computing these quantities explicitly. The simultaneous confidence band theory assumes fixed $m$; extending to a growing number of evaluation points $m = m_n \to \infty$ would require uniform-in-$m$ control of the remainder terms and a different limiting process (e.g., a Gaussian process approximation), which is beyond the scope of this paper.

For simultaneous confidence bands, the fixed-$m$ theory identifies the correct componentwise-Studentized limit distribution for fixed $V$ conditional on the true correlation matrix $R_0$, while consistency of the correlation estimate requires $V \to \infty$. \textcolor{black}{Formal plug-in simultaneous coverage therefore also requires $V \to \infty$ while the fold-level remainder-difference condition remains valid; under the sufficient conditions in Appendix~\ref{app:V_rate}, slow choices such as $V=O(\log n)$ are covered.} The strong performance observed for small fixed $V$ should thus be interpreted as empirical support, informed by the low-effective-rank argument, rather than as a formal fixed-$V$ coverage theorem.

The bias-corrected bands from \cref{thm:bias_corrected} improve coverage when the estimator's bias is moderate relative to its standard error, as demonstrated for DGPs~1 and~2 in the HAL simulations. However, when the bias is large---as in the piecewise and discontinuous dose-response settings (DGPs~3 and~4)---the bias correction provides partial but insufficient improvement. The bias diagnosis in \cref{fig:hal_bias_diagnosis} makes the mechanism precise: coverage loss occurs at dose points where the absolute bias exceeds the standard error, and these points coincide with features of the true curve (oscillations for $A > 2$ in DGP~3, discontinuities at $A = 2$ and $A = 4$ in DGP~4) that the first-order HAL's piecewise-linear basis cannot well approximate. This is a limitation of the estimator's approximation class, not of the variance estimation method. Achieving nominal coverage in such settings requires reducing the estimator's bias---through higher-order HAL basis functions, adaptive smoothness selection as in \citet{Shi2024}, or other bias-reduction techniques; the $V$-fold jackknife correctly quantifies sampling variability but cannot compensate for systematic bias in the point estimator.

Finally, the results are conditional on a single random partition of the data into $V$ folds. While the theory establishes asymptotic validity for any partition satisfying the balanced-fold condition, finite-sample performance may exhibit partition-to-partition variability, particularly when $V$ is small. Averaging the variance estimator over multiple independent partitions could reduce this variability at additional computational cost.

\subsection{Connections and extensions}

The $V$-fold jackknife shares a natural connection with cross-fitting procedures used in modern semiparametric inference. In debiased machine learning \citep{Chernozhukov2018} and cross-validated TMLE \citep{zheng2011cross}, the nuisance functions are already trained on each leave-fold-out subsample $P_{n,-v}$, and the target parameter is evaluated on the held-out fold~$I_v$. The $V$-fold jackknife requires the complementary quantity $\hat\Psi(P_{n,-v})$---the estimator evaluated on the $V-1$ training folds rather than on the single held-out fold. Since the expensive nuisance estimation step has already been performed, computing $\hat\Psi(P_{n,-v})$ requires only an additional targeting or evaluation pass on the training folds using the existing nuisance estimates, avoiding the cost of a full refit. The $V$-fold jackknife thus integrates naturally into cross-fitting workflows at modest additional computational cost.

The $V$-fold jackknife should also be distinguished from the jackknife+ and CV+ methods of \citet{BarberCandesRamdasTibshirani2021}, which use fold-based leave-one-out residuals to construct distribution-free prediction intervals for a new observation. While both methods partition the data into folds, the inferential targets are fundamentally different: the jackknife+ provides marginal coverage for a future observation under exchangeability, whereas the $V$-fold jackknife developed here provides asymptotic coverage for fixed population parameters.

The extension to non-$\sqrt{n}$ rates developed in Section~\ref{sec:slower_rates} opens several directions. The generalized asymptotic linearity framework applies to any sieve MLE whose influence-curve-like representation has growing variance, including not only HAL but also kernel and series estimators. The simultaneous confidence bands, combined with the generalized theory, provide a principled approach to uniform inference for HAL functional estimators.

Several further extensions merit investigation. First, adapting the framework to dependent data---for example, time series or clustered observations---would require replacing the i.i.d.\ fold structure with block or cluster-based partitions, analogous to block bootstrap methods. Second, allowing $m$ to grow with $n$ in the simultaneous band construction would connect the $V$-fold jackknife to nonparametric function estimation and strong approximation theory. Third, an adaptive procedure that selects $V$ based on a data-driven criterion---balancing the bias-variance properties of the variance estimator against computational cost---could further improve practical performance. We leave these directions for future work.

\subsection{Concluding remarks}

The grouped jackknife has a long history in statistics, but its properties for the semiparametric and machine learning--based estimators used in modern practice have not been systematically examined. This paper fills that gap, demonstrating that the $V$-fold jackknife is a simple, reliable tool for variance estimation and inference in settings where analytic standard errors are unavailable, difficult to derive, or numerically unreliable. Its primary appeal is computational: only $V$ refits are required, compared to the thousands demanded by the bootstrap. The fixed-$V$ $t_{V-1}$ inference is especially attractive for expensive estimators, where even $V = 5$ or $10$ folds yield valid confidence intervals. For simultaneous confidence bands, the key formal message is that the fixed-$V$ theory identifies the correct componentwise-Studentized limit conditional on $R_0$, whereas \textcolor{black}{the plug-in procedure becomes formally valid once $V \to \infty$ while the fold-level remainder-difference condition remains valid; under the sufficient bounds in Appendix~\ref{app:V_rate}, slow choices such as $V=O(\log n)$ are covered.} In practice, low-effective-rank correlation structure helps explain why smaller values of $V$ can still perform well. The further extension to estimators with non-standard convergence rates demonstrates that the same jackknife procedure yields valid inference even when the influence function variance diverges, with the scale-invariant Studentization eliminating the need to know or estimate the effective convergence rate. We anticipate that the $V$-fold jackknife will be particularly useful in conjunction with machine learning--based semiparametric estimators, where it offers a computationally feasible path to valid statistical inference.

\section*{Code availability}

All code required to reproduce the simulation studies of
Section~\ref{sec:simulations}, together with the simulation output and the
scripts that generate every figure and table in this paper, is available at
\url{https://github.com/yiberkeley/Vfold_Jackknife_paper}. The repository contains the three
simulation studies---the average treatment effect via TMLE, the Kaplan--Meier
survival curve, and the HAL dose-response curve---each with its own
documentation, and a single script that regenerates all figures and tables from
the stored output in a few minutes.

\bibliographystyle{biometrika}
\bibliography{jack}

\appendix

\section{Proofs}\label{app:proofs}

This appendix collects the proofs of the numbered lemmas, theorems, and corollaries from the main text, grouped by section.

\subsection*{Proofs for Section~\ref{sec:properties} (Properties of V-fold Jackknife)}

\begin{proof}[Proof of Lemma~\ref{lem:Sequiv}]
Let $W_v=P^1_{n,v}\varphi$ and $\bar W = V^{-1} \sum_{v=1}^V W_v$. By substituting the asymptotically linear expansion into the definition of $IC_{Jack}(v)$, we observe an exact algebraic cancellation of the full-sample remainder $R(P_n, P_0)$ when we center the pseudo-values. Specifically, we can write:
    \begin{align*}
        IC_{Jack}(v)-\hat\Psi_{Jack}(P_n) &= (W_v - \bar W) - (V-1) \left\{R(P_{n,-v}, P_0) - \frac{1}{V} \sum_{v'=1}^V R(P_{n,-v'}, P_0) \right\} \\
        &= (W_v - \bar W) + \Delta_v,
    \end{align*}
where $\Delta_v = -(V-1) \left\{R(P_{n,-v}, P_0) - \bar{R}_{n,-v} \right\}$, with $\bar{R}_{n,-v} = \frac{1}{V}\sum_{v'=1}^V R(P_{n,-v'}, P_0)$. Under the stated assumptions, $\max_v|\Delta_v|=o_p(n^{-1/2})$ for fixed $V$. Expanding squares yields
\[
\hat S_{Jack}^2-S_W^2
=\frac{1}{V-1}\sum_{v=1}^V\Big\{\Delta_v^2+2(W_v-\bar W)\Delta_v\Big\}.
\]
Since $W_v-\bar W=O_p(n^{-1/2})$ and $\max_v|\Delta_v|=o_p(n^{-1/2})$,
with $V$ fixed we obtain
\[
\hat S_{Jack}^2-S_W^2
= o_p(n^{-1})+ O_p(n^{-1/2})\,o_p(n^{-1/2})
=o_p(n^{-1}).
\]
{\color{black} By the central limit theorem applied to the independent fold means, $(\sqrt{n/V}\,W_1,\dots,\sqrt{n/V}\,W_V)\CiD (Z_1,\dots,Z_V)$ with independent $Z_v\sim\Normal(0,\sigma^2)$. Therefore, by the continuous mapping theorem,
\[
\frac{n}{V}S_W^2 \CiD \frac{1}{V-1}\sum_{v=1}^V (Z_v-\bar Z)^2
=\sigma^2\frac{\chi^2_{V-1}}{V-1}.
\]}
Moreover $S_W^2\asymp_p V/n$ (with $V$ fixed), so
\[
\frac{\hat S_{Jack}^2-S_W^2}{S_W^2}=o_p(1),
\]
which gives $\hat S_{Jack}^2=S_W^2\{1+o_p(1)\}$.
\end{proof}

\begin{proof}[Proof of Theorem~\ref{thm:fixedV}]
     Using the linear representation of the estimator, the numerator of the standardized score satisfies
\begin{align*}
    \sqrt{V}\,(\hat\Psi_{Jack}(P_n)-\Psi(P_0)) = \sqrt{V}\,\bar W + \sqrt{V}\,\Big(V R(P_n, P_0) - (V-1) \bar{R}_{n,-v}\Big).
\end{align*}
The standard error term $\hat S_{Jack}$ is of order $\sqrt{V/n}$ (that is, $\frac{n}{V}\hat S^2_{Jack} = O_p(1)$). Hence, multiplying the numerator and denominator by $\sqrt{n/V}$ yields a remainder component in the numerator equal to
\[
\sqrt{n}\,\Big(V R(P_{n},P_0) - (V-1)\bar{R}_{n,-v}\Big).
\]
Since $V$ is fixed, the assumptions $R(P_{n},P_0)=o_p(n^{-1/2})$ and $\max_v |R(P_{n,-v}, P_0)| = o_p(n^{-1/2})$ imply this remainder term is $o_p(1)$.
Moreover, using the result of Lemma \ref{lem:Sequiv}, we have $\hat S_{Jack} = S_W\{1+o_p(1)\}$.
 Hence,
\[
\frac{\sqrt{V}\,(\hat\Psi_{Jack}(P_n)-\Psi(P_0))}{\hat S_{Jack}} = \frac{\sqrt{V}\,\bar{W}}{S_{W}}+o_p(1).
\]
Since the folds are disjoint and $n/V\to\infty$, the vector $(\sqrt{n/V}\,W_1,\dots,\sqrt{n/V}\,W_V)$ converges in distribution to a multivariate normal with independent identically distributed components centered at 0. The classical result for Studentized normal samples then yields that the left-hand side converges in distribution to the $t_{V-1}$ distribution as $n \to \infty$ with $V$ fixed. Therefore,
\[
\frac{\sqrt{V}\,(\hat\Psi_{Jack}(P_n)-\Psi(P_0))}{\hat S_{Jack}} \CiD \ t_{V-1}.
\]
\end{proof}

\begin{proof}[Proof of Lemma~\ref{lem:centering}]
By definition, $\hat\Psi_{Jack}(P_n) = V\hat\Psi(P_n) - (V-1)\bar\Psi_{(-v)}$, so
\[
\hat\Psi_{Jack}(P_n) - \hat\Psi(P_n) = (V-1)\big(\hat\Psi(P_n) - \bar\Psi_{(-v)}\big) = -\hat b.
\]
Substituting the asymptotic linear expansions $\hat\Psi(P_n) - \Psi(P_0) = P_n\varphi + R(P_n,P_0)$ and $\hat\Psi(P_{n,-v}) - \Psi(P_0) = P_{n,-v}\varphi + R(P_{n,-v},P_0)$, the influence function contributions cancel exactly: since each observation $i$ belongs to exactly one fold, $V^{-1}\sum_{v=1}^V P_{n,-v}\varphi = P_n\varphi$. Hence the difference is driven entirely by the remainders:
\[
-\hat b = (V-1)\left\{R(P_n,P_0) - \frac{1}{V}\sum_{v=1}^V R(P_{n,-v},P_0)\right\}.
\]
To connect with the remainder structure from the proof of \cref{lem:Sequiv}, define $\delta_v = (V-1)\big(R(P_n,P_0) - R(P_{n,-v},P_0)\big)$. Then $-\hat b = V^{-1}\sum_{v=1}^V \delta_v = \bar\delta$, while the centered remainders from \cref{lem:Sequiv} satisfy $\Delta_v = \delta_v - \bar\delta$. Under the assumptions of \cref{thm:fixedV}, $R(P_n,P_0) = o_p(n^{-1/2})$ and $\max_v |R(P_{n,-v},P_0)| = o_p(n^{-1/2})$, so each $|\delta_v| = o_p(n^{-1/2})$. Since $V$ is fixed, $|\hat b| = |\bar\delta| \le \max_v |\delta_v| = o_p(n^{-1/2}) = o_p(\widehat{SE})$.
\end{proof}

\begin{proof}[Proof of Corollary~\ref{cor:CI}]
Write
\[
\frac{\sqrt{V}\,\big(\hat\Psi(P_n)-\Psi(P_0)\big)}{\hat S_{Jack}}
=
\frac{\sqrt{V}\,(\hat\Psi_{Jack}(P_n)-\Psi(P_0))}{\hat S_{Jack}}
+
\frac{\sqrt{V}\,\big(\hat\Psi(P_n)-\hat\Psi_{Jack}(P_n)\big)}{\hat S_{Jack}}.
\]
The first term converges in distribution to $t_{V-1}$ by \cref{thm:fixedV}. The second term equals $\sqrt{V}\,\hat b / \hat S_{Jack} = \hat b / \widehat{SE} = o_p(1)$ by \cref{lem:centering}. The conclusion follows from Slutsky's theorem.
\end{proof}

\begin{proof}[Proof of Theorem~\ref{thm:divergingV}]
    We first prove the consistency of our estimator $\hat \sigma^2$. Let $W_v=P^1_{n,v}\varphi$ and $\bar W = V^{-1} \sum_{v=1}^V W_v$. Following the algebraic exactness from the fixed-$V$ case, we can write:
    \begin{align*}
        IC_{Jack}(v)-\hat\Psi_{Jack}(P_n) = (W_v - \bar W) + \Delta_v,
    \end{align*}
    where $\Delta_v = -(V-1) \left\{ R(P_{n,-v}, P_0) - \bar{R}_{n,-v} \right\}$, with $\bar{R}_{n,-v} = \frac{1}{V}\sum_{v'=1}^V R(P_{n,-v'}, P_0)$. Notice that we can rewrite this by subtracting and adding the full data remainder $R(P_n, P_0):$
    \[
    \Delta_v = (V-1)\left(d_{n,v}-\bar d_n\right),
    \qquad
    \bar d_n=\frac{1}{V}\sum_{v=1}^V d_{n,v}.
    \]
    Therefore, the sum of squared remainders is bounded by the uncentered sum of squares:
    \[
    \sum_{v=1}^V \Delta_v^2
    \le (V-1)^2 \sum_{v=1}^V d_{n,v}^2
    =
    (V-1)^2 V\|d_n\|_{V,2}^2.
    \]
    Expanding the square in the variance estimator, we have
    \[
    \hat S_{Jack}^2 = \hat S_W^2 + \frac{1}{V-1}\sum_{v=1}^V \Delta_v^2 + \frac{2}{V-1}\sum_{v=1}^V (W_v - \bar W)\Delta_v.
    \]
    Multiplying by $n/V$, we analyze the remainder terms. Using $\|d_n\|_{V,2}=o_p(n^{-1/2}V^{-1})$, we obtain

    \begin{align}
        \frac{n}{V(V-1)} \sum_{v=1}^V \Delta_v^2
        &\le
        n(V-1)\|d_n\|_{V,2}^2
        =
        o_p(V^{-1}). \label{eq:rem1}
    \end{align}

    For the cross-term, we apply the Cauchy--Schwarz inequality:

    \begin{align}
        \left| \frac{n}{V(V-1)} \sum_{v=1}^V (W_v - \bar W ) \Delta_v \right| &\le \left( \frac{n}{V} \hat S_W^2 \right)^{1/2} \left( \frac{n}{V(V-1)} \sum_{v=1}^V \Delta_v^2 \right)^{1/2} \nonumber \\
        &= O_p(1) \cdot o_p(V^{-1/2}) = o_p(V^{-1/2}). \label{eq:rem2}
    \end{align}

    This establishes that $\hat\sigma^2 = \frac{n}{V} \hat S_W^2 + o_p(V^{-1/2})$.

    Now we establish an exact decomposition for $\frac{n}{V} \hat S^2_{W}$ to show it consistently estimates $\sigma^2$. We can write the sample variance of the fold averages exactly as:
    \[
    \frac{n}{V}\hat S_W^2 = \frac{n}{V(V-1)}\sum_{v=1}^V W_v^2 - \frac{n}{V-1}\bar W^2.
    \]
    Substituting $W_v = \frac{V}{n}\sum_{i \in I_v} \varphi_i$ and $\bar W = \frac{1}{n}\sum_{i=1}^n \varphi_i$, expanding the squares yields:
    \begin{align*}
        W_v^2 &= \frac{V^2}{n^2} \left( \sum_{i \in I_v} \varphi_i^2 + \sum_{\substack{i,j \in I_v \\ i\neq j}} \varphi_i \varphi_j \right), \\
        \bar W^2 &= \frac{1}{n^2} \left( \sum_{i=1}^n \varphi_i^2 + \sum_{i\neq j} \varphi_i \varphi_j \right).
    \end{align*}
    Grouping the diagonal $\varphi_i^2$ terms from both parts, their coefficient becomes:
    \[
    \frac{n}{V(V-1)} \frac{V^2}{n^2} - \frac{n}{V-1} \frac{1}{n^2} = \frac{V}{n(V-1)} - \frac{1}{n(V-1)} = \frac{V-1}{n(V-1)} = \frac{1}{n}.
    \]
    Thus, the exact decomposition is:
    \begin{align}
        \frac{n}{V} \hat S_W^2 &= \frac{1}{n}\sum_{i=1}^n \varphi_i^2 + \frac{V}{n(V-1)}\sum_{v=1}^V \sum_{\substack{i,j \in I_v \\ i\neq j}} \varphi_i \varphi_j - \frac{1}{n(V-1)} \sum_{i\neq j} \varphi_i \varphi_j \nonumber \\
        &= \frac{1}{n}\sum_{i=1}^n \varphi_i^2 + T_n + \frac{1}{V-1} T_n - \frac{1}{V-1} U_n, \label{eq:varest}
    \end{align}
    where $T_n = \frac{1}{n}\sum_{v=1}^V \sum_{i \neq j \in I_v} \varphi_i \varphi_j$ and $U_n = \frac{1}{n}\sum_{i\neq j} \varphi_i \varphi_j$.

    We know $\E \left(\frac{1}{n} \sum_{i=1}^n \varphi_i^2 \right)=\sigma^2$ and $U_n = O_p(1)$ as a standard U-statistic. Thus, to show the consistency of our estimator we must bound the variance of $T_n$. The term $T_n$ constitutes sums of mean-zero uncorrelated random variables (for $i \neq j$). We analyze its variance:
    \begin{align*}
        \Var(T_n) &= \frac{1}{n^2} \sum_{v=1}^V \Var\left( \sum_{\substack{i,j \in I_v \\ i\neq j}} \varphi_i\varphi_j \right) = \frac{1}{n^2} \sum_{v=1}^V \sum_{\substack{i\neq j \in I_v \\ k\neq \ell \in I_v}} \E(\varphi_i\varphi_j\varphi_k\varphi_\ell).
    \end{align*}
    The expectation $\E(\varphi_i\varphi_j\varphi_k\varphi_\ell)$ is non-zero only when $(k,\ell)=(i,j) $ or $(k,\ell)=(j,i)$. \textcolor{black}{Using $\E_{P_0}(\varphi_i^2)=\sigma^2$,} we have
    {\color{black}
    \begin{align*}
        \Var(T_n) &= \frac{1}{n^2} \sum_{v=1}^V 2\Big(\frac{n}{V}\Big)\Big(\frac{n}{V}-1\Big)\sigma^4 \sim \frac{1}{n^2} V \cdot 2 \frac{n^2}{V^2} \sigma^4 = \frac{2\sigma^4}{V} = O(V^{-1}).
    \end{align*}
    }
    So $T_n = O_p(V^{-1/2})$. The trailing terms in (\ref{eq:varest}) are therefore $\frac{1}{V-1}T_n = O_p(V^{-3/2})$ and $\frac{1}{V-1}U_n = O_p(V^{-1})$.
    Combining everything, we obtain
    \[
    \hat\sigma^2 =\frac{1}{n} \sum_{i=1}^n \varphi_i^2 + T_n + o_p(V^{-1/2}).
    \]
    Consequently, consistency follows as $V \to \infty$.

    \textcolor{black}{If $\E_{P_0}[\varphi^4]<\infty$, then the central limit theorem applied to $\varphi_i^2$ gives $\frac{1}{n}\sum_{i=1}^n \varphi_i^2 - \sigma^2 = O_p(n^{-1/2})$. Since $V=o(n)$ by assumption, this diagonal empirical fluctuation is $o_p(V^{-1/2})$, while $T_n=O_p(V^{-1/2})$. Therefore,}
    \[
    \sqrt V\Big(\hat\sigma^2-\sigma^2\Big)=O_p(1),
    \]
    and hence the proposed $V$-fold jackknife estimator is root-$V$ consistent.
\end{proof}

\begin{proof}[Proof of Corollary~\ref{cor:divergingV_wald}]
\begingroup\color{black}
The expansion gives $\sqrt n\{\hat\Psi(P_n)-\Psi(P_0)\}=\sqrt n\,P_n\varphi+o_p(1)$, while \cref{thm:divergingV} gives $\hat\sigma^2\CiP\sigma^2$. Slutsky's theorem yields the Studentized normal limit, and the coverage statement follows from convergence of the critical values.
\endgroup
\end{proof}

\subsection*{Proofs for Section~\ref{sec:simult} (Simultaneous V-fold Jackknife Confidence Bands)}

\begin{proof}[Proof of Theorem~\ref{thm:multi_consistency}]
Let $W_v^{(j)} = P^1_{n,v} \varphi_j$ and $\bar W^{(j)} = V^{-1} \sum_{v=1}^V W_v^{(j)}$. Using the definition $IC_{Jack}^{(j)}(v) = V\hat\Psi_j(P_n) - (V-1)\hat\Psi_j(P_{n,-v})$ and substituting the asymptotic linear expansion, the $V$-fold jackknife pseudo-value centered at its mean satisfies an exact algebraic identity:
\[
IC_{Jack}^{(j)}(v) - \hat\Psi_{Jack}^{(j)}(P_n) = (W_v^{(j)} - \bar W^{(j)}) + \Delta_v^{(j)},
\]
where $\Delta_v^{(j)} = -(V-1) \big\{ R_j(P_{n,-v}, P_0) - \frac{1}{V} \sum_{v'=1}^V R_j(P_{n,-v'}, P_0) \big\}$. Because $V$ is fixed and $\max_v |R_j(P_{n,-v}, P_0)| = o_p(n^{-1/2})$, we have $\max_{1 \le v \le V} |\Delta_v^{(j)}| = o_p(n^{-1/2})$.

Expanding the cross-product $(IC_{Jack}^{(j)}(v) - \hat\Psi_{Jack}^{(j)}(P_n))(IC_{Jack}^{(k)}(v) - \hat\Psi_{Jack}^{(k)}(P_n))$ yields:
\begin{align*}
    (W_v^{(j)} - \bar W^{(j)})(W_v^{(k)} - \bar W^{(k)}) + (W_v^{(j)} - \bar W^{(j)}) \Delta_v^{(k)} + \Delta_v^{(j)}(W_v^{(k)} - \bar W^{(k)}) + \Delta_v^{(j)} \Delta_v^{(k)}.
\end{align*}
Summing over $v$ and dividing by $V-1$, we get
\[
\hat\Sigma_{Jack}(j,k) = \hat S_W^2(j,k) + A_{jk} + A_{kj} + B_{jk},
\]
where $A_{jk} = \frac{1}{V-1} \sum_{v=1}^V (W_v^{(j)} - \bar W^{(j)}) \Delta_v^{(k)}$ and $B_{jk} = \frac{1}{V-1} \sum_{v=1}^V \Delta_v^{(j)} \Delta_v^{(k)}$.

Since $\max_v |\Delta_v^{(j)}| = o_p(n^{-1/2})$, it is immediate that $|B_{jk}| \le \max_v |\Delta_v^{(j)}| \max_v |\Delta_v^{(k)}| = o_p(n^{-1})$. For the cross term, applying the Cauchy--Schwarz inequality gives
\[
|A_{jk}| \le \left[ \frac{1}{V-1} \sum_{v=1}^V \big(W_v^{(j)} - \bar W^{(j)}\big)^2 \right]^{1/2} \left[ \frac{1}{V-1} \sum_{v=1}^V \big(\Delta_v^{(k)}\big)^2 \right]^{1/2}.
\]
\textcolor{black}{The first factor is $\hat S_W(j,j)$, which is $O_p(n^{-1/2})$ since $\Var(W_v^{(j)}) = \frac{V}{n}\sigma_j^2$ and $V$ fixed.} The second factor is bounded by $\max_v |\Delta_v^{(k)}| = o_p(n^{-1/2})$. Thus, $|A_{jk}| \le O_p(n^{-1/2}) \cdot o_p(n^{-1/2}) = o_p(n^{-1})$.

Combining these bounds yields $\hat\Sigma_{Jack}(j,k) = \hat S_W^2(j,k) + o_p(n^{-1})$, which implies $\frac{n}{V}\hat\Sigma_{Jack}(j,k) = \frac{n}{V}\hat S_W^2(j,k) + o_p(1)$.

Finally, if the folds are of exactly equal size $n/V$, the fold-level averages $W_v^{(j)}$ are independent and identically distributed with covariance $\Cov(W_v^{(j)}, W_v^{(k)}) = \frac{V}{n} \sigma_{jk}$. It follows directly that $\E[\hat S_W^2(j,k)] = \frac{V}{n} \sigma_{jk}$.
\end{proof}

\begin{proof}[Proof of Theorem~\ref{thm:multi_t}]
For each fold $v = 1, \dots, V$, define the $m$-dimensional vector $\bW_v = (P^1_{n,v} \varphi_1, \dots, P^1_{n,v} \varphi_m)^\transpose$. Since the $V$ folds are disjoint, the vectors $\bW_1, \dots, \bW_V$ are independent conditional on the partition. Each fold contains $n/V \to \infty$ i.i.d.\ observations, so by the multivariate central limit theorem,
\[
\sqrt{n/V} \, (\bW_1, \dots, \bW_V) \CiD (\tilde\bZ_1, \dots, \tilde\bZ_V),
\]
where $\tilde\bZ_1, \dots, \tilde\bZ_V$ are independent draws from $\Normal_m(\bm{0}, \bSigma)$.

By the same argument as in the proof of \cref{thm:fixedV} applied componentwise, $\hat\Psi_{Jack}^{(j)}(P_n) - \Psi_j(P_0) = \bar W^{(j)} + V R_j(P_n, P_0) - (V-1)\bar{R}_{j,n,-v}$, where $\bar{R}_{j,n,-v} = V^{-1}\sum_{v'} R_j(P_{n,-v'}, P_0)$. Because $\max_v |R_j(P_{n,-v}, P_0)| = o_p(n^{-1/2})$ and $R_j(P_n, P_0) = o_p(n^{-1/2})$, the remainder is $o_p(n^{-1/2})$, so $\hat\Psi_{Jack}^{(j)}(P_n) - \Psi_j(P_0) = \bar W^{(j)} + o_p(n^{-1/2})$. From Theorem~\ref{thm:multi_consistency}, we have $\hat\sigma_j^2 = \hat S_W^2(j,j) + o_p(n^{-1})$. Multiplying the numerator and denominator by $\sqrt{n/V}$, the $j$-th component of $\bT$ satisfies
\[
T_j = \frac{\sqrt{V}(\hat\Psi_{Jack}^{(j)}(P_n) - \Psi_j(P_0))}{\hat\sigma_j} = \frac{\sqrt{n}\,\bar W^{(j)} + o_p(1)}{\sqrt{\frac{n}{V}\hat S_W^2(j,j) + o_p(1)}}.
\]
By Slutsky's theorem and the continuous mapping theorem, the vector of componentwise-Studentized means $\big(\frac{\sqrt{n}\bar W^{(1)}}{\sqrt{\frac{n}{V}\hat S_W^2(1,1)}}, \dots, \frac{\sqrt{n}\bar W^{(m)}}{\sqrt{\frac{n}{V}\hat S_W^2(m,m)}}\big)^\transpose$ converges in distribution to the corresponding vector constructed from $\tilde\bZ_1, \dots, \tilde\bZ_V$. Because the scale factors $\sigma_j$ cancel out exactly in the componentwise ratio, this limiting vector is equivalent in distribution to $\bT^\infty$ constructed from the correlation matrix $R_0$. Thus, $\bT \CiD \bT^\infty$. Because $Z_1^{(j)}, \dots, Z_V^{(j)}$ are independent Gaussian random variables for each $j$, each marginal limit $T^\infty_j$ admits an exact $t_{V-1}$ distribution.
\end{proof}

\begin{proof}[Proof of Theorem~\ref{thm:plugin_simult_divergingV}]
\begingroup\color{black}
For each component, the centered pseudo-value decomposition is
\[
IC_{Jack}^{(j)}(v)-\hat\Psi_{Jack}^{(j)}(P_n)
=
\big(W_v^{(j)}-\bar W^{(j)}\big)+\Delta_v^{(j)},
\qquad
\Delta_v^{(j)}=(V-1)\{d_{n,v}^{(j)}-\bar d_{n}^{(j)}\}.
\]
The assumed fold RMS bound gives
\[
\frac{n}{V(V-1)}\sum_{v=1}^V\big(\Delta_v^{(j)}\big)^2=o_p(V^{-1})
\]
for every $j$. Cauchy--Schwarz then shows that all covariance terms involving at least one $\Delta_v^{(j)}$ are $o_p(1)$ after multiplication by $n/V$. Hence it suffices to analyze the oracle fold covariance $\hat S_W^2(j,k)$.

Expanding the oracle covariance exactly as in \cref{thm:divergingV} gives
\[
\frac{n}{V}\hat S_W^2(j,k)
=
\frac{1}{n}\sum_{i=1}^n \varphi_{j,i}\varphi_{k,i}
+T_{n,jk}
+\frac{1}{V-1}T_{n,jk}
-\frac{1}{V-1}U_{n,jk},
\]
where $T_{n,jk}$ is the within-fold off-diagonal cross-product term and $U_{n,jk}$ is the corresponding full-sample off-diagonal term. The finite fourth-moment assumption implies the weak law for $n^{-1}\sum_i\varphi_{j,i}\varphi_{k,i}$, and the same pairing calculation used for $T_n$ yields $T_{n,jk}=O_p(V^{-1/2})$; also $U_{n,jk}=O_p(1)$. Therefore $(n/V)\hat\Sigma_{Jack}(j,k)\CiP\sigma_{jk}$ for every fixed pair $(j,k)$. Since $m$ is fixed and all marginal variances are bounded away from zero, diagonal rescaling gives $\hat R_{Jack}\CiP R_0$ entrywise.

The multivariate central limit theorem and the full-sample remainder condition imply
\[
\left(
\frac{\hat\Psi_1(P_n)-\Psi_1(P_0)}{\widehat{SE}_1},
\dots,
\frac{\hat\Psi_m(P_n)-\Psi_m(P_0)}{\widehat{SE}_m}
\right)^\transpose
\CiD
\Normal_m(\bm{0},R_0).
\]
The plug-in componentwise-Studentized law used by Algorithm~\ref{alg:mc_critical} converges to the same Gaussian max law because $V\to\infty$ makes the simulated Wishart-diagonal denominators converge to one and $\hat R_{Jack}\CiP R_0$. The distribution of $\max_j|Z_j|$ for $Z\sim\Normal_m(\bm0,R_0)$ is continuous, so the plug-in critical value satisfies $\hat q_\alpha\CiP q_\alpha$, the $(1-\alpha)$ quantile of this limiting max law. Slutsky's theorem then gives the stated simultaneous coverage.
\endgroup
\end{proof}

\begin{proof}[Proof of Theorem~\ref{thm:bias_corrected}]
Write $\hat\Psi_j(P_n) - \Psi_j(P_0) = (\hat\Psi_j(P_n) - \E[\hat\Psi_j(P_n)]) + b_j$. The event $\{\Psi_j(P_0) \in [\hat\Psi_j(P_n) \pm (q_\alpha \widehat{SE}_j + |\hat b_j|)] \;\forall j\}$ is equivalent to $\{|\hat\Psi_j(P_n) - \Psi_j(P_0)| \le q_\alpha \widehat{SE}_j + |\hat b_j| \;\forall j\}$. By the triangle inequality,
\[
|\hat\Psi_j(P_n) - \Psi_j(P_0)| \le |\hat\Psi_j(P_n) - \E[\hat\Psi_j(P_n)]| + |b_j|.
\]
\textcolor{black}{By the triangle inequality, $|b_j| \le |\hat b_j| + |\hat b_j - b_j|$. Hence}
\[
|\hat\Psi_j(P_n) - \Psi_j(P_0)| \le |\hat\Psi_j(P_n) - \E[\hat\Psi_j(P_n)]| + |\hat b_j| + |\hat b_j - b_j|.
\]
It therefore suffices that
\begin{equation}\label{eq:bc_suffice}
    |\hat\Psi_j(P_n) - \E[\hat\Psi_j(P_n)]| + |\hat b_j - b_j| \le q_\alpha \widehat{SE}_j \qquad \forall\, j.
\end{equation}
Let \textcolor{black}{$\varepsilon_n = \max_{1 \le j \le m} |\hat b_j - b_j| / \widehat{SE}_j \CiP 0$ by the uniform relative-error condition}. Then (\ref{eq:bc_suffice}) holds whenever $\max_j |T_j^*| \le q_\alpha - \varepsilon_n$, where $T_j^* = (\hat\Psi_j(P_n) - \E[\hat\Psi_j(P_n)]) / \widehat{SE}_j$ is the centered Studentized statistic.

By Theorem~\ref{thm:multi_t}, the uncentered statistic $T_j = \sqrt{V}(\hat\Psi_{Jack}^{(j)}(P_n) - \Psi_j(P_0)) / \hat\sigma_j$ converges jointly in distribution to $\bT^\infty$. Using the exact algebraic identity $\hat\Psi_{Jack}^{(j)}(P_n) - \hat\Psi_j(P_n) = -\hat b_j$ alongside $b_j = \E[\hat\Psi_j(P_n)] - \Psi_j(P_0)$, we can write
\[
T_j = \frac{\hat\Psi_j(P_n) - \Psi_j(P_0) - \hat b_j}{\widehat{SE}_j} = \frac{\hat\Psi_j(P_n) - \E[\hat\Psi_j(P_n)] + b_j - \hat b_j}{\widehat{SE}_j} = T_j^* + \frac{b_j - \hat b_j}{\widehat{SE}_j}.
\]
Hence $T_j^* = T_j + \frac{\hat b_j - b_j}{\widehat{SE}_j} = T_j + o_p(1)$ uniformly over $j$. Thus, $\max_j |T_j^*| \CiD \max_j |T_j^\infty|$ under the asymptotic componentwise-Studentized distribution, whose cumulative distribution function $F(t) = P(\max_j |T_j^\infty| \le t)$ is continuous. Since $\varepsilon_n \CiP 0$, Slutsky's lemma gives
\[
\max_j |T_j^*| + \varepsilon_n \CiD \max_j |T_j^\infty|.
\]
Because $\{\max_j |T_j^*| \le q_\alpha - \varepsilon_n\} = \{\max_j |T_j^*| + \varepsilon_n \le q_\alpha\}$ and $F$ is continuous at $q_\alpha$, the Portmanteau theorem yields
\[
P\!\left(\max_j |T_j^*| \le q_\alpha - \varepsilon_n\right) = P\!\left(\max_j |T_j^*| + \varepsilon_n \le q_\alpha\right) \to F(q_\alpha) = 1 - \alpha.
\]
Hence the simultaneous coverage probability is at least $1 - \alpha + o(1)$.
\end{proof}

\subsection*{Proofs for Section~\ref{sec:slower_rates} (Generalized Asymptotically Linear Estimators)}

\begin{proof}[Proof of Theorem~\ref{thm:fixedV_slower}]
Let $W_{n,v} = P^1_{n,v} \varphi_n$ and $\bar W_n = V^{-1}\sum_{v=1}^V W_{n,v}$. Exact algebraic cancellation of the centered pseudo-values yields
\[
IC_{Jack}(v) - \hat\Psi_{Jack}(P_n) = (W_{n,v} - \bar W_n) + \Delta_{n,v},
\]
where $\Delta_{n,v} = -(V-1) \big\{ R_n(P_{n,-v}, P_0) - \bar{R}_{n,-v} \big\}$ and $\bar{R}_{n,-v} = V^{-1}\sum_{v'=1}^V R_n(P_{n,-v'}, P_0)$. Under Assumption \ref{assump:slower_rate}, $\max_{1 \le v \le V} |\Delta_{n,v}| = o_p(\sigma_n n^{-1/2})$.

Expanding the sample variance $\hat S_{Jack}^2$ gives:
\[
\hat S_{Jack}^2 = S_{W_n}^2 + \frac{1}{V-1}\sum_{v=1}^V \Delta_{n,v}^2 + \frac{2}{V-1}\sum_{v=1}^V (W_{n,v} - \bar W_n)\Delta_{n,v},
\]
where $S_{W_n}^2 = \frac{1}{V-1}\sum_{v=1}^V (W_{n,v} - \bar W_n)^2$. Because $\Var(W_{n,v}) = V \sigma_n^2 / n$, we have $W_{n,v} - \bar W_n = O_p(\sigma_n n^{-1/2})$. The cross-terms and squared remainder terms are therefore strictly $o_p(\sigma_n^2 n^{-1})$. Thus, $\hat S_{Jack}^2 = S_{W_n}^2 + o_p(\sigma_n^2 n^{-1}) = S_{W_n}^2 \{1 + o_p(1)\}$.

The Studentized statistic can then be written as:
\[
\frac{\sqrt{V}(\hat\Psi_{Jack}(P_n)-\Psi(P_0))}{\hat S_{Jack}}
= \frac{\sqrt{V}\bar W_n + \sqrt{V}\big(V R_n(P_n, P_0) - (V-1)\bar{R}_{n,-v}\big)}{S_{W_n} \{1+o_p(1)\}}.
\]
Dividing both the numerator and denominator by the unknown standard error scale $\sigma_n \sqrt{V/n}$ yields:
\[
\frac{\frac{\sqrt{n/V}}{\sigma_n} \sqrt{V} \bar W_n + o_p(1)}{\frac{\sqrt{n/V}}{\sigma_n} S_{W_n} \{1+o_p(1)\}}.
\]
Define the scaled fold averages $Z_{n,v} = \frac{\sqrt{n/V}}{\sigma_n} W_{n,v}$. By the Lindeberg condition in Assumption \ref{assump:slower_rate}, the vector $(Z_{n,1}, \dots, Z_{n,V})$ converges in distribution to $V$ independent standard normal random variables $Z_1, \dots, Z_V$. The scaled numerator converges to $\frac{1}{\sqrt{V}}\sum_{v=1}^V Z_v$ and the scaled denominator to the sample standard deviation $S_Z = \sqrt{\frac{1}{V-1}\sum_{v=1}^V (Z_v - \bar Z)^2}$. The continuous mapping theorem immediately yields the $t_{V-1}$ limit, with the diverging scale $\sigma_n$ canceling out perfectly.
\end{proof}

\begin{proof}[Proof of Corollary~\ref{cor:CI_slower}]
As in \cref{lem:centering}, the centering difference is $\hat\Psi_{Jack}(P_n) - \hat\Psi(P_n) = -\hat b_n$, where the jackknife bias estimate $\hat b_n = (V-1)\big(\bar\Psi_{(-v)} - \hat\Psi(P_n)\big)$. The influence curve contributions cancel exactly, leaving only the remainders: $-\hat b_n = (V-1)\{R_n(P_n, P_0) - \frac{1}{V}\sum_{v=1}^V R_n(P_{n,-v}, P_0)\}$. Under Assumption \ref{assump:slower_rate}, this difference is $o_p(\sigma_n n^{-1/2})$. Because $\widehat{SE} = \hat S_{Jack} / \sqrt{V} \asymp_p \sigma_n n^{-1/2}$, the centering difference is $o_p(\widehat{SE})$. The result follows by Slutsky's theorem.
\end{proof}

\begin{proof}[Proof of Theorem~\ref{thm:divergingV_slower}]
Following the exact algebraic decomposition from \cref{thm:divergingV}, we have:
\[
\frac{\hat{s}_n^2}{\sigma_n^2} = \frac{n}{V \sigma_n^2} \hat S_{W_n}^2 + \frac{n}{V(V-1) \sigma_n^2} \sum_{v=1}^V \Delta_{n,v}^2 + \text{cross-term}.
\]
The remainder stability assumption implies
\[
\sum_{v=1}^V \Delta_{n,v}^2
\le
(V-1)^2 V\|d_n\|_{V,2}^2
=
o_p\!\left(\frac{V\sigma_n^2}{n}\right).
\]
Multiplying by $\frac{n}{V(V-1)\sigma_n^2}$ ensures the squared-remainder contribution to the variance is $o_p(V^{-1})=o_p(1)$. By the Cauchy--Schwarz inequality, the cross-term is bounded by $o_p(V^{-1/2}) = o_p(1)$.

Thus, $\frac{\hat{s}_n^2}{\sigma_n^2} = \frac{n}{V \sigma_n^2} \hat S_{W_n}^2 + o_p(1)$. Expanding the sums in $\hat S_{W_n}^2$, we obtain the decomposition:
\[
\frac{n}{V \sigma_n^2} \hat S_{W_n}^2 = \frac{1}{n \sigma_n^2}\sum_{i=1}^n \varphi_n^2(O_i) + \frac{T_n}{\sigma_n^2} + \frac{1}{V-1}\frac{T_n}{\sigma_n^2} - \frac{1}{V-1}\frac{U_n}{\sigma_n^2},
\]
where $T_n = \frac{1}{n}\sum_{v=1}^V \sum_{i \neq j \in I_v} \varphi_{n,i} \varphi_{n,j}$ and $U_n = \frac{1}{n}\sum_{i\neq j} \varphi_{n,i} \varphi_{n,j}$. The variance calculation yields $\Var(T_n / \sigma_n^2) = \frac{1}{n^2 \sigma_n^4} \sum_{v=1}^V 2 (\frac{n}{V})(\frac{n}{V}-1) \sigma_n^4 \sim \frac{2}{V}$. Thus, $T_n / \sigma_n^2 = O_p(V^{-1/2}) = o_p(1)$, and $U_n / \sigma_n^2 = O_p(1)$ so the $(V-1)^{-1}$ terms vanish. Since $\frac{1}{n \sigma_n^2} \sum_{i=1}^n \varphi_n^2(O_i) \CiP 1$ by assumption, the conclusion follows: $\frac{\hat{s}_n^2}{\sigma_n^2} \CiP 1$.
\end{proof}

\begin{proof}[Proof of Corollary~\ref{cor:divergingV_wald_slower}]
\begingroup\color{black}
The scaled remainder condition gives $\sqrt n\{\hat\Psi(P_n)-\Psi(P_0)\}/\sigma_n=\sqrt n\,P_n\varphi_n/\sigma_n+o_p(1)$, and \cref{thm:divergingV_slower} gives $\hat s_n/\sigma_n\CiP 1$. Slutsky's theorem yields the normal limit and the coverage statement follows from the convergence of the $t$ critical values.
\endgroup
\end{proof}

\section{Connection to High-Dimensional Gaussian Approximation}\label{app:gaussian_approx}

This appendix records a broader connection between our simultaneous-inference problem and the literature on Gaussian and multiplier-bootstrap approximation for maxima of high-dimensional means. These results are not needed for the main theorems of the paper, but they help situate the grouped-jackknife construction relative to existing high-dimensional approximation methods.

The fold-level pseudo-value vectors provide a natural bridge to the Gaussian multiplier bootstrap framework of \citet{ChernozhukovChetverikovKato2013}. By Theorem~\ref{thm:multi_t}, the $m$-dimensional centered pseudo-value vectors $(IC_{Jack}^{(1)}(v) - \hat\Psi_{Jack}^{(1)}(P_n), \dots, IC_{Jack}^{(m)}(v) - \hat\Psi_{Jack}^{(m)}(P_n))^\transpose$, $v = 1, \dots, V$, behave asymptotically as $V$ independent draws from a mean-zero distribution with covariance proportional to $\bSigma$. This places the problem of constructing simultaneous inference for the mean vector $\Psi(P_0)$ squarely within the framework of the Gaussian multiplier bootstrap for high-dimensional means.

In the standard formulation of \citet{ChernozhukovChetverikovKato2013}, given $n$ independent centered observations $\bX_1, \dots, \bX_n \in \R^p$, the Gaussian multiplier bootstrap approximates the distribution of the max statistic $\max_{1 \le j \le p} |n^{-1/2} \sum_{i=1}^n X_{ij}|$ as follows: for $b = 1, \dots, B$, generate independent standard normal multipliers $\xi_1^{(b)}, \dots, \xi_n^{(b)} \stackrel{i.i.d.}{\sim} \Normal(0,1)$, compute the multiplier process $\max_{1 \le j \le p} |n^{-1/2} \sum_{i=1}^n \xi_i^{(b)} X_{ij}|$, and use the empirical distribution of these $B$ maxima to estimate quantiles. The key theoretical result is that this bootstrap consistently estimates the distribution of the max statistic even when $p$ greatly exceeds $n$, with no rank restriction on the covariance matrix, and the approximation error converges to zero uniformly over hyperrectangles even when $p$ grows as fast as $\exp(C\,n^c)$ for certain constants $C, c > 0$.

\citet{ChenKato2020} develop a ``jackknife multiplier bootstrap'' (JMB) for $U$-processes of general order $r \geqslant 2$ indexed by VC-type function classes. In their setting, the standard Gaussian multiplier bootstrap is infeasible because it requires the H\'{a}jek projection $P^{r-1}h$, which depends on the unknown distribution $P$. Their key idea is to replace each $P^{r-1}h(X_i)$ with its leave-one-out jackknife estimate, and they establish non-asymptotic coupling and Kolmogorov distance bounds showing that the resulting bootstrap consistently approximates the distribution of the $U$-process supremum with polynomial error rates. While this is the closest existing combination of jackknife ideas with simultaneous inference, it differs from our approach in several respects: (i)~it uses the leave-one-out (not grouped) jackknife, where each deletion removes a single observation; (ii)~the jackknife serves to estimate an unknown nuisance quantity (the H\'{a}jek projection) that is then fed into an external Gaussian multiplier scheme, rather than producing pseudo-values that are directly Studentized; (iii)~it targets suprema over infinite-dimensional function classes rather than finite-dimensional parameter vectors; and (iv)~the theoretical guarantees are non-asymptotic finite-sample bounds for $U$-processes, whereas our results provide asymptotic distributional limits for general RAL estimators of pathwise differentiable parameters.

In our setting, $V$ plays the role of $n$ and $m$ (the number of evaluation points) plays the role of $p$. One could, in principle, apply the Gaussian multiplier bootstrap directly to the $V$ centered pseudo-value vectors to approximate the distribution of $\max_{1 \le j \le m} |\sqrt{V}\,(\hat\Psi_{Jack}^{(j)}(P_n) - \Psi_j(P_0))/ \hat\sigma_j|$. However, the theoretical guarantees of \citet{ChernozhukovChetverikovKato2013} require the sample size to satisfy $n \gg (\log p)^c$ for certain $c > 0$, which translates to $V \gg (\log m)^c$ in our framework. Since typical choices of $V \in \{10, 20\}$ are modest, this high-dimensional rate condition may not be satisfied when $m$ is large, and the multiplier bootstrap approximation may be unreliable in this regime. This observation helps explain why our main text instead focuses on the exact componentwise-Studentized limit together with slow divergence of $V$, while viewing the high-dimensional Gaussian approximation as contextual rather than essential.

\section{True Dose-Response Curves for the HAL Simulation}\label{sec:app_dgp_truth}

\cref{fig:dgp_truth} displays the true dose-response curves $\psi_0(a) = \E_{P_0}[\E_{P_0}(Y \mid A=a, W)]$ for the four data generating processes used in the HAL simulation study (Section~\ref{sec:sim_hal}). DGP~1 is a smooth, monotonically increasing curve. DGP~2 exhibits pronounced oscillations driven by a $\sin(1.25\, a^{3/2})$ term. DGP~3 is smooth and nearly linear for $a \le 2$, then transitions to an oscillating pattern for $a > 2$ via an indicator-modulated sine component. The transition at $a = 2$ is continuous (since $\sin((0.8 \cdot 2)^2 - 2.56) = \sin(0) = 0$), but the derivative is discontinuous---a kink that the first-order HAL can represent, though the subsequent oscillations remain challenging for a piecewise-linear basis. DGP~4 features jump discontinuities at $a = 2$ and $a = 4$, where the curve jumps from approximately $0.12$ to $0.50$ and from $0.84$ to $0.16$, respectively. These four settings represent a progression from easy (smooth) to challenging (discontinuous) for nonparametric curve estimation, and the coverage results in Section~\ref{sec:sim_hal} should be interpreted in light of the varying smoothness of the underlying truth.

\begin{figure}[htbp]
\centering
\includegraphics[width=\textwidth]{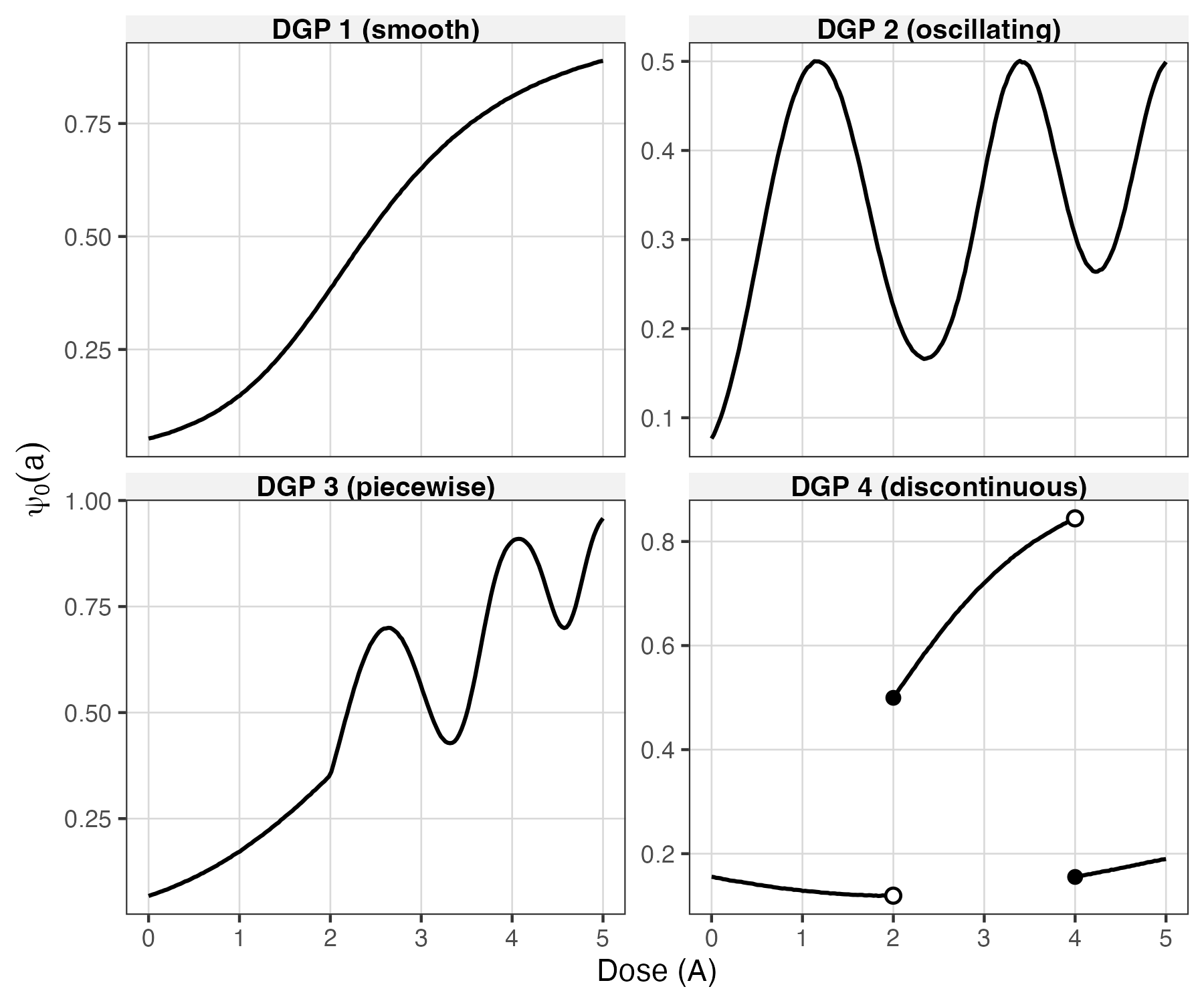}
\caption{True dose-response curves $\psi_0(a) = \E[Y \mid \text{do}(A=a)]$ for the four data generating processes used in the HAL simulation study. DGP~1: smooth monotone increase; DGP~2: oscillating; DGP~3: smooth for $a \le 2$, oscillating for $a > 2$; DGP~4: jump discontinuities at $a = 2$ and $a = 4$. True values computed by Monte Carlo integration ($10^6$ samples per evaluation point).}
\label{fig:dgp_truth}
\end{figure}

\section{Pointwise Coverage of Simultaneous Confidence Bands}\label{sec:app_simult_pw}

The simultaneous confidence bands reported in Tables~\ref{tab:km} and~\ref{tab:hal_simult} are designed to cover the entire curve simultaneously with probability $1-\alpha$. A natural diagnostic is to examine the \emph{pointwise} coverage of these simultaneous bands at each individual evaluation point. By construction, pointwise coverage must be at least as large as simultaneous coverage, since the event that the band covers every point implies coverage at each individual point. Here we display pointwise coverage profiles for the Kaplan--Meier and HAL simulation studies to examine how the choice of $V$ affects the uniformity of coverage across evaluation points.

\begin{figure}[htbp]
\centering
\includegraphics[width=0.7\textwidth]{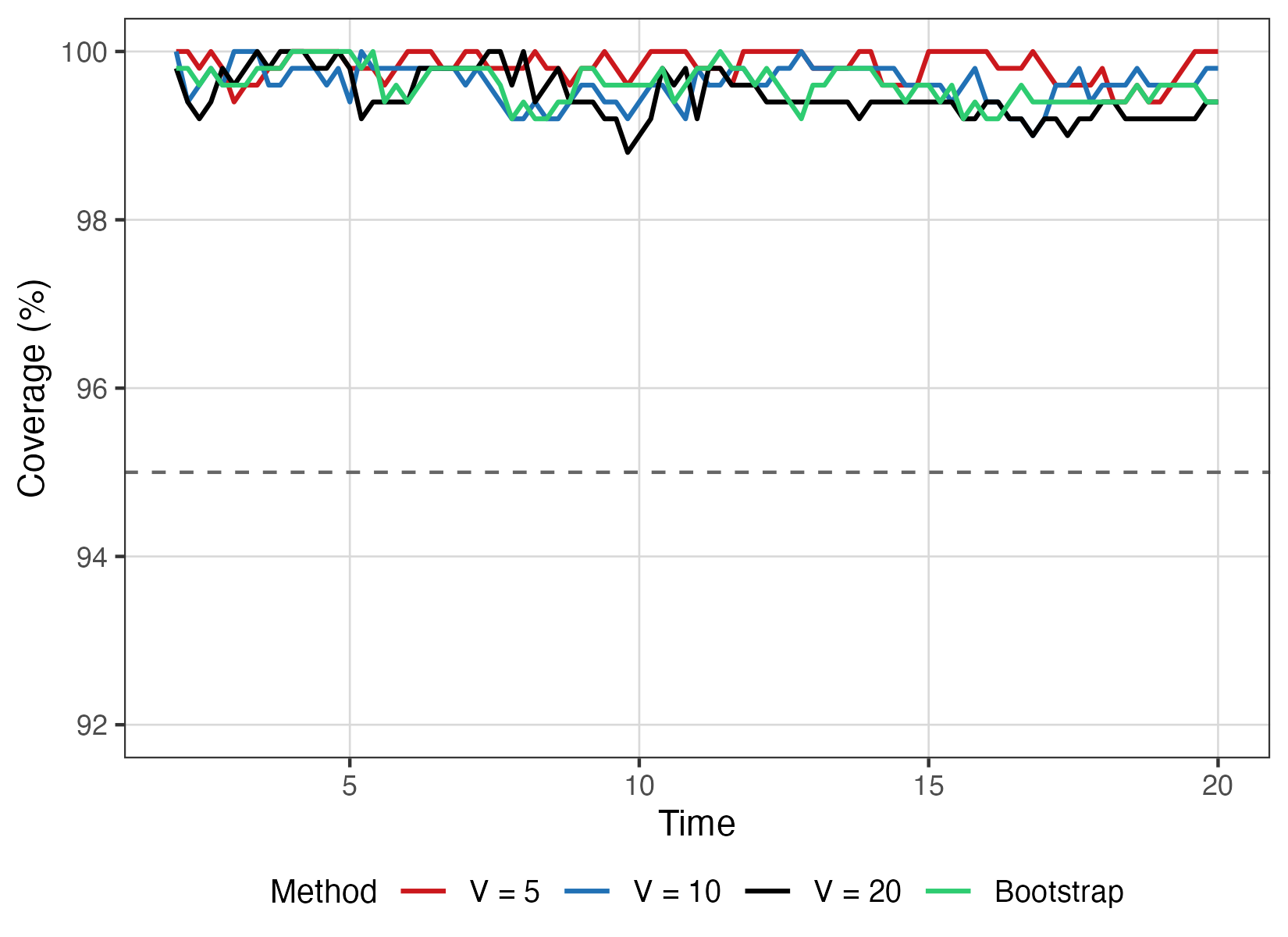}
\caption{Pointwise coverage (\%) of simultaneous 95\% confidence bands for the Kaplan--Meier survival function across evaluation times ($n = 200$, 500 replications). Methods: $V$-fold simultaneous band with $V \in \{5, 10, 20\}$ and the studentized bootstrap simultaneous band. Horizontal dashed line: nominal 95\% level. All methods achieve 98--100\% pointwise coverage uniformly across time points, with no systematic differences across $V$.}
\label{fig:km_simult_pointwise}
\end{figure}

\cref{fig:km_simult_pointwise} displays pointwise coverage for the Kaplan--Meier simultaneous bands. All methods---$V$-fold with $V \in \{5, 10, 20\}$ and the studentized bootstrap---achieve approximately 99--100\% pointwise coverage uniformly across the entire time range. The coverage profiles are nearly indistinguishable across $V$, showing no additional fluctuation for small $V$. This uniformity reflects the conservatism inherent in simultaneous bands: the simultaneous critical value $q_\alpha$ is substantially larger than the pointwise $t_{V-1}$ quantile, yielding pointwise coverage well above the nominal 95\% level at every evaluation point.

\begin{figure}[htbp]
\centering
\includegraphics[width=\textwidth,height=0.78\textheight,keepaspectratio]{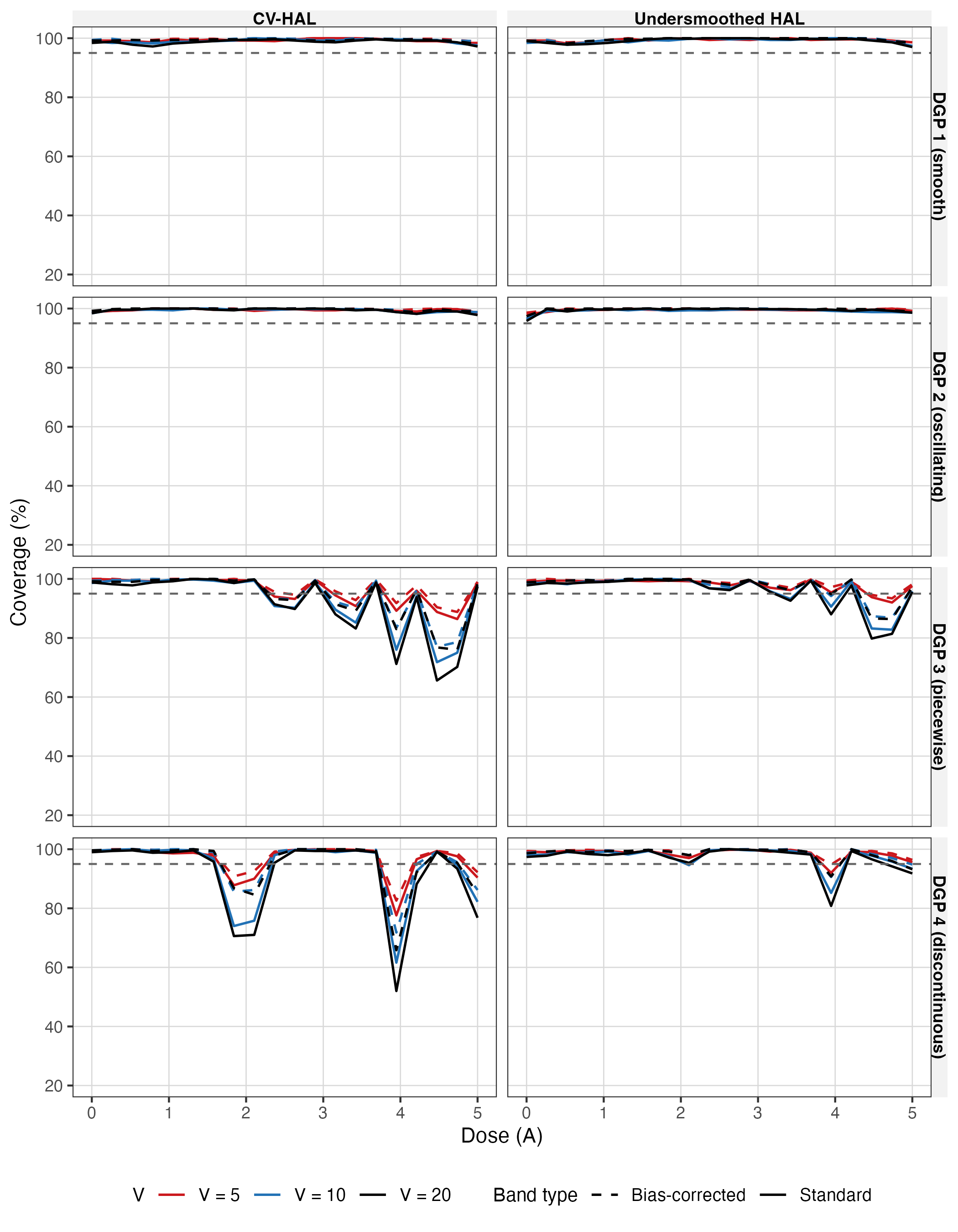}
\caption{Pointwise coverage (\%) of simultaneous 95\% confidence bands for the HAL dose-response curve across the dose grid ($n = 500$, 500 replications). Rows: DGP~1 (smooth), DGP~2 (oscillating), DGP~3 (piecewise), DGP~4 (discontinuous). Columns: CV-HAL and undersmoothed HAL. Solid lines: standard simultaneous bands; dashed lines: bias-corrected simultaneous bands. Colors indicate $V \in \{5, 10, 20\}$. Horizontal dashed line: nominal 95\% level.}
\label{fig:hal_simult_pointwise}
\end{figure}

\cref{fig:hal_simult_pointwise} shows the corresponding pointwise coverage profiles for the HAL dose-response simultaneous bands. For DGPs~1 (smooth) and~2 (oscillating), pointwise coverage is uniformly high (approximately 95--100\%) across all dose points for both CV-HAL and undersmoothed HAL, and the bias-corrected bands achieve even more uniform coverage. For DGPs~3 (piecewise) and~4 (discontinuous), coverage dips at specific dose points where the true curve is non-smooth, but these dips are consistent across all values of $V$---the coverage profiles for $V = 5$, $V = 10$, and $V = 20$ are nearly superimposed.

A striking feature of these figures is that small $V$ does not produce more variable pointwise coverage across evaluation points compared to large $V$. The coverage profiles are remarkably stable across $V$. This stability arises because the simultaneous critical value $q_\alpha$ compensates for the additional uncertainty associated with fewer folds: smaller $V$ yields larger $t_{V-1}$-based critical values, which uniformly inflate the band width at every evaluation point, producing stable (if conservative) pointwise coverage. The variation in pointwise coverage across evaluation points is driven primarily by the estimator's bias profile---which is a property of the estimator and the data generating process---rather than by the choice of $V$.

This behavior contrasts with pointwise confidence intervals (cf.\ \cref{fig:hal_pointwise}), where different methods (delta method, $V$-fold $t$, bootstrap) produce visibly different coverage profiles across dose points. The simultaneous bands homogenize coverage across evaluation points because the simultaneous critical value acts as a uniform multiplier of the standard error, so the pointwise coverage at each point is determined primarily by the ratio of the estimator's bias to its standard error at that point.

\section{Jackknife-Centered versus Bias-Corrected Confidence Intervals}\label{sec:app_jc}

The standard $V$-fold jackknife confidence intervals center at the full-sample estimator $\hat\Psi_j(P_n)$, while the bias-corrected band from \cref{sec:bias_corrected_bands} widens each interval by $|\hat b_j|$ to accommodate bias. A natural alternative is to \emph{shift} the center to the jackknife estimator $\hat\Psi_{Jack}^{(j)}(P_n) = \hat\Psi_j(P_n) - \hat b_j$, which subtracts the estimated bias while keeping the interval width unchanged. The resulting ``jackknife-centered'' (JC) interval is
\[
\hat\Psi_{Jack}^{(j)}(P_n) \pm \text{crit} \cdot \widehat{SE}_j = \left(\hat\Psi_j(P_n) - \hat b_j\right) \pm \text{crit} \cdot \widehat{SE}_j,
\]
where $\text{crit} = t_{V-1,\,1-\alpha/2}$ for pointwise or $\hat q_\alpha$ for simultaneous bands. Because JC has the same width as the standard interval (and narrower than the BC band), it would be more efficient if the bias estimate $\hat b_j$ were accurate.

\cref{fig:hal_jc_pointwise} compares pointwise coverage across the dose grid for the three centering strategies (standard, JC, and BC) using $V = 20$. Contrary to the expectation that bias correction via centering might help, the JC intervals exhibit substantially \emph{lower} coverage than the standard intervals in every DGP and HAL variant. For DGP~1 (smooth), pointwise coverage drops from 93.7\% to 86.0\% (CV-HAL) and from 95.6\% to 85.7\% (undersmoothed). For DGP~2 (oscillating), the drop is more severe: from 95.2\% to 71.9\% (CV-HAL). For DGPs~3 and~4, where bias is already problematic, JC coverage falls to 64--72\%. The simultaneous results are even more striking (\cref{fig:hal_jc_simult_pointwise}): for DGP~2, simultaneous coverage drops from 91.8\% to 21.0\% (CV-HAL), and for DGP~4 from 25.4\% to 5.8\%. In contrast, the BC approach consistently achieves the highest coverage across all settings, reaching near-nominal levels for DGPs~1--2.

\begin{figure}[htbp]
\centering
\includegraphics[width=\textwidth,height=0.78\textheight,keepaspectratio]{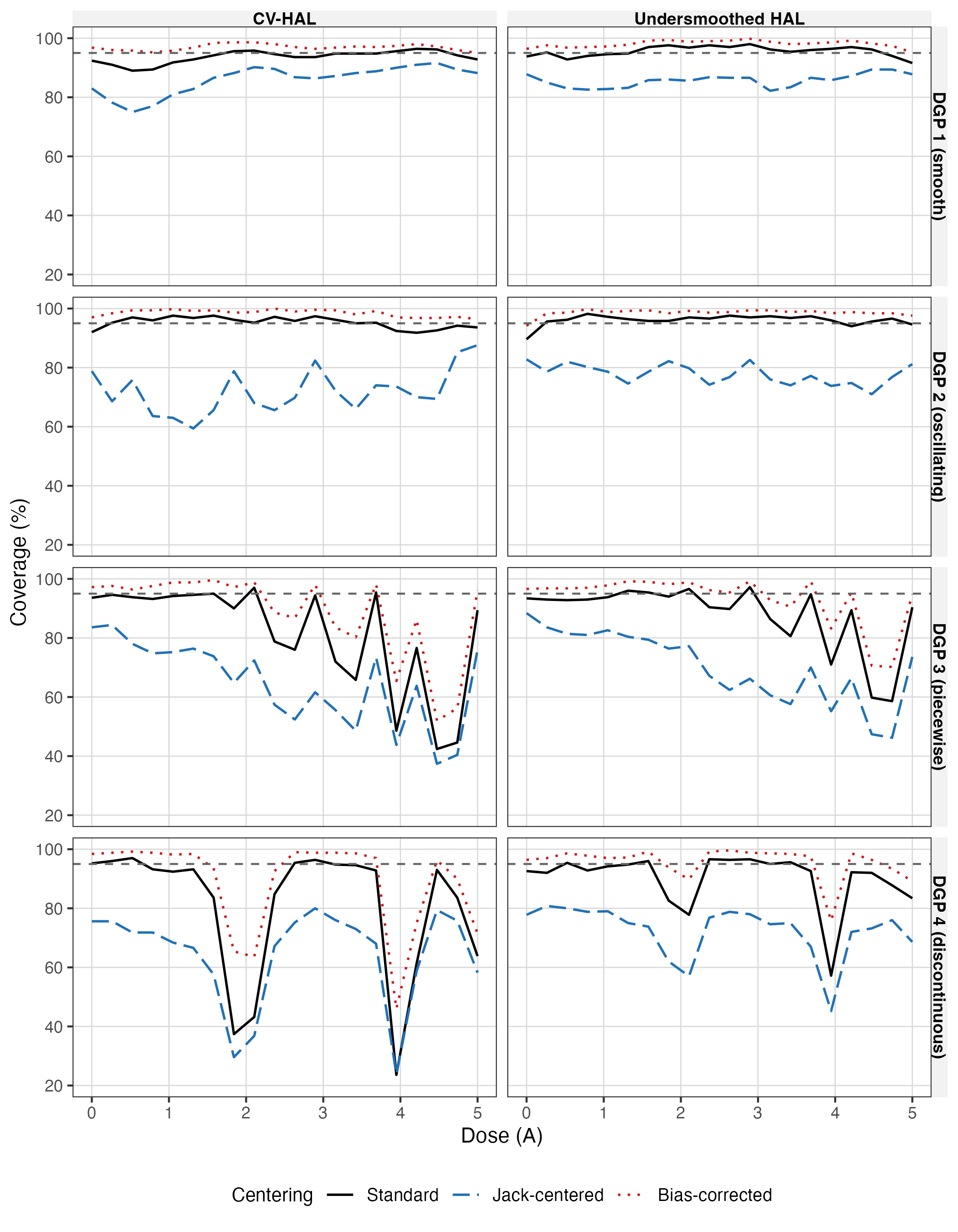}
\caption{Pointwise coverage (\%) of 95\% confidence intervals for the HAL dose-response curve, comparing three centering strategies ($V = 20$, $n = 500$, 500 replications). Standard: centered at $\hat\Psi_j(P_n)$; Jack-centered (JC): centered at $\hat\Psi_{Jack}^{(j)}(P_n) = \hat\Psi_j(P_n) - \hat b_j$; Bias-corrected (BC): centered at $\hat\Psi_j(P_n)$ with interval widened by $|\hat b_j|$. The JC intervals have lower coverage than the standard intervals in all settings, while BC achieves the highest coverage.}
\label{fig:hal_jc_pointwise}
\end{figure}

\begin{figure}[htbp]
\centering
\includegraphics[width=\textwidth,height=0.78\textheight,keepaspectratio]{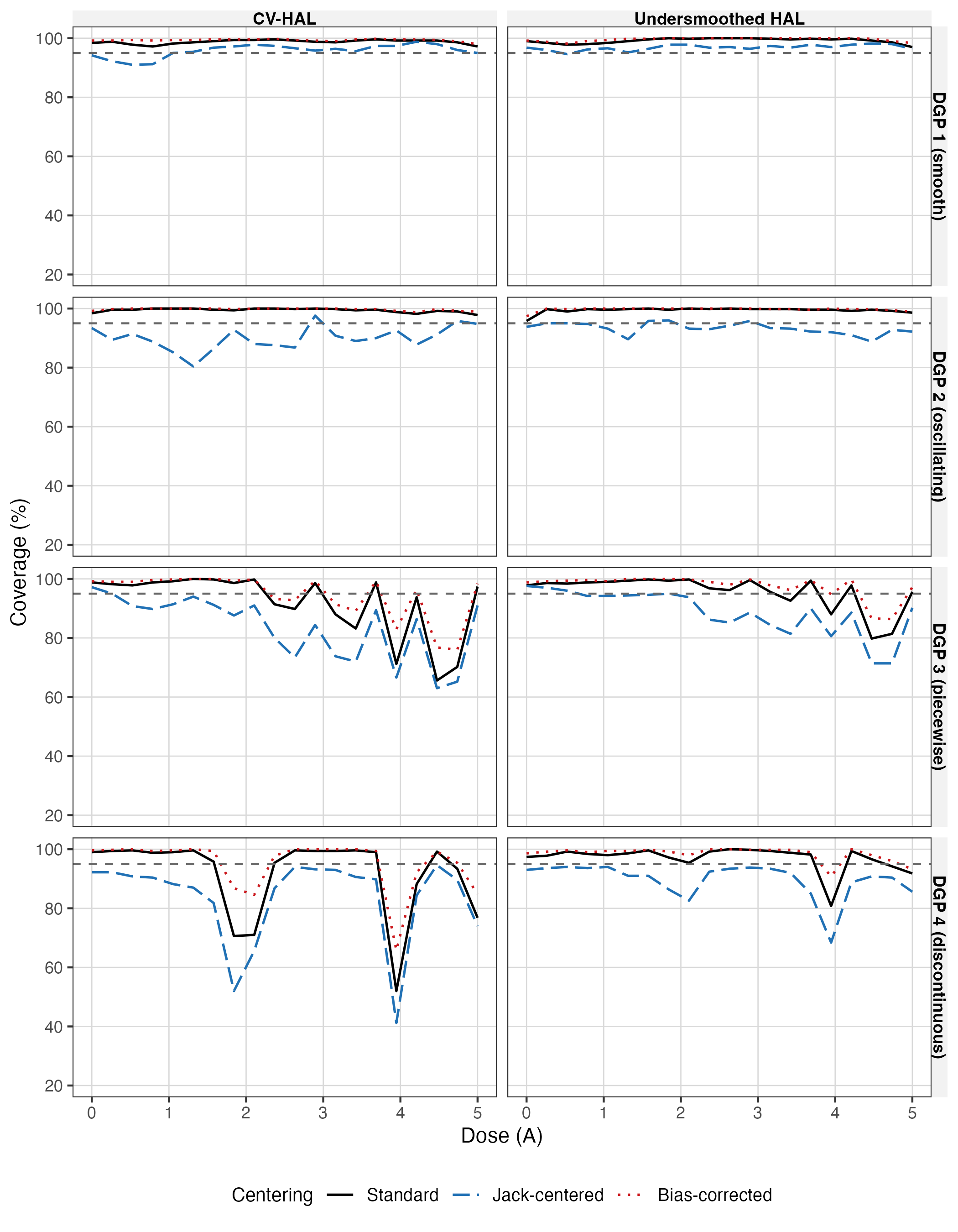}
\caption{Pointwise coverage (\%) of simultaneous 95\% confidence bands for the HAL dose-response curve, comparing three centering strategies ($V = 20$, $n = 500$, 500 replications). Same conventions as \cref{fig:hal_jc_pointwise} but using the simultaneous critical value $\hat q_\alpha$ in place of $t_{V-1,\,0.975}$. The coverage degradation from JC centering is even more pronounced for simultaneous bands.}
\label{fig:hal_jc_simult_pointwise}
\end{figure}

The failure of JC centering can be understood through the lens of \cref{lem:centering}. For pathwise differentiable parameters, the centering mismatch $\hat b_j = \hat\Psi_j(P_n) - \hat\Psi_{Jack}^{(j)}(P_n)$ satisfies $\hat b_j = o_p(\widehat{SE}_j)$, so the JC and standard intervals are asymptotically equivalent. However, for the HAL dose-response estimator with slower convergence rate $\sqrt{n/J_n}$, the bias $b_j$ can be of the same order as $\widehat{SE}_j$ or larger, as \cref{fig:hal_bias_diagnosis} confirms. The jackknife bias estimate $\hat b_j$ is computed from only $V = 20$ pseudo-values and is therefore noisy: its estimation error $\hat b_j - b_j$ is not negligible relative to $\widehat{SE}_j$. Shifting the center by $\hat b_j$ introduces a first-order noise term that inflates the effective variance of the centered estimator without reliably reducing its bias.  In contrast, the BC approach widens the interval by $|\hat b_j|$, which is valid by the triangle inequality argument in \cref{thm:bias_corrected}: the widening accommodates bias regardless of the sign of the estimation error $\hat b_j - b_j$. These results confirm the paper's recommendation to center confidence intervals at the full-sample estimator $\hat\Psi_j(P_n)$ and address bias through widening (BC) rather than shifting (JC).

\section{Diagnostic: Highlighting Unreliable Grid Points in a Simultaneous Band}\label{sec:app_diagnostic}

The simultaneous confidence band \eqref{eq:simul_band} presumes that the asymptotically linear expansion is equally reliable at every grid point $j$. In finite samples this may fail: at some $j$ the remainder $R_{n,j}$ can be large relative to the standard error, so the pointwise interval --- and hence its contribution to the simultaneous band --- is suspect. This is a property of the underlying \emph{estimator}, not of the $V$-fold construction; it would be present under any inferential procedure that uses the same estimator.

Because the jackknife already produces a fold-level bias estimate $\hat b_j = (V-1)\big(\bar\Psi_{j,(-v)} - \hat\Psi_j(P_n)\big)$ (see \eqref{eq:jack_bias}) and a fold-level standard error $\widehat{SE}_j$, the ratio
\[
\rho_j = \frac{|\hat b_j|}{\widehat{SE}_j}
\]
is computed at no extra cost and serves as a natural point-level reliability diagnostic: large $\rho_j$ indicates grid points where the linearization is under stress, and small $\rho_j$ indicates well-targeted points.

We therefore recommend plotting the simultaneous band with a \emph{color overlay} proportional to $\rho_j$ --- for example, shading grid points with $\rho_j \ge \tau$ (e.g.\ $\tau = 1$) to warn the reader that coverage may be unreliable there. This is a visualization aid, not a modification of the simultaneous band: the band \eqref{eq:simul_band} remains defined on the full grid $\{1,\dots,m\}$, and the overall $(1-\alpha)$-simultaneous coverage guarantee is unaffected. Analysts who prefer to report a confidence band over a restricted set --- e.g.\ to describe the dose-response curve only at well-targeted doses --- may use the flagged points as an \emph{informal} data-adaptive restriction, with the understanding that post-selection inference issues then apply and the simultaneous coverage over the selected set is not formally controlled at level $1-\alpha$.

The bias-corrected bands from \cref{sec:bias_corrected_bands} and this diagnostic serve complementary roles: the former widens intervals by $|\hat b_j|$ to provide asymptotic coverage when $\hat b_j$ estimates $b_j$ well, while the latter transparently communicates to the user where the linearization assumption is most strained.

\section{Isolating the Effect of Correlation Matrix Estimation}\label{sec:app_corr_isolation}

In the simultaneous confidence band simulations of Section~\ref{sec:sim_km}, changing $V$ simultaneously alters three components of the inference procedure: (i)~the quality of the estimated correlation matrix $\hat R_{\mathrm{Jack}}$ (more folds yield a higher-rank, more precise estimate), (ii)~the degrees of freedom $\nu = V - 1$ in the $t$-distribution used for the Monte Carlo critical value, and (iii)~the standard error estimates. This confounding makes it difficult to attribute coverage differences across $V$ to any single factor. To disentangle these effects, we conduct a supplementary simulation that \emph{isolates} the correlation matrix by fixing the other two components.

Specifically, we fix the degrees of freedom at $\nu = 19$ (corresponding to $V = 20$), use the standard errors from a $V = 20$ jackknife, and center all confidence bands at the full-data Kaplan--Meier estimate. The only quantity that varies is the source of the correlation matrix: $\hat R_{\mathrm{Jack}}$ is estimated from $V_{\mathrm{corr}} \in \{5, 10, 20\}$ folds. The simultaneous confidence band is then
\[
\hat S(t) \pm \hat q_\alpha\!\bigl(\hat R_{V_{\mathrm{corr}}},\, \nu = 19\bigr) \cdot \widehat{SE}_{20}(t),
\]
where $\hat q_\alpha(\hat R_{V_{\mathrm{corr}}}, \nu = 19)$ denotes the Monte Carlo critical value computed from Algorithm~1 using $\hat R_{V_{\mathrm{corr}}}$ as the correlation matrix and $\nu = 19$ degrees of freedom. The data generating process, sample size ($n = 200$), evaluation grid ($m = 91$ points on $[2, 20]$), and number of replications ($500$) are identical to Section~\ref{sec:sim_km}. We note that the ``Standard'' and ``Isolation'' rows in \cref{tab:corr_isolation} come from separate simulation runs with independent fold assignments, though both use the same data-generating mechanism and random seed for dataset generation.

\begin{table}[htbp]
\centering
\caption{Decomposing the effect of $V$ on simultaneous coverage for the Kaplan--Meier survival function ($n=200$, $m=91$, 500 replications). ``Standard'' rows use $V$ folds for all components (correlation matrix $\hat R_{\mathrm{Jack}}$, degrees of freedom $\nu = V{-}1$, and standard errors). ``Isolation'' rows fix $\nu = 19$ and SE from $V{=}20$, varying only the correlation matrix source $\hat R_{\mathrm{Jack}}$ estimated from $V_{\mathrm{corr}}$ folds.}
\label{tab:corr_isolation}
\begin{tabular}{llccccc}
\toprule
Setting & $V$ & $\nu$ & Simult.\ Cov.\ (\%) & Mean $q_\alpha$ & Mean Width & $\hat r(\hat R)$ \\
\midrule
Standard & 5 & 4 & 93.6 & 6.52 & 0.454 & 1.6 \\
Standard & 10 & 9 & 93.4 & 3.91 & 0.282 & 1.7 \\
Standard & 20 & 19 & 94.6 & 3.28 & 0.239 & 1.7 \\
\addlinespace[3pt]
Isolation & 5 & 19 & 92.8 & 3.11 & 0.229 & 1.6 \\
Isolation & 10 & 19 & 93.8 & 3.23 & 0.237 & 1.7 \\
Isolation & 20 & 19 & 95.2 & 3.28 & 0.240 & 1.7 \\
\bottomrule
\end{tabular}
\end{table}

\cref{tab:corr_isolation} reveals that the correlation matrix quality and the $t$-distribution tail behavior act in \emph{opposite directions} and partially cancel. In the isolation rows, where only the correlation matrix varies, the coverage gap between $V_{\mathrm{corr}} = 5$ and $V_{\mathrm{corr}} = 20$ is $2.4$ percentage points ($92.8\%$ versus $95.2\%$): the rank-deficient $\hat R$ from $V_{\mathrm{corr}} = 5$ underestimates the $\max|T_j^\infty|$ quantile, yielding a smaller critical value ($\bar q_\alpha = 3.11$ versus $3.28$) and narrower bands. In the standard simulation, however, the total gap between $V = 5$ and $V = 20$ is only $1.0$ percentage point ($93.6\%$ versus $94.6\%$). The difference is explained by the heavy tails of $t_4$: for the standard $V = 5$ procedure, the $t_4$-distributed Monte Carlo draws inflate the critical value to $\bar q_\alpha = 6.52$---more than double the isolation value---producing much wider bands ($0.454$ versus $0.229$) that compensate for the imprecise $\hat R$. Thus, the heavier $t_{V-1}$ tails for small $V$ serve as a natural safeguard, inflating the critical value enough to partially offset the less precise correlation matrix estimate. The estimated effective rank $\hat r(\hat R) \approx 1.6$--$1.7$ is stable across all settings, confirming the low intrinsic dimensionality of the Kaplan--Meier correlation structure discussed in \cref{sec:eigen_validity}. Even $V = 5$, which yields a correlation matrix of rank at most~4, far exceeds this effective rank, and the resulting simultaneous coverage remains close to nominal.

\section{\textcolor{black}{Derivation of the Growth-Rate Condition on $V$}}\label{app:V_rate}
\begingroup\color{black}

The diverging-$V$ results (Theorems~\ref{thm:divergingV} and~\ref{thm:divergingV_slower}) assume a fold-level RMS bound on the remainder differences. This appendix gives sufficient TMLE, one-step, and AIPW-style structural conditions under which that assumption follows from an explicit upper bound on the growth rate of $V$.

\textcolor{black}{This derivation proceeds under an oracle formulation treating $D_n$ as conditionally fixed. If $D_n$ is fully data-adaptive, rigorous verification requires standard stochastic equicontinuity conditions in addition to the rate calculations below.}

{
Before imposing the structural decomposition below, a coarse direct-rate diagnostic is useful. Suppose that, for some total-remainder rate $r_n$,
\[
\|d_n\|_{V,2}=O_p(r_n),
\qquad
d_{n,v}=R_n(P_n,P_0)-R_n(P_{n,-v},P_0).
\]
This follows, for example, from the crude bound
$|d_{n,v}|\le |R_n(P_n,P_0)|+|R_n(P_{n,-v},P_0)|$ if the full-sample and leave-fold-out total remainders are uniformly $O_p(r_n)$. The stability condition in Theorem~\ref{thm:divergingV_slower} then requires
\[
r_n=o\!\left(\frac{\sigma_n}{\sqrt n\,V}\right),
\qquad\text{or equivalently}\qquad
V=o\!\left(\frac{\sigma_n}{\sqrt n\,r_n}\right).
\]
For a standard TMLE construction, where the target influence-curve variance is bounded so that $\sigma_n=O(1)$, the total remainder rate $r_n$ typically combines the exact second-order remainder and the empirical-process term involving the estimated efficient influence curve. If the nuisance estimators are fit by $0$-order HAL, these two pieces are controlled at the HAL rate $r_n=O_p(n^{-2/3})$, and the direct diagnostic gives the sufficient growth restriction $V=o(n^{1/6})$. If the nuisance estimators are fit by $1$st-order HAL, the corresponding rate is $r_n=O_p(n^{-4/5})$, giving $V=o(n^{3/10})$. Throughout this appendix, $0$-order HAL enters only as a \emph{nuisance} estimator, for which only convergence rates (and the heuristic oracle linearization scale $\tilde\sigma_n \asymp n^{1/6}$) are required; the pointwise-normality theory invoked for the \emph{target} dose-response curve requires $k \ge 1$ \citep[Corollary~1]{vanderLaan2023}. This direct diagnostic is mathematically valid but conservative, because it does not exploit that $d_{n,v}$ is a fold perturbation. The structured calculation below uses this perturbation to obtain the sharper sufficient bound.
}

\subsection*{G.1 Bounding the leading linear difference}

Let $P_{n,-v}$ be the empirical measure after deleting fold $v$, with $|I_v|=n/V$ and $n_v=n(V-1)/V$. For notational simplicity, write fold $v$ as the last block. Then the difference between the full-sample and leave-fold-out empirical averages of an arbitrary function $f$ is
\begin{align*}
P_n f - P_{n,-v} f
&= \Big(\tfrac{1}{n} - \tfrac{1}{n_v}\Big)\sum_{i=1}^{n_v}f(O_i) + \frac{1}{n}\sum_{i=n_v+1}^{n}f(O_i) \\
&= \underbrace{\frac{n_v - n}{n}\cdot\frac{1}{n_v}\sum_{i=1}^{n_v}f(O_i)}_{A}
 + \underbrace{\frac{1}{V}\cdot\frac{1}{n/V}\sum_{i=n_v+1}^{n}f(O_i)}_{B}.
\end{align*}
Because $(n_v-n)/n = -1/V$, direct variance calculations yield $\Var(A) \asymp \Var_{P_0}(f) / (nV^2)$ and $\Var(B) = \Var_{P_0}(f) / (nV)$. Since $\Var(B)$ dominates by a factor of $V$, we obtain the following foldwise rate.

\begin{lemma}[Linear-term rate]
\label{lem:linear_rate}
For a function $f$ with $\Var_{P_0}(f) = \sigma^2 < \infty$,
\[
|P_n f - P_{n,-v} f|
= O_p\!\left(\frac{\sigma}{\sqrt{n}\,\sqrt{V}}\right).
\]
Consequently,
\[
\left\{\frac{1}{V}\sum_{v=1}^V (P_n f-P_{n,-v}f)^2\right\}^{1/2}
=
O_p\!\left(\frac{\sigma}{\sqrt{n}\,\sqrt{V}}\right).
\]
\end{lemma}

\subsection*{G.2 Formalizing the remainder and nuisance expansion}

We formalize the properties of second-order remainders into the following sufficient structural assumption, explicitly separating the variance rate of the target parameter from the variance rate of the nuisance estimation error. Let $\sigma_n^2 = \Var_{P_0}(\varphi_n)$ denote the variance of the target parameter's influence curve, and let $\tilde\sigma_{n}^2$ denote the variance scale of the nuisance estimator's leading linear approximation.

\begin{assumption}[Nuisance Linearization and Remainder Structure]\label{assump:remainder_structure}
Let $\hat{\eta}_n$ denote the nuisance estimator based on a sample of size $n$, let $\hat{\eta}_{n,-v}$ denote the corresponding leave-fold-out nuisance estimator, and let $\|\cdot\|_{\eta}$ be a nuisance-error norm, such as an $L_2(P_0)$ norm or a product norm over nuisance components. We assume:
\begin{enumerate}
    \item \textbf{Nuisance Linearization:} The difference $\hat{\eta}_n - \eta_0$ admits the expansion
    \[
    \hat{\eta}_n - \eta_0 = P_n D_n + r_n,
    \qquad
    \hat{\eta}_{n,-v} - \eta_0 = P_{n,-v}D_n + r_{n,-v},
    \]
    where $D_n$ is a mean-zero score-like term satisfying $\E_{P_0}[D_n^2] \asymp \tilde\sigma_{n}^2$, \textcolor{black}{and the residual terms satisfy}
    {\color{black}
    \[
    \left\{ \frac{1}{V}\sum_{v=1}^V
    \|r_n-r_{n,-v}\|_\eta^2
    \right\}^{1/2}
    =
    o_p\!\left(\frac{\tilde\sigma_n}{\sqrt{nV}}\right)
    \]
    and
    \[
    \|r_n\|_\eta+
    \left\{\frac{1}{V}\sum_{v=1}^V
    \|r_{n,-v}\|_\eta^2
    \right\}^{1/2}
    =
    o_p\!\left(\frac{\tilde\sigma_n}{\sqrt n}\right).
    \]}
    \item \textbf{Remainder Structure:} The remainder admits the exact symmetric bilinear expansion
    \[
    R(P_n,P_0)
    =
    \mathcal B_n(\hat\eta_n-\eta_0,\hat\eta_n-\eta_0)+\rho_n,
    \]
    with the analogous expression for $R(P_{n,-v},P_0)$, where $\mathcal B_n$ is a bounded symmetric bilinear form satisfying
    \[
    |\mathcal B_n(f,g)|\lesssim \|f\|_{\eta}\|g\|_{\eta},
    \]
    and \textcolor{black}{the residual remainder difference satisfies}
    {\color{black}
    \[
    \left\{\frac{1}{V}\sum_{v=1}^V
    (\rho_n-\rho_{n,-v})^2
    \right\}^{1/2}
    =
    o_p\!\left(\frac{\tilde\sigma_n^2}{n\sqrt V}\right).
    \]}
\end{enumerate}
This assumption is a sufficient condition for common TMLE, one-step, and AIPW second-order remainders, whose leading terms are controlled by products of nuisance estimation errors. It is not a generic consequence of regular asymptotic linearity.
\end{assumption}

{
\begin{remark}[Interpretation of the residual term]
For TMLE, one-step, and AIPW-style estimators, $R_n$ is the total remainder in the asymptotic-linear expansion. The bilinear term above represents the exact second-order remainder. The residual $\rho_n$ should therefore be read as the part not captured by that exact bilinear form, typically an empirical-process term involving the estimated efficient influence curve, such as
$(P_n-P_0)\{D^*(\hat P_n)-D^*(P_0)\}$, together with its leave-fold-out analogue.

For $0$-order HAL nuisance estimators in a TMLE construction, $\tilde\sigma_n\asymp n^{1/6}$, so the residual condition in Assumption~\ref{assump:remainder_structure} becomes
\[
\left\{\frac{1}{V}\sum_{v=1}^V(\rho_n-\rho_{n,-v})^2\right\}^{1/2}
=o_p\!\left(\frac{n^{-2/3}}{\sqrt V}\right).
\]
For $1$st-order HAL nuisance estimators, $\tilde\sigma_n\asymp n^{1/10}$, and the same condition becomes
\[
\left\{\frac{1}{V}\sum_{v=1}^V(\rho_n-\rho_{n,-v})^2\right\}^{1/2}
=o_p\!\left(\frac{n^{-4/5}}{\sqrt V}\right).
\]
Thus the intended HAL verification is that the empirical-process residual $\rho_n$ has the HAL rate, and the full-versus-leave-fold-out comparison supplies the same fold-perturbation factor $V^{-1/2}$. With the required little-$o_p$ margin, this gives the fold-RMS residual-difference rates $o_p(n^{-2/3}/\sqrt V)$ for $0$-order HAL and $o_p(n^{-4/5}/\sqrt V)$ for $1$st-order HAL, so the residual condition is satisfied. The rate statement should be understood as applying to the fold-RMS difference $\rho_n-\rho_{n,-v}$, rather than only to the marginal size of a single residual term.
\end{remark}
}

\subsection*{G.3 Bounding the remainder difference}

Define
\[
d_{n,v}=R(P_n,P_0)-R(P_{n,-v},P_0),
\qquad
\|d_n\|_{V,2}=\left(\frac{1}{V}\sum_{v=1}^V d_{n,v}^2\right)^{1/2}.
\]
Under Assumption~\ref{assump:remainder_structure}, introduce the two nuisance-error factors
\[
A_{n,v}=\hat\eta_n-\hat\eta_{n,-v},
\qquad
C_{n,v}=\hat\eta_n+\hat\eta_{n,-v}-2\eta_0 .
\]
Using the symmetric bilinear form, the leading part of the remainder difference can be written as
\begin{align*}
d_{n,v}
&=
\mathcal B_n(\hat\eta_n-\eta_0,\hat\eta_n-\eta_0)
-
\mathcal B_n(\hat\eta_{n,-v}-\eta_0,\hat\eta_{n,-v}-\eta_0)
\\
&\quad +(\rho_n-\rho_{n,-v})\\
&=
\mathcal B_n(A_{n,v},C_{n,v})
+(\rho_n-\rho_{n,-v}).
\end{align*}

The two factors have different orders. The first factor is the leave-fold-out perturbation of the nuisance estimator. By the linear expansion in Assumption~\ref{assump:remainder_structure},
\[
A_{n,v}
=
(P_n-P_{n,-v})D_n+(r_n-r_{n,-v}).
\]
Thus Lemma~\ref{lem:linear_rate}, applied to $D_n$ with variance scale $\tilde\sigma_n^2$, and the fold RMS residual bound give
\[
\left\{\frac{1}{V}\sum_{v=1}^V
\|A_{n,v}\|_{\eta}^2\right\}^{1/2}
=
O_p\!\left(\frac{\tilde\sigma_n}{\sqrt n\,\sqrt V}\right)
+o_p\!\left(\frac{\tilde\sigma_n}{\sqrt{nV}}\right)
=
O_p\!\left(\frac{\tilde\sigma_n}{\sqrt n\,\sqrt V}\right),
\]
where the $V^{-1/2}$ improvement comes from comparing the full-sample estimator with a leave-fold-out estimator.

The second factor is the overall nuisance estimation error on the two samples:
\[
C_{n,v}
=
(P_n+P_{n,-v})D_n+(r_n+r_{n,-v}).
\]
It therefore has the ordinary nuisance-error size, rather than the leave-fold-out difference size:
\[
\left\{\frac{1}{V}\sum_{v=1}^V
\|C_{n,v}\|_{\eta}^2\right\}^{1/2}
=
O_p\!\left(\frac{\tilde\sigma_n}{\sqrt n}\right).
\]
Equivalently, since $C_{n,v}=2(\hat\eta_n-\eta_0)-A_{n,v}$, the common part of $C_{n,v}$ is of order $\tilde\sigma_n/\sqrt n$ and its fold-varying part is no larger than the already controlled $A_{n,v}$ term. In the oracle rate calculation below, we use this ordinary-size control foldwise:
\[
\max_{1\le v\le V}\|C_{n,v}\|_\eta
=
O_p\!\left(\frac{\tilde\sigma_n}{\sqrt n}\right).
\]

Now apply the boundedness of the bilinear form fold by fold:
\[
|\mathcal B_n(A_{n,v},C_{n,v})|
\lesssim
\|A_{n,v}\|_\eta \|C_{n,v}\|_\eta .
\]
Combining the fold RMS bound for $A_{n,v}$ with the ordinary nuisance-error size of $C_{n,v}$ gives
\[
\left\{\frac{1}{V}\sum_{v=1}^V
|\mathcal B_n(A_{n,v},C_{n,v})|^2
\right\}^{1/2}
\lesssim
\left(\max_{1\le v\le V}\|C_{n,v}\|_\eta\right)
\left\{\frac{1}{V}\sum_{v=1}^V
\|A_{n,v}\|_\eta^2
\right\}^{1/2}
=
O_p\!\left(\frac{\tilde\sigma_n^2}{n\sqrt V}\right).
\]
The residual remainder difference is smaller by Assumption~\ref{assump:remainder_structure}:
\[
\left\{\frac{1}{V}\sum_{v=1}^V
(\rho_n-\rho_{n,-v})^2
\right\}^{1/2}
=
o_p\!\left(\frac{\tilde\sigma_n^2}{n\sqrt V}\right).
\]
Therefore,
\[
\|d_n\|_{V,2}
\le
\left\{\frac{1}{V}\sum_{v=1}^V
|\mathcal B_n(A_{n,v},C_{n,v})|^2
\right\}^{1/2}
+
\left\{\frac{1}{V}\sum_{v=1}^V
(\rho_n-\rho_{n,-v})^2
\right\}^{1/2}
=
O_p\!\left(\frac{\tilde\sigma_n^2}{n\sqrt V}\right).
\]
We therefore conclude:

\begin{lemma}[Remainder-difference rate]
\label{lem:remainder_rate}
Under Assumption~\ref{assump:remainder_structure},
\[
\|d_n\|_{V,2}
= O_p\!\left(\frac{\tilde\sigma_{n}^{2}}{n\sqrt{V}}\right).
\]
\end{lemma}

\subsection*{G.4 Translating into a growth rate on $V$}

To satisfy the stability requirement of Theorem~\ref{thm:divergingV_slower}, we match the rate from Lemma~\ref{lem:remainder_rate} against the required $o_p\!\left(\frac{\sigma_n}{\sqrt{n}V}\right)$:
\[
\frac{\tilde\sigma_{n}^{2}}{n\sqrt{V}} = o\!\left(\frac{\sigma_n}{\sqrt{n}V}\right)
\;\Longleftrightarrow\;
\sqrt{V} = o\!\left(\frac{\sqrt{n}\,\sigma_n}{\tilde\sigma_{n}^{2}}\right)
\;\Longleftrightarrow\;
\boxed{\,V = o\!\left(n\frac{\sigma_n^2}{\tilde\sigma_{n}^{4}}\right).\,}
\]
This explicitly restricts how fast $V$ can grow to infinity, cleanly separating the variance of the target parameter's influence curve ($\sigma_n^2$) from the variance of the nuisance parameter's leading term ($\tilde\sigma_{n}^2$). For standard $\sqrt{n}$-rate RAL estimators (Theorem~\ref{thm:divergingV}), we have $\sigma_n = O(1)$, so the requirement reduces to $V = o(n / \tilde\sigma_n^4)$.

\subsection*{G.5 Examples}

Substituting specific rates yields several sufficient regimes depending on the complexity of the target and nuisance estimators:

\begin{itemize}
    \item \textbf{Standard RAL estimators with bounded-complexity nuisance linearization:} If $\sigma_n=O(1)$ and $\tilde\sigma_n=O(1)$, then the sufficient upper bound is $V=o(n)$. Thus standard RAL estimators can allow $V$ to grow nearly linearly, for example $V=\log n$ or $V=n^{\gamma}$ for any $\gamma<1$.
    \item \textbf{Standard RAL estimators with adaptive machine-learning nuisances:} If the target parameter is pathwise differentiable ($\sigma_n = O(1)$), but the nuisance functions have a diverging linearization scale $\tilde\sigma_n$, the sufficient upper bound tightens to $V = o(n / \tilde\sigma_n^4)$.
    \begin{itemize}
        \item If a $0$-order HAL-type nuisance linearization has $\tilde\sigma_n \asymp n^{1/6}$, then
        \[
        V = o\!\left(\frac{n}{(n^{1/6})^4}\right) \implies \boxed{\,V = o\!\bigl(n^{1/3}\bigr).\,}
        \]
        \item If a $1$st-order HAL-type nuisance linearization has $\tilde\sigma_n \asymp n^{1/10}$, then
        \[
        V = o\!\left(\frac{n}{(n^{1/10})^4}\right) \implies \boxed{\,V = o\!\bigl(n^{3/5}\bigr).\,}
        \]
    \end{itemize}
    
    \item \textbf{Generalized RAL estimators (e.g., HAL dose-response):} If the target parameter itself has a diverging influence-curve variance scale $\sigma_n \to \infty$, and if the nuisance and target linearization scales are comparable ($\tilde\sigma_n \asymp \sigma_n$), then the general sufficient bound $V = o\left(n \frac{\sigma_n^2}{\tilde\sigma_n^4}\right)$ simplifies to $V = o(n/\sigma_n^2)$. The following HAL exponents should be read as sufficient examples under this comparability assumption. They are illustrative oracle-rate consequences of the comparability assumption $\tilde\sigma_n \asymp \sigma_n$, not properties demonstrated in our simulations, which use first-order HAL only (Section~\ref{sec:sim_hal}).
    \begin{itemize}
        \item For first-order HAL (one degree of smoothness, $k=1$), if $\sigma_n \asymp \tilde\sigma_n \asymp n^{1/10}$, then
        \[
        V = o\!\left( \frac{n}{(n^{1/10})^2} \right) \implies \boxed{\,V = o\!\bigl(n^{4/5}\bigr).\,}
        \]
        \item For second-order HAL ($k=2$, so $k^{*}=k+1=3$ and the oracle dimension $J_n \asymp n^{1/(2k^{*}+1)} = n^{1/7}$ gives $\sigma_n \asymp \sqrt{J_n} \asymp n^{1/14}$), if $\sigma_n \asymp \tilde\sigma_n \asymp n^{1/14}$, then
        \[
        V = o\!\left( \frac{n}{(n^{1/14})^2} \right) \implies \boxed{\,V = o\!\bigl(n^{6/7}\bigr).\,}
        \]
	    \end{itemize}
	\end{itemize}
	For these polynomial-rate examples, a slow divergence such as $V \asymp \log n$ satisfies the required upper bounds.
\endgroup

\end{document}